\documentclass[a4paper,11pt]{article}
\usepackage[utf8]{inputenc}
\usepackage{amsmath, amssymb, amsthm, amscd}
\usepackage{bbm}
\usepackage{mathrsfs,mathtools}
\usepackage[text={5.8in,8.8in},centering]{geometry}
\usepackage{tikz}
\usepackage{caption}
\usepackage{graphicx}
\usepackage{float}
\usepackage{subcaption}

\usepackage{color}

\newtheorem{theorem}{Theorem}[section]
\newtheorem{proposition}[theorem]{Proposition}
\newtheorem{lemma}[theorem]{Lemma}

\newtheorem{corollary}[theorem]{Corollary}
\newtheorem{remark}[theorem]{Remark}

\def\bbbone{{\mathchoice {\rm 1\mskip-4mu l} {\rm 1\mskip-4mu l}
{\rm 1\mskip-4.5mu l} {\rm 1\mskip-5mu l}}}
\def\one{\bbbone}
\renewcommand{\i}{\mathrm{i}}
\newcommand{\e}{\mathrm{e}}
\newcommand{\real}{\mathrm{r}}
\newcommand{\I}[2]{\mathcal{I}_{#1, #2}}
\newcommand{\K}[2]{\mathcal{K}_{#1, #2}}
\newcommand{\J}[2]{\mathcal{J}_{#1, #2}}
\newcommand{\Hp}[2]{\mathcal{H}^+_{#1, #2}}
\newcommand{\Hm}[2]{\mathcal{H}^-_{#1, #2}}
\newcommand{\Hpm}[2]{\mathcal{H}^\pm_{#1, #2}}

\newcommand{\Wr}{\mathscr{W}}
\newcommand{\cH}{\mathcal{H}}
\newcommand{\cV}{\mathcal{V}}
\newcommand{\sF}{\mathscr{F}}
\newcommand{\cI}{\mathcal{I}}
\newcommand{\cK}{\mathcal{K}}
\newcommand{\cJ}{\mathcal{J}}
\newcommand{\cY}{\mathcal{Y}}
\newcommand{\cD}{\mathcal{D}}
\newcommand{\cG}{\mathcal{G}}
\newcommand{\pder}{\partial}

\newcommand{\suma}[2]{\sum\limits_{#1}^{#2}}

\newcommand{\naw}[1]{\left( {#1} \right)}
\newcommand{\abs}[1]{{\left| #1 \right|}}

\renewcommand{\Re}{\mathrm{Re}}
\renewcommand{\Im}{\mathrm{Im}}

\newcommand{\bbZ}{\mathbb{Z}}
\newcommand{\bbC}{\mathbb{C}}
\newcommand{\bbR}{\mathbb{R}}

\newcommand{\bbS}{\mathbb{S}}

\newcommand{\eps}{\epsilon}
\newcommand{\Dom}{{\mathcal D}}

\newcommand{\zesp}[1]{\overline{#1}}
\newcommand{\D}{\mathrm{D}}
\newcommand{\N}{\mathrm{N}}
\newcommand\slim{{\rm s-}\lim}

\def\Ka{\mathcal K}
\def\Ia{\mathcal I}
\def\Ha{\mathcal H}
\def\Ja{\mathcal J}
\def\Ya{\mathcal Y}

\def\R{\mathbb R}
\def\C{\mathbb C}
\def\N{\mathbb N}
\def\Z{\mathbb Z}

\def\B{\mathcal B}
\def\d{\mathrm d}

\begin{document}

\title{On radial Schr\"odinger operators with a Coulomb potential}

\author{
Jan Derezi\'{n}ski,\footnote{The financial support of the National Science
Center, Poland, under the grant UMO-2014/15/B/ST1/00126, is gratefully
acknowledged.}
\\
Department of Mathematical Methods in Physics, Faculty of Physics\\
University of Warsaw, Pasteura 5, 02-093, Warszawa, Poland\\
email: jan.derezinski@fuw.edu.pl,
\\  \\
Serge Richard\footnote{On leave of absence from
Univ Lyon, Universit\'e Claude Bernard Lyon 1, CNRS UMR 5208, Institut Camille Jordan, 43 blvd. du 11 novembre 1918, F-69622 Villeurbanne cedex, France.}
\footnote{
Partially supported by the grant
\emph{Topological invariants through scattering theory and noncommutative geometry} from Nagoya University.}\\
Graduate school of mathematics,
Nagoya University, \\
Chikusa-ku,
Nagoya 464-8602, Japan \\
email: richard@math.nagoya-u.ac.jp}
\maketitle

\begin{abstract}
This paper presents a thorough analysis of
1-dimensional Schr\"odinger operators whose potential
is a linear combination of the {\em Coulomb term} $1/r$ and the {\em centrifugal term} $1/r^2$. We allow both coupling constants to be complex.
Using natural boundary conditions at $0$, a two parameter holomorphic family of closed operators
on $L^2(\R_+)$ is introduced. We call them
the {\em Whittaker operators}, since in the mathematical literature their eigenvalue equation is called the {\em Whittaker equation}.
Spectral and scattering theory for Whittaker operators is studied.
Whittaker operators appear in quantum mechanics as the radial part
of the Schr\"odinger operator with a Coulomb potential.
\end{abstract}

\tableofcontents

\section{Introduction}
\setcounter{equation}{0}

Consider the differential expression
\begin{equation}\label{whita}
L_{\beta,\alpha}:=-\partial_z^2+\Big(\alpha-\frac14\Big)\frac{1}{z^2}-
\frac{\beta}{z},
\end{equation}
where the parameters $\beta,\alpha$ are arbitrary complex numbers.
This expression can be understood as an operator acting on functions
holomorphic outside of $0$, or acting on compactly supported smooth functions on $]0,\infty[$,
or acting on distributions on $]0,\infty[$.
We call \eqref{whita} the \emph{formal Whittaker operator}.

In this paper we are interested not so much in the formal operator $L_{\beta,\alpha}$ but in some of its realizations
as closed operators on $L^2(\R_+)$, with $\R_+:=]0,\infty[$.
To describe these closed operators it is natural to write $\alpha=m^2$.
Then, for any $m\in \C$ with $\Re(m)>-1$ we introduce an operator $H_{\beta,m}$ which is defined as the closed operator that equals
$L_{\beta,m^2}$ on the domain of functions that behave as $x^{\frac12+m}\big(1-\frac{\beta}{1+2m} x \big)$ near zero, see \eqref{eq_def_H} for a precise definition.
With this definition we obtain a two-parameter family of closed operators in $L^2(\R_+)$
\begin{equation*}
\C\times\{m\in \C\ |\ \Re(m)>-1\}\ni (\beta,m)\mapsto H_{\beta,m}\;\!,
\end{equation*}
which is holomorphic except for a singularity at $(\beta,m)=(0,-\frac12)$.
For $\Re(m) \geqslant 1$ the operator $H_{\beta,m}$ is simply the closure of $L_{\beta,m^2}$ restricted to $C_\mathrm{c}^\infty(\R_+)$. In fact, for
$\Re(m) \geqslant 1$ it is the unique closed realization of
$L_{\beta,m^2}$ on $L^2(\R_+)$.

This is not the case when $-1<\Re(m)<1$. Among various closed realizations of
$L_{\beta,m^2}$ one can distinguish the minimal one $L_{\beta,m^2}^{\min}$ and the maximal one $L_{\beta,m^2}^{\max}$. The operators $H_{\beta,m}$ lie between
$L_{\beta,m^2}^{\min}$ and $L_{\beta,m^2}^{\max}$. They are distinguished by the fact that they are obtained by analytic continuation from the region
$\Re(m) \geqslant 1$ where the uniqueness holds.
This continuation stops at the vertical line $\Re(m)=-1$, which cannot be passed
because on the left of this line the singularity $x^{\frac12+m}$ is not square integrable near $0$.

The operators $H_{\beta, m}$ are not the only closed realizations of
$L_{\beta,m^2}$ inside the strip $-1<\Re(m)<1$, but they are the {\em distinguished} ones. In fact, for generic $(m^2,\beta)$ in this strip there are two distinguished boundary conditions with the behavior
\begin{equation}\label{kappa1}
x^{\frac12+m}\big(1-\frac{\beta}{1+ 2m} x \big),\quad
x^{\frac12-m}\big(1-\frac{\beta}{1- 2m} x \big)
\end{equation}
near zero, and they correspond to the operators $H_{\beta,m}$ and $H_{\beta,-m}$.
In our paper we consider only the distinguished boundary conditions. We do not discuss other boundary conditions, except for the short remark below.

In the generic case, with $-1<\Re(m)<1$, there exist {\em mixed boundary conditions} corresponding to the behavior
\begin{equation}\label{kappa}
x^{\frac12+ m}\big(1-\frac{\beta}{1+ 2m} x \big)+\kappa
x^{\frac12-m}\big(1-\frac{\beta}{1-2m} x \big)
\end{equation}
near zero, where $\kappa$ is a complex parameter or $\kappa=\infty$ (with an appropriate interpretation of \eqref{kappa}). There are also two {\em degenerate cases}, for which the boundary conditions \eqref{kappa} do not work:
 If $m^2=0$, then both behaviors in \eqref{kappa1} coincide. If $m^2=\frac14$, $\beta\neq 0$, then only $m=\frac12$ makes sense in \eqref{kappa1}.
 In the degenerate cases one needs to modify \eqref{kappa}
by including appropriate logarithmic terms.

The goal of our paper is to study the properties of the family of operators $H_{\beta,m}$.
We do not restrict ourselves to real parameters, when $H_{\beta,m}$ are self-adjoint, but we consider general complex parameters.
In particular we would like to determine which properties survive in the non-self-adjoint setting and which do not.
Our paper is in many ways parallel to \cite{BDG} and especially to \cite{DR}, where the special case $\beta=0$ is studied.
These papers showed that the theory of Schr\"odinger operators with complex potentials can be very similar to the theory involving real potentials,
when we can have the self-adjointness. This includes functional calculus, spectral and scattering theory.

However, the present paper is not just a boring extension of \cite{BDG}---new interesting phenomena appear.
First of all, the operators $H_{\beta,m}$ usually have a sequence of eigenvalues accumulating at zero,
while for $\beta=0$ these eigenvalues are absent. Depending on the value of the parameters,
these eigenvalues disappear into the \emph{nonphysical sheet of the complex plane} and become resonances.
In the Appendix we give a few pictures of the spectrum of $H_{\beta,m}$, which illustrate the dependence
of eigenvalues and resonances on the parameters.

Another phenomenon, which we found quite unexpected, is the presence of
a non-removable singularity of the holomorphic function
$(\beta,m)\mapsto H_{\beta,m}$ at $(\beta,m)=(0,-\frac12)$.
This singularity
is closely related to the behavior of the potential at the origin. It
is quite curious: it is invisible when we consider just the variable $m$.
In fact, as proven already in \cite{BDG}, the map $m\mapsto H_m=H_{0,m}$ is holomorphic around $m=-\frac12$, and
$H_{0,-\frac12}$ is the Laplacian on the half-line with the Neumann boundary condition. It is also holomorphic around
$m=\frac12$, and
$H_{0,\frac12}$ is the Laplacian on the half-line with the Dirichlet boundary condition. Thus one has
\begin{align}\label{huli0}
H_{0,-\frac12}&\neq H_{0,\frac12}.
\end{align}
If we introduce the Coulomb potential, then whenever $\beta \neq 0$,
\begin{align}\label{huli}
H_{\beta,-\frac12}&=H_{\beta,\frac12}.
\end{align}
The function $(\beta,m)\mapsto H_{\beta,m}$ is holomorphic around $(0,\frac12)$, in particular,
\begin{equation}\label{huli1}
\lim_{\beta\to0}
(\one+H_{\beta,\frac12})^{-1}=(\one+H_{0,\frac12})^{-1}.
\end{equation}
But \eqref{huli} implies that
\begin{equation*}
\lim_{\beta\to0}
(\one+H_{\beta,-\frac12})^{-1}=(\one+H_{0,\frac12})^{-1}.
\end{equation*}
Thus
$\beta\mapsto(\one+H_{\beta,-\frac12})^{-1}$ is not even continuous near $\beta =0$. This singularity is closely related
to a rather irregular behavior of eigenvalues of $ H_{\beta,m}$, see Proposition \ref{more}.

As proven in \cite{BDG,DR}, the operators $H_{0,m}$ are rather well-behaved, also in the case of complex $m$.
The limiting absorption principle holds, namely the boundary values of the resolvent exist between the usual weighted spaces,
and scattering theory works the usual way. In particular, the {\em M{\o}ller operators} (also called {\em wave operators}) exist.
They are closely related to the {\em Hankel transformation}, which diagonalizes $H_{0,m}$, or equivalently which intertwines
them with a multiplication operator.

Most differences between $H_{0,m}$ and $H_{\beta,m}$ for $\beta\neq0$ are caused by the long-range character to the Coulomb potential.
In this context, it becomes critical whether $\beta$ is real or not.
As is well-known, for real $\beta$, we still have limiting absorption principle with the usual weighted spaces.
The usual M{\o}ller operators do not exist, but {\em modified M{\o}ller operators} do.
They can be expressed in terms of an isometric operator, which we call the {\em Hankel-Whittaker transformation}.

These properties mostly do not survive when $\beta$ becomes non-real.
In the limiting absorption principle we need to change the usual weighted spaces, see Theorem \ref{thm_conv_resolvent}.
The Hankel-Whittaker transform is no longer bounded, and to our understanding there is no sensible scattering theory.

Some remnants of scattering theory remain for complex $m$ but real non-zero $\beta$:
we show that in this case the {\em intrinsic scattering operator} is well defined, bounded and invertible unless $\Re(m)=-\frac12$.

It is usually stressed that constructions of long-range scattering theory are to some degree arbitrary \cite{DG0}.
More precisely, one says that modified M{\o}ller operators and the scattering operator have an arbitrary momentum dependent phase factor.
However, in the context of Whittaker operators there are distinguished choices for the M{\o}ller operators and for the scattering operator.
These choices appear more or less naturally when one wants to write down formulas for these operators in terms of special functions.
So one can argue that they were known before in the literature. However, to our knowledge this observation has not been formulated explicitly.

Let us sum up the properties of operators $H_{\beta,m}$ in various parameter regions.

\begin{enumerate}
\item If $\beta=0$ and $-1<m<\infty$, then $H_{\beta,m}$ is self-adjoint and the usual M{\o}ller operators exist.
\item If $\beta=0$ and $-1<\Re(m)<\infty$ with $\Im(m)\neq0$, then
$H_{\beta,m}$ is not self-adjoint, it is however similar to self-adjoint; the usual M{\o}ller operators exist \cite{BDG,DR}.
\item If $\beta\neq0$ with $\Im(\beta)=0$ and if $-1<m<\infty$, then
$H_{\beta,m}$ is self-adjoint, and the modified M{\o}ller operators exist.
\item If $\beta\neq0$ with $\Im(\beta)=0$ and if $-1<\Re(m)<\infty$ with $\Im(m)\neq0$, then
$H_{\beta,m}$ is not self-adjoint; maybe some kind of long-range scattering theory holds;
what we know for sure is the boundedness and invertibility of the intrinsic scattering operator unless $\Re(m)=-\frac{1}{2}$.
\item If $\Im(\beta)\neq0$ and $-1<\Re(m)<\infty$ with $\Im(m)\neq0$, then
$H_{\beta,m}$ is not self-adjoint; it seems that no reasonable scattering theory applies.
\end{enumerate}

The operator $H_{\beta,m}$ is one of the most important exactly solvable differential operators. Its eigenvalue equation for the eigenvalue (energy) $-\frac{1}{4}$
\begin{equation}\label{Whittaker-hyper.1}
\Big(-\pder_z^2 +\big(m^2 - \frac{1}{4}\big)\frac{1}{z^2} - \frac{\beta}{z}+\frac{1}{4}\Big)v = 0
\end{equation}
is known in mathematical literature as the {\em Whittaker equation}.
In fact, Whittaker published in 1903 a paper \cite{Whi} where he expressed solutions to \eqref{Whittaker-hyper.1} in terms of confluent functions.
This is the reason why we call $H_{\beta,m}$ the {\em Whittaker operator}.

The best known application of Whittaker operators concerns the Hydrogen Hamiltonian, that is, the Schr\"odinger
operator with a Coulomb potential in dimension 3. More generally, in any dimension the radial part of the Schr\"odinger operator
with Coulomb potential reduces to the Whittaker operator.
We sketch this reduction in Section \ref{The Coulomb problem in $d$ dimensions}.
A brief introduction to the subject can also be found in many textbooks on quantum mechanics, and we refer for example \cite[Sec.~135]{LL}
or \cite{GTV} for a recent approach.
The literature on the subject is vast and
we list only a few classical papers relevant for our manuscript, namely \cite{Dol,Ges,Gui,Her,Hum,MOC,Mic,MZ,Muk,Sea,TB,Yaf}
or more recently \cite[App. C]{KP}.
However, in all these references only real coupling constants are considered.
Note that the study of all possible self-adjoint extensions of the Whittaker operator in the real case
goes back to the work of Rellich \cite{R}, and was reconsidered with more generality by Bulla-Gesztesy \cite{BG}.
In particular, \cite{BG} discuss both mixed boundary conditions of the form \eqref{kappa} and their logarithmic modifications
needed in degenerate cases.

Let us finally describe the content of this paper.
Section \ref{sec_B_functions} is devoted to special functions that we need in our paper.
These functions are essentially eigenfunctions of the formal Whittaker operator \eqref{whita} corresponding to the eigenvalues $-\frac{1}{4}$,
$\frac{1}{4}$ and $0$. All of them can be expressed in terms of confluent and Bessel functions.
Note that we use slightly different conventions from those in most of the literature.
We follow our previous publication \cite{DR}, where we advocated the use of Bessel functions for dimension $1$, denoted
$\cI_{m}$, $\cK_{m}$, $\cJ_{m}$ and $\cH_{m}^\pm$.
Here we mimic this approach
and introduce systematically the functions $\cI_{\beta,m}$, $\cK_{\beta,m}$, $\cJ_{\beta,m}$
and $\cH_{\beta,m}^\pm$, which are particularly convenient in the context of the Whittaker operator.
Note that $\cI_{\beta,m}$, $\cK_{\beta,m}$ essentially coincide with the usual Whittaker functions, and $\cJ_{\beta,m}$
and $\cH_{\beta,m}^\pm$ are obtained by analytic continuation to imaginary arguments.
In particular, we present the asymptotic behavior of these functions near $0$ and near infinity for any parameters
$\beta$ and $m$ in $\C$.

Note that the theory of special functions related to the Whittaker equation is beautiful, rich and useful.
We try to present it in a concise and systematic way, which some readers should appreciate.
However, the readers who are more interested in operator-theoretic aspects of our paper can skip most of the material of Section
\ref{sec_B_functions} and go straight to the next section which constitutes the core of our paper.

In Section \ref{sec_Whi_op} we define the closed operators $H_{\beta,m}$ for any $m,\beta\in \C$
with $\Re(m)>-1$, and investigate their properties. A discussion about the complex eigenvalues of these operators is provided, as well as a description of a
limiting absorption principle on suitable spaces. At this point, the distinction between $\Im(\beta)=0$ and $\Im(\beta)\neq 0$
will appear. In the final part of the paper, we introduce Hankel-Whittaker transformations which diagonalize our operators,
and provide some information about the scattering theory. Some open questions are formulated
in the last subsection.

\subsection{The Coulomb problem in $d$ dimensions}\label{The Coulomb problem in $d$ dimensions}

Let us briefly describe the manifestation of the Whittaker operator in quantum mechanics.
We consider the space $L^2(\R^d)$ and the Schr\"odinger operator with the Coulomb potential in dimension $d$\;\!:
\begin{equation}\label{whit1}
-\Delta - \frac{\beta}{r},
\end{equation}
where $r$ denotes the radial coordinate.
In spherical coordinates the expression \eqref{whit1} reads
\begin{equation}\label{whit2}
-\pder_r^2 - \frac{d-1}{r}\pder_r
-\frac{1}{r^2}\Delta_{\bbS^{d-1}}-\frac{\beta}{r},
\end{equation}
where $\Delta_{\bbS^{d-1}}$ is the Laplace--Beltrami operator on the sphere $\bbS^{d-1}$.
Eigenvectors of $-\Delta_{\bbS^{d-1}}$ are the spherical harmonics and the corresponding eigenvalues are
$\ell(\ell+d-2)$, with $\ell=0,1,2,\dots$ for $d\geq2$; $\ell=0,1$ for $d=1$.
Thus on the spherical harmonics of order $\ell$ the expression \eqref{whit2} becomes
\begin{align*}
&-\pder_r^2 - \frac{d-1}{r}\pder_r +\frac{\ell(\ell+d-2)}{r^2}-\frac{\beta}{r}\\
&=-\pder_r^2 - \frac{d-1}{r}\pder_r +\frac{m^2 - \naw{\frac{d}{2}-1}^2}{r^2}-\frac{\beta}{r},
\end{align*}
where $m:= \ell+\frac{d}{2}-1$.
By letting $m$ take an arbitrary complex value and by considering $d=1$,
we obtain the \emph{Whittaker operator}
\begin{equation}\label{Whittaker-dim-one}
-\pder_r^2 +\Big(m^2 - \frac{1}{4}\Big)\frac{1}{r^2} - \frac{\beta}{r}.
\end{equation}

For $\beta=0$ the Whittaker operator simplifies to
the \emph{Bessel operator}, see for example \cite{BDG,DR}.
As for the Bessel operators, the Whittaker operators for distinct dimensions are related by a simple similarity transformation, namely
\begin{equation}\label{mimi6}
\begin{split}
&-\pder_r^2 - \frac{d-1}{r}\pder_r +\Big(m^2 - \Big(\frac{d}{2}-1\Big)^2\Big)\frac{1}{r^2} - \frac{\beta}{r}\\
&=r^{-\frac{d}{2}+\frac{1}{2}}\Big(-\pder_r^2 +\Big(m^2 - \frac{1}{4}\Big)\frac{1}{r^2} - \frac{\beta}{r}\Big)r^{\frac{d}{2}-\frac{1}{2}}.
\end{split}
\end{equation}
It is then a matter of taste to decide which dimension
should be treated as the standard one. From the physical point of view
$d=3$ is the most important, from the mathematical point
of view one can hesitate between $d=2$ and $d=1$. We choose $d=1$,
following the tradition going back to Whittaker \cite{Whi}, and consistently with \cite{DR}.

The Coulomb problem in the physical dimension $d=3$ has a considerable practical importance.
Therefore, there is a lot of literature devoted to the equation
\begin{equation*}
\Big(\pder_r^2 -\ell(\ell+1)\frac{1}{r^2} - \frac{2\eta}{r}+1\Big)v = 0,
\end{equation*}
called \emph{the Coulomb wave equation}, see \cite[Chap.~14]{AS}, which is directly obtained from the physical problem. For this equation, $\ell$
is a non-negative integer and $\eta$ is a real parameter.
Solutions of this equation are often denoted by
$F_\ell(\eta,r)$, $G_\ell(\eta,r)$ and $H_\ell^\pm(\eta,z):=G_\ell(\eta,r)\pm \i F_\ell(\eta,r)$,
and are called \emph{Coulomb wave functions}. Alternatively, the equation
\begin{equation*}
\Big(\pder_r^2 -\ell(\ell+1)\frac{1}{r^2} + \frac{2}{r}+\varepsilon\Big)v = 0
\end{equation*}
has been considered for $\varepsilon \in \R$, and its solution
are often denoted by $f(\varepsilon, \ell;r)$, $h(\varepsilon,\ell;r)$, and also $s(\varepsilon,\ell;r)$
and $c(\varepsilon,\ell;r)$. Properties of these functions have been studied for example in \cite{Hum,Sea,TB} and compiled in \cite{NIST}
(see also the more recent work \cite{Gas}).

Our aim is to consider the Whittaker operator in its mathematically most natural form,
including complex values of parameters, which do not have an obvious physical meaning. This explains some differences
of our set-up and conventions compared with those used in the above literature.

\subsection{Notation}

We shall use the notations $\R_+$ for $]0,\infty[$, $\N$ for $\{0,1,2,\dots\}$,
while $\N^\times:=\{1,2,3,\dots\}$.
For $\alpha\in \C$, $\bar \alpha$ means the complex conjugate.
 $C_{\rm c}^\infty(\R_+)$ denotes
the set of smooth functions on $\R_+$ with compact support.

For an operator $A$ we denote by $\cD(A)$ its domain and by $\sigma_{\rm p}(A)$
the set of its eigenvalues (its point spectrum). We also use the notation
$\sigma(A)$ for its spectrum, $\sigma_{\rm ess}(A)$ for its essential spectrum
and $\sigma_\d(A)$ for its discrete spectrum.
If $z$ is an isolated point of $\sigma(A)$, then $\one_{\{z\}}(A)$
denotes the Riesz projection of $A$ onto $z$. Similarly, if $A$ is
self-adjoint and $\Xi$ is a Borel subset of $\sigma(A)$, then
$\one_\Xi(A)$ denotes the spectral projection of $A$ onto $\Xi$.

The following holomorphic functions are understood as their \emph{principal
bran\-ches}, that is, their domain is $\C\setminus]-\infty,0]$ and on
$]0,\infty[$ they coincide with their usual definitions from real analysis:
$\ln(z)$, $\sqrt z$, $z^\lambda$. We set $\arg (z):=\Im \big(\ln(z)\big)$.
The extensions of these functions to $]-\infty,0]$ or to $]-\infty,0[$ are from the upper half-plane.

\medskip
{\bf Acknowledgement.} The authors thank M.~Karczmarczyk for his contributions at an early stage of this project.
They are also grateful to D.~Siemssen, who helped them to make pictures with Mathematica.

\section{Bessel and Whittaker functions}\label{sec_B_functions}
\setcounter{equation}{0}

An important role in our paper is played by various kinds of {\em Whittaker functions}, closely related to {\em confluent hypergeometric functions}.
We will also use several varieties of {\em Bessel functions}.
In this section we fix the notation concerning these special
functions and describe their basic properties.

This section plays an auxiliary role in our paper, since almost all its result can be found in the literature.
The readers interested mainly in our operator-theoretic results can only briefly skim this section, and then go
to the next one, which constitutes the main part of our paper.

We start by recalling the definition of Bessel functions for dimension $1$,
which we prefer to use instead of the usual Bessel functions.
Their main properties have been discussed in \cite{DR}, therefore there is no need to repeat them here.
We then introduce the Whittaker functions $\cI_{\beta,m}$, $\cK_{\beta,m}$, $\cJ_{\beta,m}$
and $\cH_{\beta,m}^\pm$.
These functions are solutions of the hyperbolic-type and
trigonometric-type Whittaker equation, as explained below.
In our notation and presentation, as much as possible, we stress the analogy of Whittaker functions and Bessel functions.
The section ends with a description of zero-energy solutions of the Whittaker operator.

\subsection{Hyperbolic and trigonometric Whittaker equation}

A simple argument using complex scaling shows that the eigenvalue problem with non-zero energies
for the Whittaker operator \eqref{Whittaker-dim-one} can be derived from the
following equation, which is known in the literature as
the \emph{Whittaker equation}
\begin{equation}\label{Whittaker-hyper}
\Big(-\pder_z^2 +\big(m^2 - \frac{1}{4}\big)\frac{1}{z^2} - \frac{\beta}{z}+\frac{1}{4}\Big)v = 0.
\end{equation}
It is convenient to consider in parallel to
\eqref{Whittaker-hyper} the additional equation
\begin{equation}\label{Whittaker-trig}
\Big(-\pder_z^2 +\big(m^2 - \frac{1}{4}\big)\frac{1}{z^2} - \frac{\beta}{z}-\frac{1}{4}\Big)v = 0.
\end{equation}
We call it \emph{the trigonometric-type Whittaker equation}.
For consistency, the equation \eqref{Whittaker-hyper}
is sometimes referred to as \emph{the hyperbolic-type Whittaker equation}. Note that one can pass from \eqref{Whittaker-hyper} to
\eqref{Whittaker-trig} by replacing $z$ with $\pm\i z$ and $\beta$ with $\mp\i\beta$.

\subsection{Bessel equations and functions}\label{sec_B}

In the special case $\beta=0$, by rescaling the independent variable in \eqref{Whittaker-hyper} and
 \eqref{Whittaker-trig}, we obtain the {\em modified (or hyperbolic-type) Bessel equation for dimension $1$}
\begin{equation}\label{lap7}
\Big(\partial_z^2-\big(m^2-\frac14\big)\frac{1}{z^2}-1\Big)v = 0,
\end{equation}
and the {\em standard (or trigonometric-type) Bessel equation
for dimension $1$}
\begin{equation}\label{lap8}
\Big(\partial_z^2-\big(m^2-\frac14\big)\frac{1}{z^2}+1\Big)v = 0.
\end{equation}
As explained in \eqref{mimi6} they are equivalent to
the modified (or hyperbolic-type) Bessel equation
\begin{equation}\label{lap5}
\Big(\partial_z^2+\frac{1}{z}\partial_z-\frac{m^2}{z^2}-1\Big)v = 0,
\end{equation}
respectively to the standard (or trigonometric-type) Bessel equation
\begin{equation}\label{lap6}
\Big(\partial_z^2+\frac{1}{z}\partial_z-\frac{m^2}{z^2}+1\Big)v = 0,
\end{equation}
which are usually considered in the literature.

The distinguished solutions of \eqref{lap5} are
\begin{align*}
\rm{the}\ \emph{modified Bessel function}&&I_m(z),\\
\rm{the}\ \emph{MacDonald function}&&K_m(z),
\end{align*}
and of \eqref{lap6} are
\begin{align*}\rm{the}\
\emph{Bessel function}&&J_m(z),\\
\rm{the}\ \emph{Hankel function of the 1st kind}&&H_m^+(z)=H_n^{(1)}(z),\\
\rm{the}\ \emph{Hankel function of the 2nd kind}&&H_m^-(z)=H_n^{(2)}(z),\\
\rm{the}\ \emph{Neumann function}&&Y_m(z).
\end{align*}
Following \cite{DR} we prefer
functions which solve the two Bessel equations for dimension~1.
Namely, we shall use the following functions solving \eqref{lap7}
\begin{align*}
\rm{the}\ \emph{modified (or hyperbolic) Bessel function for dim.~$1$}\qquad&\Ia_m(z):=\sqrt{\frac{\pi z}{2}} I_m(z),\\
\rm{the}\ \emph{MacDonald function for dim.~$1$}\qquad&\Ka_m(z):=\sqrt{\frac{2 z}{\pi}} K_m(z).
\end{align*}
We will also use the following functions solving \eqref{lap8}
\begin{align*}
\rm{the}\ \emph{(trigonometric) Bessel function for dim.~$1$}\qquad&\Ja_m(z):= \sqrt{\frac{\pi z}{2}} J_m(z),\\
\rm{the}\ \emph{Hankel function of the 1st kind for dim.~$1$}\qquad&\Ha_m^+(z):=\sqrt{\frac{\pi z}{2}} H^+_m(z),\\
\rm{the}\ \emph{Hankel function of the 2nd kind for dim.~$1$}\qquad&\Ha_m^-(z):=\sqrt{\frac{\pi z}{2}} H^-_m(z),\\
\rm{the}\ \emph{Neumann function for dim.~$1$}\qquad&\Ya_m(z):= \sqrt{\frac{\pi z}{2}} Y_m(z).
\end{align*}
We refer the reader to the Appendix of \cite{DR} for the properties of these functions.

\subsection{The function $\cI_{\beta,m}$}\label{sec_f1}

The hyperbolic-type Whittaker equation \eqref{Whittaker-hyper}
can be reduced to the ${}_1F_1$-equation, also known as the {\em confluent equation}:
\begin{equation}\label{confluent-equation}
\big(z\pder_z^2 + (c-z)\pder_z - a\big)v=0.
\end{equation}
Indeed, one has
\begin{equation*}
-z^{\frac{1}{2}\mp m}\e^{\frac{z}{2}}\Big(-\pder_z^2 + \big(m^2 - \frac{1}{4}\big)\frac{1}{z^2} - \frac{\beta}{z} + \frac{1}{4}\Big)z^{\frac{1}{2}\pm m}\e^{-\frac{z}{2}}
= z\pder_z^2 + (c-z)\pder_z - a
\end{equation*}
for the parameters $c = 1 \pm 2m$ and $a = \frac{1}{2} \pm m-\beta$.
Here the sign $\pm$ has to be understood as two possible choices.

One of the solutions of the confluent equation is
Kummer's confluent hypergeometric function ${}_1F_1(a;c;\cdot)$ defined by
\begin{equation}\label{kummer-function}
{}_1F_1(a;c;z) = \suma{k=0}{\infty}\frac{(a)_k}{(c)_k}\frac{z^k}{k!},
\end{equation}
where $(a)_k:=a(a+1)\cdots(a+k-1)$ is the usual Pochhammer's symbol.
It is the only solution of \eqref{confluent-equation} behaving as $1$ in the vicinity of $z=0$.

It is often convenient to use the closely-related function ${}_1\mathbf{F}_1(a;c;\cdot)$
defined by
\begin{equation}\label{whi}
{}_1\mathbf{F}_1(a;c;z) = \suma{k=0}{\infty}\frac{(a)_k}{\Gamma(c+k)}\frac{z^k}{k!} = \frac{{}_1F_1(a;c;z)}{\Gamma(c)}.
\end{equation}
We prefer the normalization \eqref{whi}, and in the sequel the following
function $\cI_{\beta,m}$ will be treated as one of the standard solutions
of the hyperbolic-type Whittaker equation \eqref{Whittaker-hyper}:
\begin{align}\label{eq_serie_I}
\nonumber \cI_{\beta,m}(z) & := z^{\frac{1}{2}+m}\e^{\mp \frac{z}{2}} {}_1\mathbf{F}_1\Big(\frac{1}{2}+m\mp\beta;\,1+2m;\,\pm z\Big) \\
&= z^{\frac{1}{2}+m}\e^{\mp\frac{z}{2}}\suma{k=0}{\infty}\frac{\naw{\frac{1}{2}+m\mp\beta }_k}{\Gamma(1+2m+k)}\;\!\frac{(\pm z)^k}{k!}.
\end{align}
Note that the sign independence comes from the $1^{\mathrm{st}}$
Kummer's identity
\begin{equation*}
{}_1F_1(a;\,c;\,z) = \e^z {}_1F_1(c-a;\,c;\,-z).
\end{equation*}

In the special case $\beta=0$, the function $\cI_{0,m}$ essentially coincides with the
modified Bessel function. More precisely, one has
\begin{equation}\label{eqI_{0,m}}
\cI_{0,m}(z) = \frac{\sqrt{\pi z}}{\Gamma\big(\frac{1}{2}+m\big)}I_m\Big(\frac{z}{2}\Big)
= \frac{2}{\Gamma\big(\frac{1}{2}+m\big)}\cI_m\Big(\frac{z}{2}\Big).
\end{equation}

For $- \frac{1}{2}-m\pm\beta: = n\in\N$, the series of \eqref{kummer-function} is finite and we get
\begin{equation*}
\cI_{\pm (\frac{1}{2}+m+n),m}(z) = \frac{n!\;\!z^{\frac{1}{2}+m}\e^{\mp\frac{z}{2}}}{\Gamma(1+2m+n)}L_n^{(2m)}(\pm z),
\end{equation*}
where
\begin{equation}\label{eq_Laguerre}
L_n^{(2m)}(z) = \frac{z^{-2m}\e^z}{n!} \frac{\mathrm{d}^n}{\mathrm{d}z^n}\big(\e^{-z}z^{2m+n}\big)
\end{equation}
are the {\em Laguerre polynomials} (or {\em generalized Laguerre polynomials}).

Finally, from equation \eqref{whi} one can deduce the asymptotic behaviour around $0$\;\!:
\begin{equation}\label{Ibm-around-zero}
\cI_{\beta,m}(z)=\frac{z^{\frac{1}{2}+m}}{\Gamma(1+2m)} \Big(1 -\frac{\beta}{1+2m}z+ O(z^2)\Big),
\end{equation}
while from the asymptotic properties of the ${}_1F_1$-function one obtains for $|\arg(z)|<\frac{\pi}{2}$
and large $|z|$
\begin{equation}\label{Ibm-around-infinity}
\cI_{\beta,m}(z) = \frac{1}{\Gamma(\frac{1}{2}+m-\beta)}z^{-\beta}\;\!\e^{\frac{z}{2}}\big(1+O(z^{-1})\big).
\end{equation}

\subsection{The function $\cK_{\beta,m}$}\label{sec_f2}

The hyperbolic-type Whittaker equation \eqref{Whittaker-hyper} has
also a solution with a simple behavior at $\infty$. However, its analysis is somewhat more difficult than that of solutions with a simple behavior at $z=0$,
because $z=\infty$ is an irregular singular point.
The most convenient way to look for solutions
with a simple behavior at $\infty$ is to
reduce the Whittaker equation to the {\em ${}_2F_0$ equation}
\begin{equation*}
\big(w^2\partial_w^2+(-1+(1+a+b)w)\partial_w + ab\big)v=0.
\end{equation*}
Indeed by setting $w=-z^{-1}$ we obtain
\begin{align*}
&-z^{2-\beta}
\e^{\frac{z}{2}}\Big(-\pder_z^2 + \big(m^2 - \frac{1}{4}\big)\frac{1}{z^2} - \frac{\beta}{z} + \frac{1}{4}\Big)z^{\beta}\e^{-\frac{z}{2}}\\
&= w^2\partial_w^2+(-1+(1+a+b)w)\partial_w + ab
\end{align*}
for the parameters $a = \frac{1}{2} +m -\beta$ and $b = \frac{1}{2} -m -\beta$.

The ${}_2F_0$ equation has a distinguished solution
\begin{equation*}
{}_2F_0(a,b;-;z):=\lim_{c\to\infty}{}_2F_1(a,b;c;cz),
\end{equation*}
where we take the limit over
$|\arg(c)-\pi|<\pi-\epsilon$ with $\epsilon>0$,
and the above definition is valid
for $z\in\bbC\setminus[0,+\infty[$.
Obviously one has
\begin{equation}\label{obvio}
{}_2 F_0(a,b;-;z)={}_2F_0(b,a;-;z).
\end{equation}
The function extends to an analytic function on the universal cover of
$\C\setminus\{0\}$ with a branch point of an infinite order at 0,
and the following asymptotic expansion holds:
\begin{equation*}
{}_2F_0(a,b;-;z)\sim\sum_{n=0}^\infty\frac{(a)_n(b)_n}{n!}z^n,\quad |\arg(z)|<\pi-\epsilon.
\end{equation*}
In the literature the ${}_2F_0$ function is seldom used.
Instead one uses {\em Tricomi's function}
\begin{align*}
U(a,c,z) & :=z^{-a}{}_2F_0(a;a- c +1;-; -z^{-1}) \\
& =\frac{\Gamma(1-c)}{\Gamma(1+a-c)}{}_1F_1(a;c;z) + \frac{\Gamma(c-1)}{\Gamma(a)}z^{1-c}{}_1F_1(1+a-c;2-c;z).
\end{align*}
Tricomi's function is one of solutions of the confluent equation \eqref{confluent-equation}.

We then define
\begin{align*}
\cK_{\beta,m}(z) &:=z^\beta\e^{-\frac{z}{2}} {}_2F_0\Big(\frac12+m-\beta,\frac12-m-\beta;-;-z^{-1}\Big)\\
&= z^{\frac{1}{2}+m}\e^{-\frac{z}{2}} U\Big(\frac{1}{2} + m -\beta;\,1 + 2m;\,z\Big),
\end{align*}
which is thus a solution of the hyperbolic-type Whittaker equation \eqref{Whittaker-hyper}.
The symmetry relation \eqref{obvio} implies that
\begin{equation}\label{eq_sym}
\cK_{\beta,m}(z) = \K{\beta}{-m}(z).
\end{equation}
The following connection formulas hold for $2m\notin\bbZ$\;\!:
\begin{align}\label{connec}
\cK_{\beta,m}(z)& = -\frac{\pi}{\sin(2\pi m)}\Big(\frac{\cI_{\beta,m}(z)}{\Gamma(\frac{1}{2}-m-\beta)}
- \frac{\I{\beta}{-m}(z)}{\Gamma(\frac{1}{2}+m-\beta)}\Big),\\
\nonumber \cI_{\beta,m}(z)& = \frac{\Gamma(\frac{1}{2}-m+\beta)}{2\pi}\Big(\e^{\i\pi m}\K{-\beta}{m}(\e^{\i\pi}z) + \e^{-\i\pi m}\K{-\beta}{m}(\e^{-\i\pi}z)\Big).
\end{align}

Recall that the Wronskian of two functions $f,g$ is defined as
\begin{equation}\label{wron}
\Wr( f,g;x):=f(x)g'(x)-f'(x)g(x).
\end{equation}
The Wronskian of $\cI_{\beta,m}$ and $\cK_{\beta,m}$ can be easily computed, and one finds
\begin{equation}\label{wron1}
\Wr(\cI_{\beta,m},\cK_{\beta,m};x) = - \frac{1}{\Gamma(\frac{1}{2} + m - \beta)}.
\end{equation}

For the special cases, the relation of the function $\K{0}{m}$ with
the usual Macdonald function $K_m$ or with the
Macdonald function $\cK_m$ for dimension $1$ reads
\begin{equation}\label{eqK_{0,m}}
\K{0}{m}(z) = \sqrt{\frac{z}{\pi}}K_m\Big(\frac{z}{2}\Big)=\cK_m\Big(\frac{z}{2}\Big).
\end{equation}
Also for $\beta - \frac{1}{2}-m=:n\in \N$ we obtain
\begin{equation*}
\K{\frac{1}{2} + m+n}{\pm m}(z) = (-1)^n n!\, z^{\frac{1}{2}+m}\e^{-\frac{z}{2}}L_n^{(2m)}(z),
\end{equation*}
where $L_n^{(2m)}$ are the Laguerre polynomials introduced in \eqref{eq_Laguerre}.
Note that for these values of $\beta$ the functions $\cI_{\beta,m}$ and $\cK_{\beta,m}$ are essentially the same, except for a $z$-independent factor.
However, for $\beta = -\big(\frac{1}{2}+m+n\big)$ the function $\cK_{\beta,m}$ has a more complicated representation, see \cite{confluent}.

Finally, for $2m\not \in \Z$ the behaviour of $\cK_{\beta,m}$ around zero can be derived
from that of $\cI_{\beta,m}$ together with the relation \eqref{connec}, while for
$2m\in \Z$ the l'H\^ opital's rule has to be used, see the next subsection. For simplicity, we provide the asymptotic behavior only for $\Re(m)\geqslant 0$,
since similar results for $\Re(m)\leqslant 0$ can be obtained by taking \eqref{eq_sym} into account. Thus one has:
\begin{equation}\label{Kbm-around-zero}
\cK_{\beta,m}(z) =
\begin{cases}
-\frac{z^{\frac{1}{2}}\ln(z)}{\Gamma(1-\beta)} + O\big(\abs{z}^{\frac{1}{2}}\big) & \mbox{ for } m=0,\\
z^\frac{1}{2}\Big(\frac{\Gamma(-2m)}{\Gamma(\frac{1}{2}-m-\beta)}z^m + \frac{\Gamma(2m)}{\Gamma(\frac{1}{2}+m-\beta)}z^{-m}\Big) + O\big(\abs{z}^\frac{3}{2}\big)
& \mbox{ for } \Re(m)=0,\ m\ne 0, \\
\frac{\Gamma(2m)}{\Gamma(\frac{1}{2}+m-\beta)}z^{\frac{1}{2}-m}+ O\big(\abs{z}^{\frac{1}{2}+\Re(m)}\big) & \mbox{ for } \Re(m)\in ]0,\frac{1}{2}], \ m\neq \frac{1}{2}, \\
\frac{1}{\Gamma(1-\beta)} + O\big(z\ln(z)\big) & \mbox{ for } m=\frac{1}{2}, \\
\frac{\Gamma(2m)}{\Gamma(\frac{1}{2}+m-\beta)}z^{\frac{1}{2}-m}+ O\big(\abs{z}^{\frac{3}{2}-\Re(m)}\big) & \mbox{ for } \Re(m)>\frac{1}{2}.
\end{cases}
\end{equation}
For the behaviour for large $z$, if $\epsilon>0$ and $\abs{\arg(z)}<\pi-\epsilon$ then one has
\begin{equation}\label{Kbm-around-infinity}
\cK_{\beta,m}(z) = z^\beta\;\!\e^{-\frac{z}{2}} \big(1+O(z^{-1})\big).
\end{equation}

\begin{remark}	
In the literature one can find various conventions for solutions of the
Whittaker equation. In part of the literature
 $M_{\beta,m}:=\Gamma(1 + 2m)\;\!\cI_{\beta,m}$ is called
the Whittaker function of the first kind.
$\cK_{\beta,m}$ is called the Whittaker function of
the second kind and denoted by $W_{\beta,m}$.
In \cite{confluent}, our functions $\cI_{\beta,m}$ and $\cK_{\beta,m}$
correspond to the functions ${\mathscr M}_{\varkappa,\mu/2}$ and $W_{\varkappa,\mu/2}$
with $\varkappa=\beta$ and $\mu/2=m$.

With our notation we try to be parallel to the notation for
the modified Bessel equation.
In fact, for $\beta=0$ our functions $\cI_{\beta,m}$ and
$\cK_{\beta,m}$ are closely related to the modified
Bessel function $\cI_m$ and the Macdonald function $\cK_m$,
and this will also hold for $\J{\beta}{m}$ and
$\Hpm{\beta}{m}$ with the Bessel function $\cJ_m$ and the Hankel functions $\cH_m^\pm$.
\end{remark}

\subsection{Degenerate case}

In this section we consider the hyperbolic Whittaker equation in the special case $m=\pm \frac{p}{2}$ for any $p\in\N$.
It is sometimes called the degenerate case, because the two solutions $\cI_{\beta,m}$ and $\cI_{\beta,-m}$
in this case are proportional to one another and do not span the solution space.
Therefore, we are forced to use the function
$\cK_{\beta,m}$
to obtain all solutions.

Let us fix $p\in\N$. We have the identity
\begin{equation}\label{dege}
\cI_{\beta,-\frac{p}{2}}(z) = \Big(-\beta - \frac{p-1}{2}\Big)_{p}\, \cI_{\beta,\frac{p}{2}}(z),
\end{equation}
or equivalently,
\begin{equation*}
\frac{\I{\beta}{-\frac{p}{2}}(z)} {\Gamma\big(\frac{1+p}{2}-\beta\big)}
= \frac{\cI_{\beta,\frac{p}{2}}(z)} {\Gamma\big(\frac{1-p}{2}-\beta\big)}.
\end{equation*}
Indeed, the confluent function $_1F_1(a;\,c;\,z)$ is divergent for $c\to-p$,
however, the divergence is of the same order as the
divergence of $\Gamma(c)$ for $c\to-p$.
Then, by a straightforward calculation we obtain from \eqref{whi} the equality
\begin{equation*}
{}_1\mathbf{F}_1(a;\,-p+1;\,z) = (a)_{p}\;\!z^{p}\;\!{}_1\mathbf{F}_1(a+p;\,1+p;\,z),
\end{equation*}
which implies \eqref{dege}.
Note that \eqref{dege} also implies that
\begin{equation}\label{dege1}
\cI_{\beta,-\frac{p}{2}}(z) = 0,\ \ \ \hbox{ for } \beta\in\Big\{\frac{1-p}{2},\frac{3-p}{2},\dots,\frac{p-1}{2}\Big\}.
\end{equation}

Let us now compare the symmetry \eqref{dege} with
similar properties of the modified Bessel functions for dimension $1$. For such functions we have
$\cI_{-m}(z)=\cI_m(z)$, for any $m\in\bbZ$,
which is consistent with \eqref{dege}. But for
$m\in\bbZ+\frac12$, $\cI_{-m}$ is not proportional to
$\cI_m$, which at the first sight contradicts
\eqref{dege}. However, $\cI_{0,m}(z) = \frac{2}{\Gamma\naw{\frac{1}{2}+m}}\cI_m\naw{\frac{z}{2}}$ vanishes for
$m\in\{\dots,-\frac32,-\frac12\}$, and this makes it consistent with \eqref{dege1}.

The function $\cK_{\beta,m}$ is quite complicated in the degenerate case.
In order to describe it, let us introduce the \emph{digamma function}
\begin{equation*}
\psi(z) = \partial_z\ln\big(\Gamma(z)\big)=\frac{\Gamma'(z)}{\Gamma(z)}.
\end{equation*}
Let us also set for $k\in\N$
\begin{align*}
H_k(z)&:= \frac1z+\frac1{z+1}+\dots+\frac{1}{z+k-1},\\
H_k&:= H_k(1)=1+\frac12+\cdots\frac1k.
\end{align*}
Obviously, this means that $H_0(z)=H_0=0$.
We have $\psi(1)=-\gamma$, $\psi(\frac12)=-\gamma-2\ln(2)$. Besides,
one has $\psi(z+k)= \psi(z)+H_k(z)$ and $\psi(1+k)= -\gamma+H_k$,
and for $k\in \N$
\begin{equation*}
\partial_z\frac{1}{\Gamma(z)}\Big|_{z=-k} = -\frac{\psi(z)}{\Gamma(z)}\Big|_{z=-k} = (-1)^k k!\,.
\end{equation*}

The following statement can be proven by l'H\^ opital's rule.

\begin{theorem}
For $p\in \N$, we have
\begin{align*}
\cK_{\beta,\frac{p}{2}}(z)
& = \frac{(-1)^{p+1}\ln(z)\cI_{\beta,\frac{p}{2}}(z)}
{\Gamma\big(\tfrac{1-p}{2}-\beta\big)}\nonumber\\
&+\frac{(-1)^{p+1}\e^{-\frac{z}{2}}}{\Gamma\big(\frac{1-p}{2}-\beta\big)} \sum_{k=0}^\infty
\frac{\big(\frac{1+p}{2}-\beta\big)_k\;\!z^{\frac{1+p}{2}+k}}{(p+k)!\;\!k!}\\
&\quad \times \Big(\psi\big(\tfrac{1+p}{2}-\beta+k\big)-\psi(p+1+k)
-\psi(1+k)\Big)\\
&+\e^{-\frac{z}{2}}\sum_{j=0}^{p-1}
\frac{(\frac{1-p}{2}-\beta\big)_j
(-1)^j(p-j-1)!}{j!\;\!\Gamma\big(\frac{1+p}{2}-\beta\big)}
z^{\frac{1-p}{2}+j}.
\end{align*}
For $p=0$ the above formula simplifies:
\begin{align*}
\cK_{\beta,0}(z) & = -\frac{\ln(z) \cI_{\beta,0}(z)}{\Gamma\big(\tfrac{1}{2}-\beta\big)} \\
& \quad - \frac{\e^{-\frac{z}{2}}}{\Gamma\naw{\frac{1}{2}-\beta}} \sum_{k=0}^\infty
\frac{(\frac12-\beta)_k\;\!z^{\frac12+k}}{(k!)^2}
\big(\psi(\tfrac12-\beta+k)-2\psi(1+k)\big).
\end{align*}
\end{theorem}

\subsection{The function $\cJ_{\beta,m}$}

In this and the next subsections, we consider the trigonometric-type Whittaker equation and its solutions.

The function $\J{\beta}{m}$ is defined by the formula
\begin{equation}\label{Jbm-definition}
\J{\beta}{m}(z) = \e^{\mp\i\frac{\pi}{2}(\frac{1}{2}+m)}\I{\mp\i\beta}{m}\big(\e^{\pm\i\frac{\pi}{2}}z\big).
\end{equation}
It is a solution of the trigonometric-type Whittaker equation \eqref{Whittaker-trig}
 which behaves as $\frac{z^{\frac12+m}}{\Gamma(1+2m)}$ for $z$ near $0$.
More precisely, one infers from \eqref{Ibm-around-zero} that for $z$ near $0$
\begin{equation}\label{Jbm-around-zero}
\cJ_{\beta,m}(z)=\frac{z^{\frac{1}{2}+m}}{\Gamma(1+2m)} \Big(1 -\frac{\beta}{1+2m}z+ O(z^2)\Big).
\end{equation}
It satisfies\begin{equation*}
\J{\zesp{\beta}}{\zesp{m}}(\zesp{z}) = \zesp{\J{\beta}{m}(z)}.
\end{equation*}
By starting again from the asymptotics of the ${}_1F_1$-function provided for example in \cite[Eq.~13.5.1]{AS} one can also obtain the asymptotic expansion near infinity.
However, note that we consider a real variable $x$ and only $x\to \infty$ since for a complex variable $z$ the asymptotic behaviour
highly depends on the argument of $z$. One thus gets for $x\in \R_+$ with $x$ large:
\begin{align}\label{Jbm-around-infinity}
\nonumber & \cJ_{\beta,m}(x) = \frac{\e^{\i\frac{\pi}{2}(\frac{1}{2}+m+\i\beta)}}{\Gamma(\frac{1}{2}+m-\i \beta)} \e^{-\i\frac{x}{2}} x^{-\i \beta}
\Big(1+\i\Big(\frac{1}{2}+m+\i\beta\Big)\Big(\frac{1}{2}-m+\i \beta\Big)x^{-1}+O\big(x^{-2}\big)\Big) \\
&\qquad + \frac{\e^{-\i\frac{\pi}{2}(\frac{1}{2}+m-\i\beta)}}{\Gamma(\frac{1}{2}+m+\i \beta)}\e^{\i\frac{x}{2}}x^{\i \beta}
\Big(1-\i\Big(\frac{1}{2}+m-\i\beta \Big)\Big(\frac{1}{2}-m-\i \beta\Big)x^{-1}+O\big(x^{-2}\big)\Big).
\end{align}

In the special case $\beta=0$ one has
\begin{equation}\label{eq_J0m_B}
\J{0}{m}(z) = \frac{\sqrt{\pi z}}{\Gamma\naw{\frac{1}{2}+m}}J_m\naw{\frac{z}{2}} = \frac{2}{\Gamma\naw{\frac{1}{2}+m}}\cJ_m\naw{\frac{z}{2}}.
\end{equation}

\subsection{The functions $\Hpm{\beta}{m}$}\label{sec_f4}

Let us define the functions $\Hpm{\beta}{m}$ by the formula
\begin{equation}\label{Hpm-definition}
\Hpm{\beta}{m}(z) = \e^{\mp \i\frac{\pi}{2}\naw{\frac{1}{2}+m}}\K{\pm\i\beta}{m}(\e^{\mp\i\frac{\pi}{2}}z).
\end{equation}
Note that here the sign $\pm$ means that we have two functions:
one for the sign $+$ and one for the sign $-$.
The functions $\Hpm{\beta}{m}$ are solutions of the trigonometric-type Whittaker equation \eqref{Whittaker-trig}.

One can observe that the property
$\Hpm{\beta}{-m}(z) = \e^{\pm\i\pi m}\Hpm{\beta}{m}(z)$ holds.
For these functions one has the following connection formulas:
\begin{align}
\nonumber \Hpm{\beta}{m}(z)&= \frac{\pm \i\pi}{\sin(2\pi m)}\Big(\frac{\e^{\mp\i\pi m}\J{\beta}{m}(z)}{\Gamma\naw{\frac{1}{2}-m\mp\i\beta}} - \frac{\J{\beta}{-m}(z)}{\Gamma\naw{\frac{1}{2}+m\mp\i\beta}}\Big), \\
\label{eq_missing} \J{\beta}{m}(z)& = \e^{-\pi\beta}\Big(\frac{\Hp{\beta}{m}(z)}{\Gamma\naw{\frac{1}{2}+m+\i\beta}} + \frac{\Hm{\beta}{m}(z)}{\Gamma\naw{\frac{1}{2}+m-\i\beta}}\Big).
\end{align}

The behaviour of $\Hpm{\beta}{m}$ depends qualitatively on $m$, and can be deduced from
the asymptotic behaviour of the function $\cK_{\beta,m}$ provided in \eqref{Kbm-around-zero}\;\!:
\begin{equation}\label{Hpm-around-zero}
\Hpm{\beta}{m}(z) =
\begin{cases}
\pm\i\frac{z^{\frac{1}{2}}\big(\ln(\e^{\mp\i\frac{\pi}{2}}z)\big)}{\Gamma(1\mp\i\beta)} + O\big(|z|^{1/2}\big) & \mbox{if } m=0,\\
\mp \i z^{\frac{1}{2}}\naw{\frac{\e^{\mp\i\pi m}\Gamma(-2m)z^m}{\Gamma(\frac{1}{2}- m\mp i\beta)}
+ \frac{\Gamma(2m)z^{-m}}{\Gamma(\frac{1}{2}+m\mp i\beta)}} + O\big(\abs{z}^\frac{3}{2}\big) & \mbox{if } \Re(m)=0,\ m\ne 0, \\
\mp\i\frac{\Gamma(2m)}{\Gamma(\frac{1}{2}+m\mp\i\beta)}z^{\frac{1}{2}-m}+ O\big(\abs{z}^{\frac{1}{2}+\Re(m)}\big) & \mbox{if } \Re(m)\in ]0,\frac{1}{2}], \ m\neq \frac{1}{2},\\
\mp\i\frac{1}{\Gamma(1\mp\i\beta)} + O\big(z\ln(\e^{\mp\i\frac{\pi}{2}}z)\big) & \mbox{if } m=\frac{1}{2},\\
\mp\i\frac{\Gamma(2m)}{\Gamma(\frac{1}{2}+m\mp\i\beta)}z^{\frac{1}{2}-m}+ O\big(\abs{z}^{\frac{3}{2}-\Re(m)}\big)& \mbox{if } \Re(m)>\frac{1}{2}.
\end{cases}
\end{equation}		
For the behaviour around $\infty$, we have for $|\arg(z)\mp\frac\pi2|< \pi$
\begin{equation}\label{Hpm-around-infinity}
\Hpm{\beta}{m}(z) = \e^{\mp\i\frac{\pi}{2}\naw{\frac{1}{2}+m}}\e^{\frac{\pi\beta}{2}}z^{\pm\i\beta}\;\!\e^{\pm\i\frac{z}{2}}\big(1 + O(z^{-1})\big),
\end{equation}

\subsection{Zero-energy eigenfunctions of the Whittaker operator}\label{Zero-energy eigensolutions of the Whittaker operator}

Bessel functions, which we recalled in Section \ref{sec_B},
play two roles in the present paper.
Firstly and as already explained, they are solutions of \eqref{lap7} and \eqref{lap8} in the special case of the Whittaker operator corresponding to $\beta=0$.
Secondly, after a small modification they are annihilated by the general Whittaker operator. More precisely,
for $\beta \neq 0$ let us define the following two functions on $\R_+$\;\!:
\begin{align}\label{eq_0_en}
\begin{split}
j_{\beta,m}(x)&:= x^{1/4}\cJ_{2m}\big(2\sqrt{\beta x}\big), \\
y_{\beta,m}(x)&:= x^{1/4} \cY_{2m}\big(2\sqrt{\beta x}\big).
\end{split}
\end{align}
Then, the equation
\begin{equation*}
\Big(-\partial_x^2+\big(m^2-\frac14\big)\frac{1}{x^2} - \frac{\beta}{x}\Big)v = 0,
\end{equation*}
is solved by the functions $j_{\beta,m}$ and $y_{\beta,m}$.
Indeed, this is easily observed by the following direct computation:
\begin{align*}
&\Big[\Big(-\partial_x^2+\big(m^2-\frac14\big)\frac{1}{x^2} - \frac{\beta}{x}\Big)j_{\beta,m}\Big](x) \\
& =-\beta x^{-3/4}\Big(\cJ_{2m}''\big(2\sqrt{\beta x}\big)
- \big(m^2-\tfrac{1}{16}\big)\frac{1}{\beta x} \cJ_{2m}\big(2\sqrt{\beta x}\big) + \cJ_{2m}\big(2\sqrt{\beta x}\big) \Big) \\
& =- \beta x^{-3/4}\Big(\cJ_{2m}'' - \big((2m)^2-\tfrac{1}{4}\big)\frac{1}{(2\sqrt{\beta x})^2} \cJ_{2m} + \cJ_{2m}\Big)\big(2\sqrt{\beta x}\big)
\end{align*}
and the big parenthesis vanishes.
The same argument holds for $\cJ_{2m}$ replaced by $\cY_{2m}$,
and therefore for $y_{\beta,m}$ instead of $j_{\beta,m}$.
These two functions are linearly independent. Indeed, a short computation yields
\begin{equation*}
\Wr(j_{\beta,m},y_{\beta,m};x)=\sqrt{\beta},
\end{equation*}
where the Wronskian has been introduced in \eqref{wron}.

We will need the asymptotics of these functions near zero.
Note that $y_{\beta,m}$ has the same type of asymptotics as $y_{\beta,-m}$,
which follows from the relations $\Ya_m(z)=\frac{1}{2\i}\big(\Ha_m^+(z)-\Ha_m^-(z)\big)$
together with $\Ha_{-m}^{\pm}(z)=\e^{\pm \i \pi m}\Ha_{m}^{\pm}(z)$.
Therefore, in the case of $y_{\beta,m}$ we give only the asymptotics for $\Re(m)\geqslant 0$\;\!:
\begin{align}
j_{\beta,m}(x)
\label{eq_asymp1} & =\frac{\sqrt{\pi} \beta^{\frac{1}{4}+m}}{\Gamma(1+2m)}\;\!x^{\frac{1}{2}+m}
\Big(1-\frac{\beta}{1+2m} x + O\big(x^2\big) \Big),\quad \hbox{if } -2m\not\in\N^\times,\\
\nonumber j_{\beta,m}(x) & =(-1)^{2m}\frac{\sqrt{\pi} \beta^{\frac{1}{4}-m}}{\Gamma(1-2m)}\;\!x^{\frac{1}{2}-m}
\Big(1-\frac{\beta}{1-2m} x + O\big(x^2\big) \Big),\quad \hbox{if }-2m\in\N^\times,
\\
\label{eq_asymp2}y_{\beta,m}(x) &
=
\begin{cases}
C_{\beta,0} \Big(x^{1/2} \ln(x) + O\big(x^{1/2}\big)\Big) & \hbox{if } m=0,\\
C_{\beta,m}\;\!\Big(x^{\frac{1}{2}-m} + O\big(x^{\frac{1}{2}+\Re(m)}\big)\Big) & \hbox{if } \Re(m)\in[0,\frac{1}{2}] \hbox{ and } 2m \neq 0,1, \\
C_{\beta, \frac{1}{2}}\Big(1 + O\big(x\ln(x)\big) \Big) & \hbox{if } m = \frac{1}{2}, \\
C_{\beta,m}\;\!\Big(x^{\frac{1}{2}-m} + O\big(x^{\frac{3}{2}-\Re(m)}\big)\Big) & \hbox{if } \Re(m)>\frac{1}{2},
\end{cases}
\end{align}
where $C_{\beta,m}$ are non-zero constants for $\beta \neq 0$.

The above analysis does not include the case $\beta=0$, that is the equation
$$
\Big(-\partial_x^2+\Big(m^2-\frac14\Big)\frac{1}{x^2} \Big)v=0.
$$
For completeness, let us mention that a linearly independent basis of solutions of this equation
is given by
\begin{align}\label{eq_0_en_0_beta}
\begin{split}
x^{\frac{1}{2}+m} \quad & \hbox{ and } \quad x^{\frac{1}{2}-m}\quad \hbox{if } m\neq 0,\\
x^{\frac{1}{2}}\quad & \hbox{ and }\quad x^{\frac{1}{2}}\ln(x) \quad \hbox{if } m= 0.
\end{split}
\end{align}

\section{The Whittaker operator}\label{sec_Whi_op}
\setcounter{equation}{0}

In this section we define and study the Whittaker operators $H_{\beta,m}$, which form a holomorphic family of closed operators on
the Hilbert space $L^2(\R_+)$.
This section is the main part of our paper.

\subsection{Preliminaries}

Our basic Hilbert space $L^2(\R_+)$ is endowed with the scalar product
$$
(h_1|h_2)=\int_0^\infty \overline{h_1(x)}h_2(x)\;\!\d x.
$$
The bilinear form defined by
$$
\langle h_1|h_2\rangle=\int_0^\infty h_1(x)h_2(x)\;\!\d x
$$
will also be useful.

For an operator $A$ we denote by $A^*$ its Hermitian conjugate. We will however often prefer to use the transpose of $A$,
denoted by $A^\#$, rather than $A^*$. If $A$ is bounded, then $A^*$ and $A^\#$ are defined by the relations
\begin{align}
(h_1|Ah_2)&=(A^*h_1|h_2),\\
\langle h_1|Ah_2\rangle&=\langle A^\#h_1|h_2\rangle.
\end{align}
The definitions of $A^*$ has the well-known generalization to the unbounded case.
The definition of $A^\#$ in the unbounded case is analogous.

Finally, we shall use the notation $X$ for the operator of multiplication by the variable $x$ in $L^2(\R_+)$.

\subsection{Maximal and minimal operators}

For any $\alpha, \beta\in \C$ we consider the differential expression
\begin{equation*}
L_{\beta,\alpha} :=-\partial_x^2+\Big(\alpha-\frac14\Big)\frac{1}{x^2}-\frac{\beta}{x}
\end{equation*}
acting on distributions on $\bbR_+$.
We denote by $L_{\beta,\alpha}^{\max}$ and $L_{\beta,\alpha}^{\min}$
the corresponding maximal and minimal operators
associated with it in $L^2(\R_+)$, see \cite[Sec.~4 \& App.~A]{BDG}
for the details. We also recall from this reference that the domain $\Dom(L_{\beta,\alpha}^{\max})$ is given by
\begin{equation*}
\Dom(L_{\beta,\alpha}^{\max}) = \big\{f\in L^2(\R_+) \mid L_{\beta,\alpha} f\in L^2(\R_+)\big\}
\end{equation*}
while $\Dom(L_{\beta,\alpha}^{\min})$ is
the closure of the restriction of $L_{\beta,\alpha}$ to
$C_{\rm c}^\infty(\bbR_+)$.
The operators $L_{\beta,\alpha}^{\min}$ and $L_{\beta,\alpha}^{\max}$ are closed and one observes that
\begin{equation*}
\big(L_{\beta,\alpha}^{\min}\big)^* = L_{\bar\beta,\bar \alpha}^{\max}\quad \hbox{ and } \quad
\big(L_{\beta,\alpha}^{\min}\big)^\# = L_{\beta, \alpha}^{\max}.
\end{equation*}

In order to compare the domains $\Dom(L_{\beta,\alpha}^{\min})$ and $\Dom(L_{\beta,\alpha}^{\max})$
a preliminary result is necessary.
We say that $f\in \Dom(L_{\beta,\alpha}^{\min})$ around
$0$, (or, by an abuse of notation, $f(x)\in \Dom(L_{\beta,\alpha}^{\min})$ around $0$) if there exists
$\zeta\in C_{\rm c}^\infty\big([0,\infty[\big)$ with
$\zeta=1$ around $0$ such that $f\zeta\in \Dom(L_{\beta,\alpha}^{\min})$.
Let us note that we will often write $\alpha=m^2$, where $m\in\C$.
We also recall that the functions $j_{\beta,m}$ and $y_{\beta,m}$ have been introduced in \eqref{eq_0_en} and \eqref{eq_0_en_0_beta}.

\begin{proposition}\label{lem_properties}
\begin{enumerate}
\item[(i)] If $f\in \Dom(L_{\beta,\alpha}^{\max})$, then $f$ and $f'$ are continuous functions on $\R_+$,
and converge to $0$ at infinity,
\item[(ii)] If $f\in \Dom(L_{\beta,\alpha}^{\min})$, then near $0$ one has:
\begin{enumerate}
\item If $\alpha=0$ then $f(x) = o\big(x^{\frac{3}{2}}|\ln(x)|\big)$ and $f'(x)=o\big(x^{\frac{1}{2}}|\ln(x)|\big)$,
\item If $\alpha\neq0$ then $f(x)=o\big(x^{\frac{3}{2}}\big)$ and $f'(x)=o\big(x^{\frac{1}{2}}\big)$.
\end{enumerate}
\item[(iii)] If $|\Re(m)|<1$ and $f\in \Dom(L_{\beta,m^2}^{\max})$, then there exist $a,b\in \bbC$
such that:
\begin{equation*}
f(x) -aj_{\beta,m} - by_{\beta,m} \in\Dom(L_{\beta,m^2}^{\min})\hbox{
around }0,
\end{equation*}
\item[(iv)] If $|\Re(m)|\geqslant 1$, then $\Dom(L_{\beta,m^2 }^{\min})=\Dom(L_{\beta,m^2 }^{\max})$.
\item[(v)] If $|\Re(m)| < 1$, then $\Dom(L_{\beta,m^2 }^{\min})$ is a subspace of $\Dom(L_{\beta,m^2 }^{\max})$ of codimension $2$.
\end{enumerate}
\end{proposition}

\begin{proof}
Since the above statements have already been proved for $\beta=0$ in \cite{BDG} we consider only the case $\beta\neq 0$.

From the asymptotics at zero given in \eqref{eq_asymp1} one can observe that $j_{\beta,m}$ belongs to $L^2$ near $0$ whenever $\Re(m)>-1$.
By \eqref{eq_asymp2}, the function $y_{\beta,m}$ also belongs to $L^2$ near $0$ but only for $|\Re(m)|<1$.
For other values of parameters, these functions are not $L^2$ near $0$.

The proof of (i) and (iii) consists now in a simple application of standard results on second order differential operators
as presented for example in the Appendix of \cite{BDG}.
More precisely, (i) is a direct consequence of Proposition A.2 of this reference,
while (iii) is an application of its Proposition A.5.
Statement (v) is a direct consequence of (iii).

For the statement (ii), let us write $\alpha=m^2$. First we consider the case
$|\Re(m)|<1$.
For any function $g$ which is $L^2$ near $0$, let us set $\|g\|_x:=\big(\int_0^x|g(y)|^2 \;\!\d y\big)^{1/2}$ for $x\in \R_+$ small enough.
It is then proved in Proposition A.7 of \cite{BDG} that if $f\in \Dom(L_{\beta,m^2}^{\min})$ then one has
\begin{align}
\label{desc_f} f(x) & = o(1) \big(|j_{\beta,m}(x)|\;\!\|y_{\beta,m}\|_x + |y_{\beta,m}(x)|\;\!\|j_{\beta,m}\|_x\big), \\
\label{desc_f'} f'(x) & = o(1) \big(|{j_{\beta,m}}'(x)|\;\!\|y_{\beta,m}\|_x + |{y_{\beta,m}}'(x)|\;\!\|j_{\beta,m}\|_x\big).
\end{align}
By computing these expressions in each case one gets that if $m=0$, then $\|j_{\beta,m}\|_x=O(x)$ and $\|y_{\beta,m}\|_x=O\big(x|\ln(x)|\big)$,
while if $\Re(m)\geqslant 0$ and $m\neq 0$ then $\|j_{\beta,m}\|_x=O\big(x^{1+\Re(m)}\big)$ and $\|y_{\beta,m}\|_x=O\big(x^{1-\Re(m)}\big)$.
Based on these estimates and on the asymptotic expansions of $j_{\beta,m}$, $y_{\beta,m}$ near $0$ one directly infers from \eqref{desc_f}
the estimate on $f$ for any $f\in \Dom(L_{\beta,m^2}^{\min})$.
For the estimate on $f'$, it is necessary to compute ${j_{\beta,m}}' $, ${y_{\beta,m}}' $, and the only surprise comes from the special case $m=\frac{1}{2}$.
By using \eqref{desc_f'} and the estimate on $\|j_{\beta,m}\|_x$, $\|y_{\beta,m}\|_x$ obtained above, one deduces the behavior of $f'$
for any $f\in \Dom(L_{\beta,m^2}^{\min})$.

The case $|\Re(m)|\geqslant 1$ of the statement (ii) follows by an obvious modification of the proof of Prop. 4.11 of \cite{BDG}.

The proof of statement (iv) is deferred to Subsection \ref{sec_resol}.
\end{proof}

\subsection{The holomorphic family of Whittaker operators}

Recall from \eqref{eq_asymp1}
that if $-2m\not\in\N^\times$, then
$ j_{\beta,m}(x)=C_{m,\beta}x^{\frac{1}{2}+m}
\big(1-\frac{\beta}{1+2m} x + O(x^2) \big)$, which
belongs to $L^2$ near $0$ if $\Re(m)>-1$.
This motivates the following definition:

For $m,\beta\in \C$ with $\Re(m)>-1$,
except for the case $m=-\frac{1}{2}$, $\beta\neq0$,
we define the closed operator $H_{\beta,m}$ as the
restriction of $L_{\beta,m^2}^{\max}$ to the domain
\begin{align}
\nonumber \Dom(H_{\beta,m}) & = \Big\{f\in \Dom(L_{\beta,m^2}^{\max})\mid
\hbox{ for some } c \in \bbC,\\
\label{eq_def_H} &\qquad f(x)- c x^{\frac{1}{2}+m}
\Big(1-\frac{\beta}{1+2m} x \Big)\in\Dom(L_{\beta,m^2}^{\min})\hbox{ around }
0\Big\}.
\end{align}
Note that for $\beta=0$, the expression $\frac{\beta}{1+2m}$ is interpreted as $0$, also in the case $m=-\frac12$.
In the exceptional case excluded above we set
\begin{equation}\label{exce}
H_{\beta,-\frac{1}{2}}:= H_{\beta,\frac{1}{2}},\quad \beta\neq0.
\end{equation}

Let us stress that \eqref{exce} does not extend to $\beta=0$. In fact,
$H_{0,-\frac12}\neq H_{0,\frac12}$, as we know from \cite{BDG,DR}:
$H_{0,-\frac12}$ is the Neumann Laplacian on $\R_+$, and $H_{0,\frac12}$ is the Dirichlet Laplacian on $\R_+$.
More information about the singularity at $(\beta,m)=(0,-\frac12)$ will be provided in Proposition \ref{more}.

The following statements can be proved directly:
\begin{theorem}
\begin{enumerate}
\item[(i)] For any $m\in \C$ with $\Re(m)>-1$ and any $\beta \in \C$ one has
\begin{equation*}
\big(H_{\beta,m}\big)^* = H_{\bar \beta,\bar m},\quad \big(H_{\beta,m}\big)^\# = H_{ \beta, m},
\end{equation*}
\item[(ii)] For any real $m>-1$ and for any real $\beta\in \R$ the operator $H_{\beta,m}$ is self-adjoint,
\item[(iii)] For $\Re(m) \geqslant 1$,
$$
L_{\beta,m^2 }^{\min}= H_{\beta,m}= L_{\beta,m^2 }^{\max}.
$$

\item[(iv)] For $-1<\Re(m)<1$,
$$
L_{\beta,m^2 }^{\min}\subsetneq H_{\beta,m}\subsetneq L_{\beta,m^2 }^{\max},
$$
and the inclusions of the corresponding domains are of codimension $1$.
\end{enumerate}\end{theorem}

\begin{proof}
Recall from \cite[Prop.~A.2]{BDG} that for any $f\in \Dom(L_{\beta,m^2}^{\max})$ and $g\in \Dom(L_{\bar \beta,\bar m^2}^{\max})$,
the functions $f,f',g,g'$ are continuous on $\R_+$. In addition, the Wronskian of $\bar f$ and $g$, as introduced in \eqref{wron},
possesses a limit at zero, and we have the equality
\begin{equation*}
(L_{\beta,m^2}^{\max}f|g) - (f|L_{\bar \beta,\bar m^2}^{\max}g) = -\Wr(\bar f,g;0).
\end{equation*}
In particular, if $f\in \Dom(H_{\beta,m})$ one infers that
\begin{equation*}
(H_{\beta,m}f|g) = (f|L_{\bar \beta, \bar m^2}^{\max}g) -\Wr(\bar f,g;0).
\end{equation*}
Thus, $g\in \Dom\big((H_{\beta,m})^*\big)$ if and only if $\Wr(\bar f,g;0)=0$, and then
$(H_{\beta,m})^*g=L_{\bar \beta,\bar m^2}^{\max}g$.
By taking into account the explicit description of $\Dom(H_{\beta,m})$,
straightforward computations show that $\Wr(\bar f,g;0)=0$ if and only
if $g\in \Dom(H_{\bar \beta,\bar m})$. One then deduces that
$(H_{\beta,m})^*= H_{\bar \beta,\bar m}$.
Note that the property for the transpose of $H_{\beta,m}$ can be proved similarly,
which finishes the proof of (i).

The statement (ii) is a straightforward consequence of the statement (i).
The statements (iii) and (iv) are consequences of Proposition \ref{lem_properties}.
\end{proof}

\begin{remark}
In the spirit of \cite{DR} one could consider more general boundary conditions, and thus other realizations
of the Whittaker operator.
However, in this paper we stick to the most natural boundary conditions introduced above.
This approach corresponds to the one of the original paper \cite{BDG}, where $\beta=0$.
\end{remark}

\subsection{The resolvent}\label{sec_resol}

From now on, we consider fixed $m,\beta\in \C$ with $\Re(m)>-1$.
In order to study the resolvent of the operator $H_{\beta,m}$, let us introduce the set $\sigma_{\beta,m}\subset \C$
which will be related later on to the spectrum of $H_{\beta,m}$\;\!:
$$
\sigma_{\beta,m}:=\Big\{k\in \C\mid \Re(k)>0 \hbox{ and } \frac{\beta}{2k}-m-\frac{1}{2}\not \in \N\Big\}.
$$
Let us consider $k \in \sigma_{\beta,m}$.
By a scaling argument together with the material of Section \ref{sec_B_functions} one easily observes that the two functions
\begin{equation}\label{eq_2_sol}
x\mapsto \K{\frac{\beta}{2k}}{m}(2kx) \qquad \hbox{and}\qquad
x\mapsto \I{\frac{\beta}{2k}}{m}(2kx)
\end{equation}
are linearly independent solutions of the equation $(L_{\beta,m^2}+k^2)v=0$.
From \eqref{Kbm-around-infinity} one infers that
the first function is always in $L^2$ near infinity, but it belongs
to $L^2$ near zero only
for $|\Re(m)|<1$.
On the other hand, the second function belongs to $L^2$ around $0$
for any $m$ with $\Re(m)>-1$ but it does not belong to $L^2$ near infinity.

If in addition $m\neq-\frac12$, then one has
\begin{equation*}
\I{\beta}{m}(x)\sim\frac{x^{\frac12+m}}{\Gamma(1+2m)}\Big(1-\frac{\beta}{(1+2m)}x\Big).
\end{equation*}
Therefore, it follows that
\begin{equation}
\I{\frac{\beta}{2k}}{m}(2kx)
\sim\frac{(2kx)^{\frac12+m}}{\Gamma(1+2m)}\Big(1-\frac{\beta}{(1+2m)}x\Big),
\label{whitti}
\end{equation}
which means that \eqref{whitti} belongs to the domain of $H_{\beta,m}$ around $0$.
Based on these observations and on the standard theory of Green's function, we expect that the inverse of the operator $H_{\beta,m}+k^2$ for suitable $k$
is given by the operator $R_{\beta,m}(-k^2)$ whose kernel is given for $x,y\in \R_+$ by
\begin{align}\label{The-resolvent}
\nonumber &R_{\beta,m}(-k^2;x,y)\\
& := \tfrac{1}{2k}\Gamma\naw{\tfrac{1}{2}+m-\tfrac{\beta}{2k}} \begin{cases} \I{\frac{\beta}{2k}}{m}(2k x)\K{\frac{\beta}{2k}}{m}(2k y) & \mbox{ for }0<x<y,\\
\I{\frac{\beta}{2k}}{m}(2k y)\K{\frac{\beta}{2k}}{m}(2k x) & \mbox{ for }0<y<x.
\end{cases}
\end{align}
We still need to check the exceptional case $m=-\frac12$. By \eqref{dege} and
\eqref{eq_sym}, we have
$$
\cI_{\beta,-\frac12}(x)=-\beta \cI_{\beta,\frac12}(x),\quad
\cK_{\beta,-\frac12}(x)=\cK_{\beta,\frac12}(x).
$$
As a consequence we infer that
\begin{equation*}
R_{\beta,-\frac{1}{2}}(-k^2;x,y)= R_{\beta,\frac{1}{2}}(-k^2;x,y),\quad\beta\neq0,
\end{equation*}
which is consistent with \eqref{exce}.

\begin{remark}
If $\beta=0$, by taking the relations \eqref{eqI_{0,m}} and
\eqref{eqK_{0,m}} into account one infers that
\begin{equation*}
R_{0,m}(-k^2;x,y) = \frac{1}{k}\begin{cases} \Ia_m(k x)\Ka_m(k y) & \mbox{ for }0<x<y,\\
\Ia_m(k y)\Ka_m(k x) & \mbox{ for }0<y<x.
\end{cases}
\end{equation*}
This expression corresponds to the starting point for the study of the resolvent in \cite{BDG}.
\end{remark}

The next statement provides the precise link between the resolvent of $H_{\beta,m}$ and the operator $R_{\beta,m}(-k^2)$.

\begin{theorem}\label{thm_resolvent}
Let $m,\beta\in \C$ with $\Re(m)>-1$ and let $k\in \sigma_{\beta,m}$.
Then the operator $R_{\beta,m}(-k^2)$ defined by the kernel \eqref{The-resolvent}
belongs to $\B\big(L^2(\R_+)\big)$ and equals $(H_{\beta,m}+k^2)^{-1}$.
Moreover, the map $(\beta,m)\mapsto H_{\beta,m}$ is a holomorphic family of closed operators except for a singularity at $(\beta,m)=(0,-\frac12)$.
\end{theorem}

Let us emphasize that this statement already provides information about the spectrum $\sigma(H_{\beta,m})$ of $H_{\beta,m}$.
Indeed, one infers that
\begin{align}\label{pqpq}
\nonumber \sigma(H_{\beta,m})&\subset \big\{-k^2 \mid \Re(k)\geqslant 0 \hbox{ and } k\not \in \sigma_{\beta,m}\big\}\\
&=[0,\infty[
\,\bigcup\Big\{\lambda_N \mid N\in \N, N+m+\frac12\neq0, \quad \Re\Big(\frac{\beta}{N+m+\frac{1}{2}}\Big)>0\Big\},
\end{align}
where we have set
\begin{equation}
\lambda_N:= -\frac{\beta^2}{4(N+m+\frac{1}{2})^2}.\label{delam}
\end{equation}

Later on, we shall see that the inclusion in \eqref{pqpq} is in fact an equality.
The proof of Theorem \ref{thm_resolvent} is based on a preliminary technical lemma.

\begin{lemma}
Let $m,\beta\in \C$ with $\Re(m)>-1$ and let $k\in \sigma_{\beta,m}$. Then for any $x,y\in \R_+$ one has:
\begin{enumerate}
\item[(i)] If $\Re(m)\geqslant 0$ with $m\neq 0$ then
\begin{align}\label{eq_est1}
\nonumber & |R_{\beta,m}(-k^2;x,y)| \\
\nonumber & \leqslant C_{\frac{\beta}{2k},m}^2 \tfrac{|\Gamma(\frac{1}{2}+m-\frac{\beta}{2k})|}{2|k|}
\e^{-|x-y|\Re(k)} \min\{1,2x|k|)\}^{\frac{1}{2}} \min\{1,2y|k|\}^{\frac{1}{2}} \\
& \quad \times
\begin{cases} \max\{1,2x|k|\}^{-\Re(\frac{\beta}{2k})}\max\{1,2y|k|\}^{\Re(\frac{\beta}{2k})} & \mbox{ for }0<x<y,\\
\max\{1,2y|k|\}^{-\Re(\frac{\beta}{2k})}\max\{1,2x|k|\}^{\Re(\frac{\beta}{2k})} & \mbox{ for }0<y<x.
\end{cases}
\end{align}
\item[(ii)] If $\Re(m)\leqslant 0$ with $m\neq 0$ then
\begin{align}\label{eq_est2}
\nonumber &|R_{\beta,m}(-k^2;x,y)| \\
\nonumber & \leqslant C_{\frac{\beta}{2k},m}^2 \tfrac{|\Gamma(\frac{1}{2}+m-\frac{\beta}{2k})|}{2|k|}
\e^{-|x-y|\Re(k)} \min\{1,2x|k|\}^{\Re(m)+\frac{1}{2}} \min\{1,2y|k|\}^{\Re(m)+\frac{1}{2}} \\
& \quad \times
\begin{cases} \max\{1,2x|k|\}^{-\Re(\frac{\beta}{2k})}\max\{1,2y|k|\}^{\Re(\frac{\beta}{2k})} & \mbox{ for }0<x<y,\\
\max\{1,2y|k|\}^{-\Re(\frac{\beta}{2k})}\max\{1,2x|k|\}^{\Re(\frac{\beta}{2k})} & \mbox{ for }0<y<x.
\end{cases}
\end{align}
\item[(iii)] If $m=0$ then
\begin{align}\label{eq_est3}
\nonumber &|R_{\beta,0}(-k^2;x,y)| \\
\nonumber & \leqslant C_{\frac{\beta}{2k}}^2 \tfrac{|\Gamma(\frac{1}{2}-\frac{\beta}{2k})|}{2|k|}
\e^{-|x-y|\Re(k)} \min\{1,2x|k|\}^{\frac{1}{2}} \min\{1,2y|k|\}^{\frac{1}{2}} \\
\nonumber & \quad \times \big(1+\big|\ln(\min\{1,2x|k|\})\big|\big) \big(1+\big|\ln(\min\{1,2y|k|\})\big|\big) \\
& \quad \times
\begin{cases} \max\{1,2x|k|\}^{-\Re(\frac{\beta}{2k})}\max\{1,2y|k|\}^{\Re(\frac{\beta}{2k})} & \mbox{ for }0<x<y,\\
\max\{1,2y|k|\}^{-\Re(\frac{\beta}{2k})}\max\{1,2x|k|\}^{\Re(\frac{\beta}{2k})} & \mbox{ for }0<y<x.
\end{cases}
\end{align}
\end{enumerate}
The constants $C_{\frac{\beta}{2k},m}$ and $C_{\frac{\beta}{2k}}$ are independent of $x$ and $y$.
\end{lemma}

\begin{proof}
Observe first that for $\eps>0$ and $|\arg(z)|<\pi-\eps$ one deduces from
\eqref{Kbm-around-zero} and \eqref{Kbm-around-infinity} that
\begin{equation*}
|\K{\frac{\beta}{2k}}{m}(z)| \leqslant C_{\frac{\beta}{2k},m}\e^{-\Re(z)/2}\min\{1,|z|\}^{-|\Re(m)|+\frac{1}{2}} \max\{1,|z|\}^{\Re(\frac{\beta}{2k})}
\end{equation*}
for $m\neq 0$, while for $m=0$
\begin{equation*}
|\K{\frac{\beta}{2k}}{0}(z)| \leqslant C_{\frac{\beta}{2k}} \e^{-\Re(z)/2} \min\{1,|z|\}^{\frac{1}{2}} \big(1+\big|\ln(\min\{1,|z|\})\big|\big) \max\{1,|z|\}^{\Re(\frac{\beta}{2k})}.
\end{equation*}
Similarly, from \eqref{Ibm-around-zero} and \eqref{Ibm-around-infinity} one infers that
\begin{align*}
|\I{\frac{\beta}{2k}}{m}(z)| & \leqslant C_{\frac{\beta}{2k},m}\e^{\Re(z)/2}\min\{1,|z|\}^{\Re(m)+\frac{1}{2}}\max\{1,|z|\}^{-\Re(\frac{\beta}{2k})} \quad \hbox{for } m\neq 0, \\
|\I{\frac{\beta}{2k}}{0}(z)| & \leqslant C_{\frac{\beta}{2k}} \e^{\Re(z)/2} \min\{1,|z|\}^{\frac{1}{2}} \max\{1,|z|\}^{-\Re(\frac{\beta}{2k})} .
\end{align*}

As a consequence of these estimates, if $m\neq 0$ one infers that for $0<x<y$
\begin{align*}
& |R_{\beta,m}(-k^2;x,y)|\\
& \leqslant C_{\frac{\beta}{2k},m}^2 \frac{\big|\Gamma(\frac{1}{2}+m-\frac{\beta}{2k})\big|}{2|k|}
\e^{(x-y)\Re(k)} \min\{1,2x|k|\}^{\Re(m)+\frac{1}{2}} \\
& \quad \times \min\{1,2y|k|\}^{-|\Re(m)|+\frac{1}{2}} \max\{1,2x|k|\}^{-\Re(\frac{\beta}{2k})}\max\{1,2y|k|\}^{\Re(\frac{\beta}{2k})}
\end{align*}
while for $0<y<x$ one has
\begin{align*}
& |R_{\beta,m}(-k^2;x,y)|\\
& \leqslant C_{\frac{\beta}{2k},m}^2 \frac{\big|\Gamma(\frac{1}{2}+m-\frac{\beta}{2k})\big|}{2|k|}
\e^{(y-x)\Re(k)} \min\{1,2y|k|\}^{\Re(m)+\frac{1}{2}} \\
& \quad \times \min\{1,2x|k|\}^{-|\Re(m)|+\frac{1}{2}} \max\{1,2y|k|\}^{-\Re(\frac{\beta}{2k})}\max\{1,2x|k|\}^{\Re(\frac{\beta}{2k})}.
\end{align*}
Then, if $\Re(m)\geqslant 0$ one observes that $\big|\frac{kx}{ky}\big|<1$ in the first case, and
$\big|\frac{ky}{kx}\big|<1$ in the second case. This directly leads to the first part of the statement.
Similarly, for $\Re(m)\leqslant 0$ one has $-|\Re(m)|=\Re(m)$, from which one infers the second part of the statement.
The special case $m=0$ is straightforward.
\end{proof}

\begin{proof}[Proof of Theorem \ref{thm_resolvent}]
Observe first that for $k\in \sigma_{\beta,m}$ the Gamma factor in \eqref{The-resolvent} is harmless.
Thus, in order to show that the kernel \eqref{The-resolvent}, with the Gamma factor removed, defines a bounded operator for any $k\in \C$ with $\Re(k)>0$, it
is sufficient to consider separately the two regions
$$
\Omega:=\Big\{(x,y)\in \R_+ \times \R_+ \mid x\geqslant (2|k|)^{-1},y\geqslant [(2|k|)^{-1}\Big\}
$$
and $\R_+ \times \R_+\setminus \Omega$.
In the latter region, thanks to the previous lemma it is easily seen that the kernel $R_{\beta,m}(-k^2;\cdot,\cdot)$
belongs to $L^2$, and thus defines a Hilbert--Schmidt operator.
For the kernel on $\Omega$ one can employ Schur's test and observe that $R_{\beta,m}(-k^2;\cdot,\cdot)$
belongs to $L^\infty\big([|k|^{-1},\infty[;L^1([|k|^{-1},\infty[)\big)$ for the two variables taken in arbitrary order.
If $\Re\big(\frac{\beta}{2k}\big)\leqslant 0$, then this computation is easy and reduced to the one already performed in the proof of \cite[Lem.~4.4]{BDG}.
We shall consider only the case $\Re\big(\frac{\beta}{2k}\big)>0$.

Thus, for $\Re\big(\frac{\beta}{2k}\big)>0$ let us check that
\begin{equation*}
\sup_{y\geqslant (2|k|)^{-1}}\int_{(2|k|)^{-1}}^\infty \big|R_{\beta,m}(-k^2;x,y)\big|\;\!\d x <\infty,
\end{equation*}
the other condition being obtained similarly.

For fixed $y\in \big[(2|k|)^{-1},\infty\big[$
we divide the above integral into three parts, namely $x\in \big[(2|k|)^{-1},\frac{(2|k|)^{-1}+y}{2}\big[$, $x\in \big[\frac{(2|k|)^{-1}+y}{2},y\big[$ and $x\in[y,\infty[$.
For the first part it is enough to observe that
\begin{align}\label{est1}
\nonumber & \int_{(2|k|)^{-1}}^{\frac{(2|k|)^{-1}+y}{2}}\e^{-(y-x)\Re(k)} \big(\tfrac{y}{x}\big)^{\Re(\frac{\beta}{2k})}\;\!\d x \\
\nonumber & \leqslant \e^{-y\Re(k)} (2|k|y)^{\Re(\frac{\beta}{2k})}\int_{(2|k|)^{-1}}^{\frac{(2|k|)^{-1}+y}{2}} \e^{x\Re(k)}\;\!\d x \\
& = \tfrac{1}{\Re(k)}\e^{-y\Re(k)} (2|k|y)^{\Re(\frac{\beta}{2k})} \big(\e^{(\frac{y+(2|k|)^{-1}}{2})\Re(k)}-\e^{(2|k|)^{-1}\Re(k)}\big).
\end{align}
For the second part one observes that
\begin{align}\label{est2}
\nonumber & \int_\frac{(2|k|)^{-1}+y}{2}^{y}\e^{-(y-x)\Re(k)} \big(\tfrac{y}{x}\big)^{\Re(\frac{\beta}{2k})}\;\!\d x \\
\nonumber & \leqslant \e^{-y\Re(k)} 2^{\Re(\frac{\beta}{2k})} \int_\frac{(2|k|)^{-1}+y}{2}^{y} \e^{x\Re(k)}\;\!\d x \\
& = \tfrac{2^{\Re(\frac{\beta}{2k})}}{\Re(k)}\e^{-y\Re(k)} \big(\e^{y\Re(k)}-\e^{(\frac{y+(2|k|)^{-1}}{2})\Re(k)}\big).
\end{align}
For the third part one has (with $y\geqslant (2|k|)^{-1}$)
\begin{align}\label{est3}
\nonumber &\int_y^\infty \e^{-(x-y)\Re(k)} \big(\tfrac{x}{y}\big)^{\Re(\frac{\beta}{2k})}\;\!\d x \\
\nonumber & = \int_0^\infty \e^{-z\Re(k)}\big(1+\tfrac{z}{y}\big)^{\Re(\frac{\beta}{2k})}\;\!\d z \\
& \leqslant \int_0^\infty \e^{-z\Re(k)}\big(1+2|k|z\big)^{\Re(\frac{\beta}{2k})}\;\!\d z .
\end{align}
Finally, it only remains to observe that the three expressions \eqref{est1}, \eqref{est2} and \eqref{est3}
are bounded for $y\geqslant (2|k|)^{-1}$. As a consequence, one deduces that the kernel restricted to $\Omega$ defines a bounded operator in $L^2(\R_+)$,
and by summing up the information one deduces the boundedness of $R_{\beta,m}(-k^2)$.
The equality of $R_{\beta,m}(-k^2)$ with $(H_{\beta,m}+k^2)^{-1}$ and the mentioned holomorphic
property follow from standard argument, see for example the Appendix A and Proposition 2.3 in \cite{BDG}.
\end{proof}

\begin{proof}[Proof of Proposition \ref{lem_properties} (iv)]
We want to show that $\Re(m)\geqslant 1$ implies $L_{\beta,m^2 }^{\min} = L_{\beta,m^2 }^{\max}$.
With the information provided in the previous statements, it can be done by copying {\it mutatis mutandis} the proof of \cite[Prop.~4.10]{BDG}.
\end{proof}

\subsection{Point spectrum and eigenprojections}

In this section we provide more information on the point spectrum of $H_{\beta,m}$
and exhibit an expression for the projection on the corresponding eigenfunctions.

\begin{theorem}\label{thm_spectrum}
Let $m, \beta\in \C$ with $\Re(m)>-1$.
Then we have
\begin{align}\notag
\sigma_\mathrm{p}(H_{\beta,m})&=
\sigma_\mathrm{d}(H_{\beta,m})\\
&=\Big\{\lambda_N \mid N\in \N, N+m+\frac12\neq0, \quad \Re\Big(\frac{\beta}{N+m+\frac{1}{2}}\Big)>0\Big\},\label{mio2}\end{align}
where $\lambda_N$ were defined in \eqref{delam}.
All eigenvalues are of multiplicity $1$.
The kernel of the Riesz projection $P_N$ corresponding to $\lambda_N$ is given for $x,y\in \R_+$ by
\begin{align}\label{Eigenprojection-eq}
\nonumber
P_N(x,y) &= \frac{N!}{\Gamma(1+2m+N)}\Big( \frac{\beta}{N+m+\frac{1}{2}}\Big)^{1+2m}\exp\Big(-\frac{\beta}{2(N+m+\frac{1}{2})}(x+y)\Big)\\
&\quad \times(xy)^{\frac{1}{2}+m} L^{(2m)}_N\Big( \frac{\beta}{N+m+\frac{1}{2}}\;\!x\Big)L^{(2m)}_N\Big( \frac{\beta}{N+m+\frac{1}{2}}\;\!y\Big),
\end{align}
where $L_N^{(2m)}$ is the Laguerre polynomial introduced in \eqref{eq_Laguerre}.
\end{theorem}

For the following proof let us observe that we can consider $\beta\neq 0$ since the second condition in \eqref{mio2}
is never satisfied for $\beta=0$. In addition, the case $\beta=0$ has already been considered in \cite{BDG}
and it was shown in this case that the operator $H_{0,m}$ had no point spectrum.

\begin{proof}	
Observe first that for any $N\in\N$ the three conditions $\Re(k)>0$, $\frac{\beta}{2k} -m-\frac{1}{2}\neq N$ and $\lambda_N=-k^2$
are equivalent to the condition in \eqref{mio2}.
Now, this situation takes place exactly when the two solutions of the equation $L_{\beta,m^2}f=-k^2f$ provided in
\eqref{eq_2_sol} are not linearly independent, see also \eqref{wron1}.
This means that modulo a multiplicative constant, for $k=\frac{\beta}{2(N+m+\frac{1}{2})}$
the map $x\mapsto \I{N+m+\frac{1}{2}}{m}\big(\frac{\beta}{N+m+\frac{1}{2}}x\big)$ and the map
$x\mapsto \K{N+m+\frac{1}{2}}{m}\big(\frac{\beta}{N+m+\frac{1}{2}}x\big)$ are equal.
From the discussion following \eqref{eq_2_sol}, one infers that these functions
belong to $L^2(\R_+)$ for any $\Re(m)>-1$. It remains to show that these functions belong to $\Dom(H_{\beta,m})$.
For that purpose let us consider one of them and use \eqref{Ibm-around-zero} to get
\begin{align*}
&\I{N+m+\frac{1}{2}}{m}\Big(\frac{\beta}{N+m+\frac{1}{2}}x\Big) \\
&= \frac{\big(\frac{\beta}{N+m+\frac{1}{2}}x\big)^{\frac{1}{2}+m}}{\Gamma(1+2m)} \Big(1 -\frac{N+m+\frac{1}{2}}{1+2m}\frac{\beta}{N+m+\frac{1}{2}}x+ O\big(x^2\big)\Big)\\
&= \frac{\big(\frac{\beta}{N+m+\frac{1}{2}}\big)^{\frac{1}{2}+m}}{\Gamma(1+2m)} x^{\frac{1}{2}+m}\Big(1-\frac{\beta}{1+2m}x+O\big(x^2\big)\Big).
\end{align*}
By comparing this expression with the description of $\Dom(H_{\beta,m})$
one directly deduces the first statement of the theorem.

Let $\gamma$ be a contour encircling an eigenvalue $\lambda_N$ in the complex plane, with no other eigenvalue inside $\gamma$ and with no intersection with $[0,\infty[$.
The Riesz projection corresponding to this eigenvalue is then given by
\begin{equation*}
P_N = -\frac{1}{2\pi\i}\int_{\gamma}R_{\beta,m}(z)\;\!\d z.
\end{equation*}
By setting $z=-k^2$ we get
\begin{equation}
P_N = -\frac{1}{2\pi\i}\int_{\gamma}R_{\beta,m}(-k^2)\;\!\d (-k^2) = \frac{1}{2\pi\i}\int_{\gamma^*}2kR_{\beta,m}(-k^2)\;\!\d k
\end{equation}
for some appropriate curve $\gamma^*$.
Now by looking at the expression for the resolvent provided in \eqref{The-resolvent}, one observes that only the first factor is singular for
$\frac{1}{2}+m-\frac{\beta}{2k}=-N$ and more precisely one gets for the residue of this term
$\mathrm{Res}(\Gamma,-N)=\frac{(-1)^N}{N!}$. By substituting $k = \frac{\beta}{2(N+m+\frac{1}{2})}$ in the expression for the resolvent one thus gets
\begin{equation*}
P_N(x,y) = \frac{(-1)^N}{N!}
\begin{cases} \cI_{N+m+\frac{1}{2},m}\big(\frac{\beta}{N+m+\frac{1}{2}} x\big)\cK_{N+m+\frac{1}{2},m}\big(\frac{\beta}{N+m+\frac{1}{2}} y\big) \\
\qquad\qquad\qquad\qquad\qquad\qquad\mbox{ for }0<x<y,\\
\cI_{N+m+\frac{1}{2},m}\big(\frac{\beta}{N+m+\frac{1}{2}} y\big)\cK_{N+m+\frac{1}{2},m}\big(\frac{\beta}{N+m+\frac{1}{2}} x\big) \\
\qquad\qquad\qquad\qquad\qquad\qquad \mbox{ for }0<y<x.
\end{cases}
\end{equation*}
Finally, by recalling that for $\beta=N+m+\frac{1}{2}$ the functions $\cI_{\beta,m}$ and $\cK_{\beta,m}$ have an easy behaviour and are essentially the same,
as mentioned in Sections \ref{sec_f1} and \ref{sec_f2}, one directly infers the explicit formula provided in \eqref{Eigenprojection-eq}.

It remains to show that there are no eigenvalues in $[0,\infty[$. We will consider separately $0$ and $]0,\infty[$.
Firstly let us consider the functions $x\mapsto h_{\beta,m}^\pm(x):=x^{1/4}\cH_{2m}^\pm\big(2\sqrt{\beta x}\big)$.
By the arguments of
Subsection \ref{Zero-energy eigensolutions of the Whittaker operator}, they satisfy
$L_{\beta,m^2}h_{\beta,m}^\pm=0$.
By the asymptotic expansions of these functions near $0$ provided in
\cite[App.~A.5]{DR} one easily infers that $h_{\beta,m}^\pm$ are $L^2$ near $0$ for $|\Re(m)|<1$ and not otherwise.
Also, since for large $z$ one has
\begin{equation*}
\Ha_m^\pm(z) = \e^{\pm\i (z-\frac{1}{2}\pi m-\frac{1}{4}\pi)}\big(1+O(|z|^{-1})\big)\ ,
\end{equation*}
we deduce that
$$
h_{\beta,m}^\pm(x) = x^{1/4}\e^{\pm\i (2\sqrt{\beta x}-\pi m-\frac{1}{4}\pi)}\big(1+O(|z|^{-1/2})\big),
$$
which means that one (and only one) of these functions is in $L^2$ near infinity if and only $\sqrt{\beta}$ has a non-zero imaginary part.
However, since none of these functions has an asymptotic behavior near $0$ of the form $x^{\frac{1}{2}+m}
\big(1-\frac{\beta}{1+2m} x \big)+o(x^{\frac{3}{2}})$, one deduces that
none of them belongs to $\Dom(H_{\beta,m})$. Hence $0$ is never an eigenvalue of $H_{\beta,m}$.

Let us now consider the equation $L_{\beta,m^2}v=\mu^2 v$ for some $\mu>0$. Two linearly independent solutions
are provided by the functions $x\mapsto \Hpm{\frac{\beta}{2\mu}}{m}(2\mu x)$ introduced in Section \ref{sec_f4}.
By the asymptotic expansion around $0$ provided in \eqref{Hpm-around-zero}, one infers that these functions
are $L^2$ near $0$ if $|\Re(m)|<1$ and not otherwise.
Then, from the asymptotic expansion near $\infty$ provided in \eqref{Hpm-around-infinity} one deduces that
\begin{equation*}
\Hpm{\frac{\beta}{2\mu}}{m}(2\mu x) = \e^{\mp\i\frac{\pi}{2}\naw{\frac{1}{2}+m}}\e^{\frac{\pi\beta}{4\mu}}(2\mu x)^{\pm\i\frac{\beta}{2\mu}}
\;\!\e^{\pm\i\mu x}\big(1 + O(x^{-1})\big).
\end{equation*}
Again, one infers that one (and only one) of these functions is in $L^2$ near infinity if and only $\beta$ has a non-zero imaginary part.
However, by taking the asymptotic expansion near $0$ provided in \eqref{Hpm-around-zero}, one observes that none of these functions
belongs to $\Dom(H_{\beta,m})$, from which we deduce that $\mu^2$ is never an eigenvalue of $H_{\beta,m}$.
\end{proof}

Let us still describe more precisely the point spectrum $\sigma_\mathrm{p}\big(H_{\beta,m}\big)$
when the operator $H_{\beta,m}$ is self-adjoint, which means when $\beta$ and $m$ are real.
\begin{corollary}
\begin{enumerate}
\item[(i)] For $m\in[-1/2,\infty[$ and $\beta<0$ one has $\sigma_\mathrm{p}\big(H_{\beta,m}\big)=\emptyset$,
\item[(ii)] For $m\in]-1/2,\infty[$ and $\beta>0$ one has $\sigma_\mathrm{p}\big(H_{\beta,m}\big)=\big\{-\frac{\beta^2}{4(N+m+\frac12)^2}\mid N\in \N\big\}$,
\item[(iii)] For $m\in]-1,-1/2[$ and $\beta<0$ one has $\sigma_\mathrm{p}\big(H_{\beta,m}\big)=\big\{-\frac{\beta^2}{4(m+\frac12)^2}\big\}$,
\item[(iv)] For $m\in]-1,-1/2]$ and $\beta>0$ one has $\sigma_\mathrm{p}\big(H_{\beta,m}\big)=\big\{-\frac{\beta^2}{4(N+m+\frac12)^2}\ |\ N\in \N^\times\big\}$.
\end{enumerate}
\end{corollary}

The singularity of the holomorphic function $(\beta,m)\mapsto H_{\beta,m}$ at $(0,-\frac12)$ may seem surprising.
The following proposition helps to explain why this singularity arises. It indicates that the point spectrum has a rather wild behavior
for parameters near this singularity.

\begin{proposition}\label{more}
For every neighborhood $\cV$ of $(0,-\frac12)$ in $\C\times\C$ and every $z\in\C$ we can find $(\beta,m)\in\cV$ such that $z\in\sigma(H_{\beta,m})$.
\end{proposition}

\begin{proof}
Clearly, one has $[0,\infty[\subset \sigma(H_{\beta,m})$. Moreover, if $z\not\in[0,\infty[$, $\epsilon>0$,
$\beta_\epsilon:=\epsilon2\sqrt{-z}$ and $m_\epsilon:=-\frac12+\epsilon$, then
one has $z\in\sigma(H_{\beta_\epsilon,m_\epsilon})$ as a consequence Theorem \ref{thm_spectrum} for $N=0$.
Clearly, $(\beta_\epsilon,m_\epsilon)\to(0,-\frac12)$ as $\epsilon \to 0$.
\end{proof}

\begin{remark}
Let us recall that when $\beta=0$ one has $\sigma_\mathrm{p}\big(H_{0,m}\big)=\emptyset$ and $\sigma(H_{0,m})=[0,\infty[$
for any $m$ with $m>-1$, as shown in \cite{BDG}.
In that respect, the result obtained in (iii) sounds surprising, since for $\beta<0$ it may seem that $-\frac{\beta}{x}$ is a positive perturbation
of $H_{0,m}$, but nevertheless $H_{\beta,m}$ has a negative eigenvalue! However, let us emphasize that there is no contradiction
since the domains of $H_{\beta,m}$ and $H_{0,m}$ are not the same: no inference can be made.
\end{remark}

\begin{remark}
For $\beta$ and $m$ as in the \emph{physical} quantum-mechanical hydrogen atom, the set $\{\lambda_N\}_{N\in \N}$
coincides with the usual hydrogen atom point spectrum.
Physicists introduce non-negative integers $\ell:=m-\frac{1}{2}$ and $n:=\ell+N+1$ which are called respectively \emph{the azimuthal quantum number} and
\emph{the main quantum number}.
Then by considering $0\leqslant \ell \leqslant n-1$ these numbers give the $n$-fold degeneracy of the eigenvalue $E_n = -\frac{\beta^2}{4n^2}$.
\end{remark}

\subsection{Dilation analyticity}

The group of dilations is defined for any $\theta \in \R$ by
$$
U_\theta f(x):=\e^{\frac{\theta}{2}}f(\e^\theta x),\quad
f\in L^2(\R_+).
$$
It is easily observed that $U_\theta \Dom(H_{\beta,m}) = \Dom(H_{\e^\theta \beta,m})$ and
\begin{equation}\label{anal}
U_\theta H_{\beta,m}U_\theta^{-1}=\e^{-2\theta}H_{\e^\theta\beta, m}.
\end{equation}
The r.h.s.~of \eqref{anal} can be extended to an analytic function
\begin{equation}\label{dila}
\C\ni\theta\mapsto H_{\beta,m}(\theta):=\e^{-2\theta}H_{\e^\theta\beta, m}.
\end{equation}
As a consequence, the operator $H_{\beta,m}$ is an example of a dilation analytic Schr\"odinger operator,
where the domain of analyticity is the whole complex plane.
In addition, there is a periodicity in the imaginary direction: we have
\begin{equation*}
H_{\beta,m}(\theta)= H_{\beta,m}(\theta+2\i\pi).
\end{equation*}

An operator $H_{\beta,m}$ with a non-real $\beta$ can always be transformed
by dilation analyticity into an operator
with a real parameter. More precisely, if $\beta=\e^{\i\phi}|\beta|$ is any complex number, then we have
\begin{align*}
H_{\beta,m}(-\i\phi)&=\e^{2\i\phi}H_{|\beta|, m},\\
H_{\beta,m}(\i\pi-\i\phi)&=\e^{2\i\phi}H_{-|\beta|, m}.
\end{align*}
Note that these relations will be used in the Appendix for the explicit
description of the spectrum of $H_{\beta,m}$.

\subsection{Boundary value of the resolvent and spectral density}

Our next aim is to look at the boundary value of the resolvent of $H_{\beta,m}$ on the real axis.
For that purpose and for any $s\in\R$ we introduce the space $\langle X\rangle^{s} L^2(\R_+)$,
where $\langle X\rangle:= (1+X^2)^{1/2}$.
Clearly, for $s\geqslant 0$ the space $ \langle X\rangle^{-s}L^2(\R_+)$ is the domain of $\langle X\rangle^{s}$,
which we endow with the graph norm, while $\langle X\rangle^s L^2(\R_+)$ can be identified with the
anti-dual of $\langle X\rangle^{-s} L^2(\R_+)$.

Let $m,\beta\in \C$ with $\Re(m)>-1$. We set
\begin{equation}\label{exco}
\Omega_{\beta,m}^\pm:= \R_+\backslash
\Big\{\pm \tfrac{\i\beta}{2(N+\frac12+m)}\mid N\in\N,\ N+\frac12+m\neq0
\Big\}
\end{equation}
and $\Omega_{\beta,m}:=\Omega_{\beta,m}^+\cap\Omega_{\beta,m}^-$.
We say that $(\beta,m)$ is \emph{an exceptional pair}
if there exists $N\in\N$ such that $N+\frac12+m\neq0$ and $\frac{\i \beta}{N+m+\frac12}\in\R\backslash\{0\}$.

Note that
\begin{align*}
\Omega_{\beta,m}&= \R_+, \text{ if $(\beta,m)$ is not an exceptional pair},\\
\Omega_{\beta,m}&= \R_+\backslash\Big\{\Big|\frac{\Re(\beta)}{2\Im(m)}\Big|\Big\}, \text{ if $(\beta,m)$ is an exceptional pair and $\beta\not\in\i\R$},\\
\Omega_{\beta,m}&= \R_+\backslash\Big\{\frac{|\beta|}{2(N+\frac12+m)}\mid N+\frac12+m>0,\ N\in\N\Big\}, \text{ if $\beta\in\i\R$, $\beta\neq0$, $m\in]-1,\infty[$.}
\end{align*}

The theorem that we state below has some restrictions when $(\beta,m)$ is an exceptional pair.
Its statement is rather involved since any $\beta \in \C$ is considered. In the special case
$\Im(\beta)=0$ some simplifications take place.

\begin{theorem}\label{thm_conv_resolvent}
Let $m,\beta\in \C$ with $\Re(m)>-1$. Let us fix $k_0>0$
and consider any $k\in]k_0,\infty[\,\bigcap\Omega_{\beta,m}^\pm$.
Then, the boundary values of the resolvent
\begin{equation*}
R_{\beta,m}(k^2 \pm \i 0):= \mathop{\lim}\limits_{\eps \searrow 0}R_{\beta,m}(k^2 \pm\i\eps)
\end{equation*}
exist in the sense of operators from $\langle X\rangle^{-s} L^2(\R_+)$ to $ \langle X\rangle^s L^2(\R_+)$ for any $s>\frac{1}{2}+\frac{|\Im(\beta)|}{2k_0}$,
uniformly in $k$ on each compact subset of $]k_0,\infty[\,\bigcap\Omega_{\beta,m}^\pm$. For $x,y\in \R_+$ the kernel of ${R_{\beta,m}(k^2\pm \i 0)}$ is given by
\begin{align}\notag
&R_{\beta,m}(k^2\pm \i 0;x,y)\\
=&\pm\tfrac{\i}{2k}\Gamma\naw{\tfrac{1}{2} + m \mp \tfrac{\i\beta}{2k}}\begin{cases} \J{\frac{\beta}{2k}}{m}(2k x)\Hpm{\frac{\beta}{2k}}{m}(2k y) & \mbox{ for }0<x<y,\\
\J{\frac{\beta}{2k}}{m}(2k y)\Hpm{\frac{\beta}{2k}}{m}(2k x) & \mbox{ for }0<y<x.
\end{cases}\label{eq_kernel_1}
\end{align}
\end{theorem}

Before starting the proof, let us emphasize
 the role played by $\beta$. If $\Im(\beta)= 0$, then
the limiting absorption principle takes place in the usual spaces, with the exponent $s>\frac{1}{2}$.
On the other hand, if $\Im(\beta)\neq 0$, an additional weight is necessary for the limiting absorption principle.

\begin{proof}
Let $k>0$. Assume that $\pm \tfrac{\i\beta}{2k}-m-\tfrac{1}{2}\not \in \N$.
Thus, let us consider the operator $\langle X\rangle^{-s}R_{\beta,m}(k^2\pm \i \eps)\langle X\rangle^{-s}$ whose kernel is
\begin{equation}\label{eq_kernel_2}
\langle x\rangle^{-s}R_{\beta,m}(k^2\pm \i \eps;x,y)\langle y\rangle^{-s},
\end{equation}
see also \eqref{The-resolvent}.
We show that the corresponding operator is Hilbert-Schmidt and converges
in the Hilbert-Schmidt norm to the operator whose kernel is provided by \eqref{eq_kernel_1}.
For that purpose, let us also set $k^\mp_\eps:=\sqrt{-k^2\mp \i\eps}$ and observe that $\Re(k^\mp_\eps)>0$
and that $\lim\limits_{\eps\searrow 0}k^\mp_\eps=\mp \i k$.
As a consequence, one has $k^\mp_\eps = \mp \i k + O(\eps)$.

We consider first the slightly more complicated case when $-1<\Re(m)\leqslant 0$ and $m\neq 0$.
By the estimate \eqref{eq_est2} the expression \eqref{eq_kernel_2} is bounded for $\eps$ small enough by
\begin{align}\label{eq_to_maj}
\nonumber & C_k \tfrac{\big|\Gamma(\frac{1}{2}+m-\frac{\beta}{2k^\mp_\eps})\big|}{2|k|}
\min\{1,2x|k|\}^{\Re(m)+\frac{1}{2}} \langle x\rangle^{-s} \\
&\quad \times
\min\{1,2y|k|\}^{\Re(m)+\frac{1}{2}} \langle y\rangle^{-s} \begin{cases} \big(\frac{\langle y\rangle}{\langle x\rangle}\big)^{\Re(\frac{\beta}{2k_\eps^\mp})} & \mbox{ for }y>x,\\
\big(\frac{\langle x\rangle}{\langle y\rangle}\big)^{\Re(\frac{\beta}{2k_\eps^\mp})} & \mbox{ for }x>y.
\end{cases}
\end{align}
for a constant $C_k$ independent of $x$ and $y$ but which depends on $k$.
We then observe that
\begin{equation*}
\lim_{\eps\searrow0} \Re\big(\tfrac{\beta}{2k_\eps^\mp}\big)
= \mp \frac{\Im(\beta)}{2k}.
\end{equation*}
Since $\mp \frac{\Im(\beta)}{2k}\leqslant \frac{|\Im(\beta)|}{2k_0}$ and by taking $\eps<\eps_0$ sufficiently small,
our assumption on $s$ implies that the expression \eqref{eq_to_maj} belongs to $L^2(\R_+ \times \R_+)$.

On the other hand, starting with the expression \eqref{The-resolvent} and by taking the equalities \eqref{Jbm-definition} and
\eqref{Hpm-definition} into account, one also observes that for fixed $x$ and $y$, the expression given in \eqref{eq_kernel_2}
converges as $\eps\searrow 0$ to
\begin{equation}\label{eq_kernel_3}
\langle x\rangle^{-s}R_{\beta,m}(k^2\pm \i 0;x,y)\langle y\rangle^{-s},
\end{equation}
with $R_{\beta,m}(k^2\pm \i 0;x,y)$ defined in \eqref{eq_kernel_1}.
We can apply the Lebesgue Dominated Convergence Theorem and deduce that \eqref{eq_kernel_2}
converges in $L^2(\R_+ \times \R_+)$ to \eqref{eq_kernel_3}.
This convergence is equivalent to
\begin{equation*}
\lim_{\eps\searrow 0} \ \langle X\rangle^{-s}R_{\beta,m}(k^2\pm \i \eps)\langle X\rangle^{-s} =
\langle X\rangle^{-s}R_{\beta,m}(k^2\pm \i 0)\langle X\rangle^{-s}
\end{equation*}
in the Hilbert-Schmidt norm.
Note finally that the uniform convergence in $k$ on each compact subset of $]k_0,\infty[$
can be checked directly on the above expressions.

For $\Re(m) \geqslant 0$ with $m\neq 0$, the same proof holds with the estimate \eqref{eq_est1} instead of \eqref{eq_est2}.
Finally for $m=0$, the result can be obtained by using \eqref{eq_est3}.
\end{proof}

Based on the previous result, the corresponding \emph{spectral density} can now be defined.

\begin{proposition}
Let $m\in \C$ with $\Re(m)>-1$, let $\beta\in \C$ and let us fix $k_0>0$.
Then for $k\in]k_0,\infty[\,\bigcap\Omega_{\beta,m}$
the {\em spectral density} defined by
\begin{align*}
p_{\beta,m}(k^2) & := \mathop{\lim}\limits_{\eps\searrow 0}\frac{1}{2\pi \i}\naw{R_{\beta,m}(k^2+ \i \eps) - R_{\beta,m}(k^2 - \i \eps)} \\
&= \frac{1}{2\pi \i}\naw{R_{\beta,m}(k^2+ \i 0) - R_{\beta,m}(k^2- \i 0)}
\end{align*}
exists in the sense of operators from $\langle X\rangle^{-s} L^2(\R_+)$ to $\langle X\rangle^s L^2(\R_+)$ for any $s>\frac{1}{2}+\frac{|\Im(\beta)|}{2k_0}$.
The kernel of this operator is provided for $x,y\in\bbR_+$ by
\begin{equation}\label{spectral-density}
p_{\beta,m}(k^2;x,y) =\frac{\e^{\frac{\pi\beta}{2k}}}{4\pi k}\Gamma\Big(\tfrac{1}{2}+m+\tfrac{\i\beta}{2k}\Big)\Gamma\Big(\tfrac{1}{2}+m-\tfrac{\i\beta}{2k}\Big)
\J{\frac{\beta}{2k}}{m}(2k x)\J{\frac{\beta}{2k}}{m}(2k y).
\end{equation}
\end{proposition}

\begin{proof}
The existence of the limit is provided by Theorem \ref{thm_conv_resolvent}
while the explicit formula \eqref{spectral-density} can be deduced from \eqref{eq_kernel_1}
together with \eqref{eq_missing}.
\end{proof}

Note that for $\beta=0$ the expression obtained above reduces to
\begin{equation}
p_{0,m}(k^2;x,y) = \frac{1}{\pi k}\Ja_m(k x)\Ja_m(k y),
\end{equation}
by taking the relation \eqref{eq_J0m_B} into account.
This expression corresponds to the one obtained in a less general context in \cite[Prop.~4.4]{DR}.
	
\subsection{Reminder about the Hankel transform}

As we mentioned before, to some extent this paper can be viewed as a continuation
of \cite{BDG,DR}.
These two papers were devoted to Schr\"odinger operators with the inverse square potential.
Among other things, certain natural transformations diagonalizing these operators were introduced.
They were called there {\em (generalized) Hankel transformations}.

In the present paper we would like to find natural transformations that diagonalize $H_{\beta,m}$. We will
mimic as closely as possible our previous constructions.
Therefore, we devote this subsection to a summary of selected results of \cite{BDG,DR}.

Recall first that for any $m\in\C$ with $\Re(m)>-1$ the operator $H_m=H_{0,m}$ can be diagonalized
using the {\em Hankel transformation} $\sF_m=\sF_m^\#=\sF^{-1}_m$, which is a bounded operator on $L^2(\R_+)$ such that
\begin{align*}
\sF_m X^2&=H_m\sF_m.
\end{align*}
For any $m,m'$ we also define the M{\o}ller operators corresponding to the pair $\big(H_m,H_{m'}\big)$ by the time dependent formula
\begin{equation*}
W_{m,m'}^\pm:=\slim_{t\to\infty}\e^{\pm\i tH_m}\e^{\mp\i tH_{m'}}.
\end{equation*}
These operators can be expressed in terms of the Hankel transformation
\begin{equation*}
W_{m,m'}^\pm=\e^{\pm\i\frac{\pi}{2}(m-m')}\sF_m\sF_{m'}.
\end{equation*}
It is also natural to introduce a pair of operators
\begin{equation}\label{emme}
\sF_m^\pm:=\e^{\pm\i\frac{\pi}{2}m}\sF_m,
\end{equation}
which can be called the {\em incoming/outgoing Hankel transformation}.
Then we can write
\begin{equation*}
W_{m,m'}^\pm=\sF_{m}^\pm \sF_{m'}^{\mp \#}\ .
\end{equation*}
Note that the definition \eqref{emme} may look trivial, but we will see that in some situations it is more natural to generalize $\sF_m^\pm$ rather than $\sF_m$.

The operators $H_m$ have very special properties. Therefore, some of the properties of Hankel transformations are specific to this class of operators.
A more general class of 1-dimensional Schr\"odinger operators $H_{m,\kappa}$ on the half-line has been considered in \cite{DR}.
They are generalizations of $H_m$ by considering general boundary condition at zero.
Exceptional cases exist for this family, however for non-exceptional pairs $(m,\kappa)$ one can generalize the construction of the incoming/outgoing Hankel transformations. In fact, one can define a pair of bounded and left-invertible operators $\sF_{m,\kappa}^\pm$ that diagonalize $H_{m,\kappa}$. They satisfy
\begin{align*}
\sF_{m,\kappa}^{\mp\#}\sF_{m,\kappa}^\pm &=\one,\\
\sF_{m,\kappa}^\pm \sF_{m,\kappa}^{\mp\#}&= \one_{\R_+}(H_{m,\kappa}),\\
\sF_{m,\kappa}^\pm X^2\sF_{m,\kappa}^{\mp \#}&=H_{m,\kappa}\;\! \one_{\R_+}(H_{m,\kappa}),
\end{align*}
where $\one_{\R_+}(H_{m,\kappa})$ in the self-adjoint
case is the projection on the continuous subspace of $H_{m,\kappa}$, and in the general case it is its obvious generalization.
We called $\sF_{m,\kappa}^\pm $ the \emph{outgoing/incoming generalized Hankel transformations}.
The operators $\sF_{m,\kappa}^+$ and $\sF_{m,\kappa}^{-}$ are linked by the relation
\begin{equation}\sF_{m,\kappa}^+\cG_{m,\kappa}=\sF_{m,\kappa}^{-},\label{linked}
\end{equation}
where $\cG_{m,\kappa}$ is a bounded and boundedly invertible operator commuting with $X$.

One can formulate scattering theory for an arbitrary pair of Hamiltonians $H_{m,\kappa}$,
$H_{m',\kappa'}$, as it is done in \cite{DR}. Alternatively, one can fix a \emph{reference Hamiltonian},
which is simpler, to which the more complicated \emph{interacting Hamiltonian}
$H_{\beta,m}$ will be compared. Two choices of the reference Hamiltonian
can be viewed as equally simple: the {\em Dirichlet Laplacian}
$H_\D:=H_{\frac12}$ and the {\em Neumann Laplacian} $H_\mathrm{N}:=H_{-\frac12}$.

Recall from \cite[Sec.~4.7]{DR} that the Hankel transformation for $H_\D$ is the sine transformation
$\sF_{\frac12}=\sF_\D=\sF_\D^{-1}=\sF_\D^\#=\sF_\D^*$,
while the Hankel transformation for $H_\mathrm{N}$ is the cosine transformation
$\sF_{-\frac12}=\sF_\mathrm{N}=\sF_\mathrm{N}^{-1}=\sF_\mathrm{N}^\#=\sF_\mathrm{N}^*$.
These transformations are defined by
\begin{align*}
(\sF_\D f)(x)&=\sqrt{\frac{2}{\pi}}\int\sin(xk)f(k)\;\!\d k,\\
(\sF_\mathrm{N} f)(x)&=\sqrt{\frac{2}{\pi}}\int\cos(xk)f(k)\;\!\d k.
\end{align*}
Following \eqref{emme} and \eqref{linked}, we also introduce
\begin{align*}
\sF_\D^\pm:=\e^{\pm\i\frac{\pi}{4}}\sF_\D,&\quad\cG_\D:=\e^{\i\frac{\pi}{2}}\one,\\
\sF_\mathrm{N}^\pm:=\e^{\mp\i\frac{\pi}{4}}\sF_\mathrm{N},&
\quad\cG_\mathrm{N}:=\e^{-\i\frac{\pi}{2}}\one.
\end{align*}

For any non-exceptional $\beta,m$ with $m>-1$ one can introduce its M{\o}ller operators with respect to the Dirichlet and Neumann dynamics:
\begin{align*}
W_{m,\kappa,\D}^\pm&:=\slim_{t\to\infty}\e^{\pm\i tH_{m,\kappa}}\e^{\mp\i tH_\D},\\
W_{m,\kappa,\mathrm{N}}^\pm&:=\slim_{t\to\infty}\e^{\pm\i tH_{m,\kappa}}\e^{\mp\i tH_\mathrm{N}}.
\end{align*}
The corresponding scattering operators are then defined by
\begin{equation*}
S_{m,\kappa;\D} = W_{m,\kappa;\D}^{-\#} W_{m,\kappa;\D}^- \ \hbox{ and }\
S_{m,\kappa;\mathrm{N}} := W_{m,\kappa;\mathrm{N}}^{-\#} W_{m,\kappa;\mathrm{N}}^-.
\end{equation*}
We can also express the M{\o}ller operators in terms of the generalized Hankel transformations
\begin{equation*}
W_{m,\kappa;\D}^\pm =\sF_{m,\kappa}^\pm\sF_\D^{\mp\#} \hbox{ and }\
W_{m,\kappa;\mathrm{N}}^\pm =\sF_{m,\kappa}^\pm\sF_\mathrm{N}^{\mp\#}.
\end{equation*}
By conjugating with $\sF_\D$, resp. $\sF_\mathrm{N}$, the scattering operator
can be brought to a diagonal
form, where up to an inessential factor it coincides with $\cG_{m,\kappa}$\;\!:
\begin{equation*}
\sF_\D S_{m,\kappa;\D}\sF_\D = \i\cG_{m,\kappa}\ \hbox{ and } \
\sF_\mathrm{N}S_{m,\kappa;\mathrm{N}}\sF_\mathrm{N} = -\i\cG_{m,\kappa}.
\end{equation*}

\subsection{Hankel-Whittaker transformation}

It is natural to ask whether the operators $H_{\beta,m}$
considered in this paper also possess diagonalizing operators
and a satisfactory scattering theory.
There exists actually a candidate for a generalization of incoming/outgoing Hankel transformations $\sF_m^\pm$.
For any $\beta,m\in \C$ with $\Re(m)>-1$ let us define the kernel
\begin{equation}\label{eq_def_F}
\sF^\pm_{\beta,m}(x,k) := \frac {1}{\sqrt{2\pi}} \e^{\pm \i \frac{\pi}{2}m}\e^{\frac{\pi\beta}{4k}}
\Gamma\big(\tfrac{1}{2}+m\pm\tfrac{\i\beta}{2k}\big)\J{\frac{\beta}{2k}}{m}(2xk),
\end{equation}
where $x,k\in \R_+$.
This kernel can be used to define a linear transformation on any $f\in C_\mathrm{c}(\R_+)$\;\!:
\begin{equation*}
\big(\sF_{\beta,m}^\pm f\big)(x):=\int_0^\infty\sF_{\beta,m}^\pm(x,k)f(k)\;\!\d k.
\end{equation*}
We call $\sF_{\beta,m}^\pm$ the \emph{outgoing/incoming Hankel--Whittaker transformation}.

We also introduce the function $g_{\beta,m}:\R_+\to \C$ by
\begin{equation}\label{eq_def_G}
g_{\beta,m}(k):=\e^{-\i\pi m}\frac{\Gamma(\frac12+m-\i\frac{\beta}{2k})}
{\Gamma(\frac12+m+\i\frac{\beta}{2k})}
\end{equation}
and the corresponding multiplication operator
\begin{equation*}
\cG_{\beta,m}:=g_{\beta,m}(X),
\end{equation*}
which we can call the \emph{intrinsic scattering operator}.

Let us collect the most obvious properties of
$\sF_{\beta,m}^\pm$ and $\cG_{\beta,m}$.
Recall that the set $\Omega_{\beta,m}$ has been introduced just after \eqref{exco},
and that if $(\beta,m)$ is not an exceptional pair then $\Omega_{\beta,m}=\R_+$.

\begin{theorem}
Let $m,\beta\in \C$ with $\Re(m)>-1$, let us also fix $k_0>0$ and let $s>\frac{1}{2}+\frac{|\Im(\beta)|}{2k_0}$.
\begin{enumerate}
\item[(i)] $\sF_{\beta,m}^\pm$ maps $C_{\rm c}\big(]k_0,\infty[\,\bigcap\Omega_{\beta,m}\big)$ into $ \langle X\rangle^{s}L^2(\R_+)$.
\item[(ii)] If $h\in C_{\rm c}\big(]k_0,\infty[\,\bigcap\Omega_{\beta,m}\big)$, then in the sense of quadratic forms on
$\langle X\rangle^{-s}L^2(\R_+)$ we have
\begin{equation}\label{diago}
\int_0^\infty h(k^2)\;\! p_{\beta,m}(k^2)\;\!\d k^2 =\sF_{\beta,m}^\mp h(X^2)\sF_{\beta,m}^{\pm\#}.
\end{equation}
\item[(iii)] The equalities $\sF_{\beta,m}^+\cG_{\beta,m}=\sF_{\beta,m}^-$ and $\sF_{\beta,m}^-\cG_{\beta,m}^{-1}=\sF_{\beta,m}^+$ hold.
\item[(iv)]
For fixed $k\in \Omega_{\beta,m}$ the following asymptotics hold as $x\to\infty$\;\!:
\begin{align}
\nonumber \sF_{\beta,m}^+(x,k) = & \frac{1}{\sqrt{2\pi}}\Big(\e^{-\i\frac{\pi}{4}}\e^{\i kx}(2kx)^{\i\frac{\beta}{2k}}\big(1+O(x^{-1})\big) \\
\label{mimi1} & \qquad +g_{\beta,m}^{-1}(k)\e^{\i\frac{\pi}{4}}\e^{-\i kx}(2kx)^{-\i\frac{\beta}{2k}}\big(1+O(x^{-1})\big)\Big),\\
\nonumber \sF_{\beta,m}^-(x,k) = & \frac{1}{\sqrt{2\pi}}\Big(\e^{\i\frac{\pi}{4}}\e^{-\i kx}(2kx)^{-\i\frac{\beta}{2k}}\big(1+O(x^{-1})\big) \\
\label{mimi2} & \qquad +g_{\beta,m}(k)\e^{-\i\frac{\pi}{4}}\e^{\i kx}(2kx)^{\i\frac{\beta}{2k}}\big(1+O(x^{-1})\big)\Big).
\end{align}
\end{enumerate}
\end{theorem}

\begin{proof}
The proof of (i) reduces to showing that the map $x\mapsto \langle x\rangle^{-s} \sup_{k\in K} \J{\frac{\beta}{2k}}{m}(2kx)$ is in $L^2(\R_+)$
for any $k$ in a compact set $K\subset ]k_0,\infty[\,\bigcap\Omega_{\beta,m}$. The $L^2$-integrability near $0$ follows from \eqref{Jbm-around-zero}, while
the $L^2$-integrability near infinity follows from \eqref{Jbm-around-infinity}.
Note that the factor $x^{\pm \i\beta}$, which becomes $(2kx)^{\pm \i \frac{\beta}{2k}}$ after the required change of variables, imposes the
dependence on $k_0$ for the lower limit of the index $s$.
The proofs of (ii) and (iii) consist in direct computations.
Finally, (iv) can be obtained by taking again into account the asymptotic expansion of $\J{\beta}{m}$ provided in \eqref{Jbm-around-infinity}.
\end{proof}

Note that \eqref{diago} essentially says that $\sF_{\beta,m}^\pm$ diagonalize the continuous part of
$H_{\beta,m}$, since the l.h.s.~of \eqref{diago} can be interpreted as $h(H_{\beta,m})$.
In the self-adjoint case, this would correspond to the absolutely continuous part of $H_{\beta,m}$.
Clearly, this condition does not fix $\sF_{\beta,m}^\pm$ completely.
The additional condition for our choice of $\sF_{\beta,m}^\pm$ comes from scattering theory, which is expressed in the asymptotics
\eqref{mimi1} and \eqref{mimi2}.
In that framework the functions $x\mapsto \sF_{\beta,m}^\pm(x,k)$ can be viewed as outgoing/incoming distorted waves (or generalized eigenfunctions)
of $H_{\beta,m}$ associated with the eigenvalue $k^2$. Note that if we set $\beta=0$, then
\eqref{mimi1} and \eqref{mimi2} have the form of usual distorted waves in the short-range case.
On the other hand, the factors $(kx)^{\i\beta}$ are needed because of the long-range part of the potential,
while the factors $\e^{\pm\i\frac{\pi}{4}}$ are related to the {\em Maslov index} and are needed to make our definitions consistent
with the case $\beta=0$ described in \cite{DR}.

Let us now recall from \cite{BDG,DR} that $\sF_{0,m}^\pm$ are unitary for real $m$, and are bounded for more general $m$.
It is natural to ask about the boundedness of $\sF_{\beta,m}^\pm$ in the general framework introduced here,
but they seem to be rather ill-behaved operators.
Note that the operators $\cG_{\beta,m}$ are better behaved, and their behavior is easier to study:

\begin{proposition}
\begin{enumerate}
\item[(i)] If $m,\beta$ are real, then $\cG_{\beta,m}$ is unitary.
\item[(ii)] If $\beta$ is real and $\Re(m)\neq-\frac12$, then $\cG_{\beta,m}$ is bounded and boundedly invertible.
\item[(iii)] In all other cases $\cG_{\beta,m}$ is either unbounded or has an unbounded inverse.
\end{enumerate}
\end{proposition}

\begin{proof}
For (i), it is sufficient to recall that $\Gamma(\overline{z})=\overline{\Gamma(z)}$.
For (ii), by assuming that $\Re(m)\neq -\frac{1}{2}$ we make sure that neither the numerator nor the denominator of $g_{\beta,m}$ go through the value $\Gamma(0)$.
In addition, by using Stirling's formula one observes that $g_{\beta,m}(k)$ remains bounded for $k\to 0$ and for $k\to \infty$.
Finally, in the case (iii) either the numerator or the denominator of $g_{\beta,m}$ can have local singularities, and in addition
either $g_{\beta,m}(k)$ or $g_{\beta,m}^{-1}(k)$ are unbounded for $k\to 0$.
This last result is again a consequence of Stirling's formula.
\end{proof}

We conjecture that $\sF_{\beta,m}^\pm$ is unbounded in $L^2(\R_+)$ for all non-real $\beta$.
If $\Im(\beta)=0$ but $\Im(m)\neq0$, we do not know.
For real $\beta$ and $m$, which correspond to self-adjoint $H_{\beta,m}$,
the transformations $\sF_{\beta,m}^\pm$ are bounded, as we discuss in the next subsections.

\subsection{Hankel-Whittaker transformation for real parameters}

Throughout this and the next subsection we assume that $\beta,m\in\R$ with $m>-1$.
The operators $H_{\beta,m}$ are then self-adjoint and their spectral and scattering theory is well understood.

In the real case, the Hankel-Whittaker transformation satisfies
\begin{equation*}
\sF_{\beta,m}^{\pm*}=\sF_{\beta,m}^{\mp\#}.
\end{equation*}
Because of this identity, we can avoid using the Hermitian conjugation in our formulas in favor of transposition.
We do this because we would like that our formulas are easy to generalize to
the non-self-adjoint case, where so far their meaning is to a large extent unclear.

\begin{theorem}\label{rewrite}
$\sF_{\beta,m}^{\pm}$ are isometries that diagonalize
$H_{\beta,m}$ on the range of the spectral projection of
$H_{\beta,m}$ onto $\R_+$\;\!:
\begin{align}
\sF_{\beta,m}^{\mp \#} \sF_{\beta,m}^{\pm} &= \one,\label{rel1}\\
\sF_{\beta,m}^\pm \sF_{\beta,m}^{\mp \#}& = \one_{\R_+}(H_{\beta,m}),\label{rel2}\\
\sF_{\beta,m}^\pm X^2&= H_{\beta,m}\sF_{\beta,m}^\pm.\label{rel3}
\end{align}
\end{theorem}

\begin{proof}
For any $0<a<b$, we can apply Stone's formula
\begin{equation}
\one_{]a,b[}(H_{\beta,m})= \slim_{\epsilon\searrow0}\frac{1}{2\pi\i}\int_a^b(R_{\beta,m}(\lambda+\i\epsilon)- R_{\beta,m}(\lambda-\i\epsilon)\big)\;\!\d \lambda.
\label{stone1}\end{equation}
We can reinterpret \eqref{stone1}
in the sense of a quadratic form on appropriate weighted spaces, writing
\begin{equation}\label{stone2}
\one_{]a,b[}(H_{\beta,m})= \int_a^b p_{\beta,m}(k^2)\;\!\d k^2.
\end{equation}
Now \eqref{rel2} and \eqref{rel3} follow from \eqref{stone2} and from the identity
\begin{equation}
p_{\beta,m}(k^2;x,y)=
\frac{1}{2k} \sF_{\beta,m}^-(x,k)\sF_{\beta,m}^+(y,k).
\label{stone5}\end{equation}

It remains to prove \eqref{rel1}. To simplify notation, we will write $\sF$ for $\sF_{\beta,m}^\pm$ and $H$ for $H_{\beta,m}$.
By \eqref{stone2} and \eqref{stone5}, for any interval $I$ we have
\begin{equation}\label{squa}
\sF\one_I(X^2) \sF^*=\one_I(H).
\end{equation}
By squaring \eqref{squa} we obtain the equality
\begin{align}\label{squa1}
\sF\one_I(X^2) \sF^*\sF\one_I(X^2) \sF^*
&=\one_I(H).
\end{align}
By setting then $P:= \sF^*\sF$ and by comparing \eqref{squa} and \eqref{squa1} we infer that
\begin{equation}\label{squa2}
P\one_I(X^2)P\one_I(X^2)P=P\one_I(X^2)P.
\end{equation}
Clearly, the r.h.s.~of \eqref{squa2} is equal to $P\one_I(X^2)^2P$, and therefore
$$
P\one_I(X^2)(\one-P)\one_I(X^2)P=0.
$$
Equivalently one has
\begin{equation*}
\big((\one-P)\one_I(X^2)P\big)^*(\one-P)\one_I(X^2)P=0.
\end{equation*}
Consequently one gets $(\one-P)\one_I(X^2)P=0$ and $P\one_I(X^2)(\one-P)=0$.
By subtraction we finally obtain
\begin{equation*}
P\one_I(X^2)=\one_I(X^2)P.
\end{equation*}
Thus $P$ is a projection commuting with all spectral projections of $X^2$ onto intervals. But $X^2$ has multiplicity $1$.
Therefore, there exists a Borel set $\Xi\subset\R_+$ such that
\begin{equation*}
P=\one_\Xi(X^2)=\sF^* \sF.
\end{equation*}

Suppose that $\R_+\backslash\Xi$ has a positive measure. Then we can find $k_0\in\R_+$ such that for any $\epsilon>0$
$$
I_\epsilon:=[k_0-\epsilon,k_0+\epsilon]\setminus\Xi
$$
has also a positive measure. Let $f_\epsilon$ be the characteristic function of $I_\epsilon$, understood as an element of $L^2(\R_+)$. Then one infers that
\begin{equation}\label{kerno}
\|f_\epsilon\|\neq0 \quad \hbox{ and }\quad \sF f_\epsilon=0.
\end{equation}

From the explicit formula for $\sF(x,k)$ we immediately see that for any $k_0\in\R_+$ we can find $x_0\in\R_+$ such that
$\sF(x_0,k_0)\neq0$. We also know that $\sF(x,k)$ is continuous in both variables. Therefore, we can find $\epsilon>0$ such that
for $x\in[x_0-\epsilon,x_0+\epsilon]$ and $k\in[k_0-\epsilon,k_0+\epsilon]$ we have
\begin{equation*}
|\sF(x,k)-\sF(x_0,k_0)|>\frac12|\sF(x_0,k_0)|.
\end{equation*}

Now one has
\begin{equation*}
\big(\sF f_\epsilon\big)(x)=\int_{I_\epsilon}\sF(x,k)\d k,
\end{equation*}
and therefore,
\begin{equation*}
\big|\big(\sF f_\epsilon\big)(x)\big|\geqslant \frac12|I_\epsilon|\,|\sF(x_0,k_0)|,\quad
x\in[x_0-\epsilon,x_0+\epsilon],
\end{equation*}
where $|I_\epsilon|$ denotes the Lebesgue measure of $I_\epsilon$.
Hence $ \sF f_\epsilon\neq0$, which is a contradiction with \eqref{kerno}.
\end{proof}

In the case of real $m,\beta$, we have $|g_{\beta,m}(k)|=1$ for any $k\in \R_+$.
Therefore, the whole information about $g_{\beta,m}(k)$ is contained in its argument.
One half of the argument of
$g_{\beta,m}(k)$ is called the {\em phase shift}
\begin{equation*}
\delta_{\beta,m}(k):=-\frac{\pi}{2}m+\frac{1}{\i}\log\Big(\Gamma
\big(\tfrac{1}{2}+m-\tfrac{\i\beta}{2k}\big)\Big).
\end{equation*}
We have the relations
\begin{align*}
g_{\beta,m}(k)&=\e^{\i 2\delta_{\beta,m}(k)},\\
\cG_{\beta,m}&=\e^{\i 2\delta_{\beta,m}(X)}.
\end{align*}

In the real case, one can avoid using the incoming/outgoing Hankel-Whittaker
transformations $\sF_{\beta,m}^{\pm}$, and instead introduce a single $\sF_{\beta,m}$ given by the kernel
\begin{equation*}
\sF_{\beta,m}(x,k) := \frac {1}{\sqrt{2\pi}} \e^{\frac{\pi\beta}{4k}}
\Big|\Gamma\big(\tfrac{1}{2}+m\pm\tfrac{\i\beta}{2k}\big)\Big|\J{\frac{\beta}{2k}}{m}(2xk).
\end{equation*}
Note that
\begin{equation*}
\sF_{\beta,m}^\pm= \sF_{\beta,m}\e^{\mp \i \delta_{\beta,m}(X)}.
\end{equation*}

We can rewrite Theorem \ref{rewrite} in terms of the operator
$\sF_{\beta,m}$\;\!:
\begin{align*}
\sF_{\beta,m}^{ \#} \sF_{\beta,m} &= \one,\\
\sF_{\beta,m} \sF_{\beta,m}^{\#}& = \one_{\R_+}(H_{\beta,m}),\\
\sF_{\beta,m} X^2&= H_{\beta,m}\sF_{\beta,m}.
\end{align*}
However, with $\sF_{\beta,m}$ one loses the analyticity, therefore we prefer to continue using $\sF_{\beta,m}^{\pm}$.

\begin{remark}
In the setting of the Coulomb problem, when $m+\frac{1}{2}=\ell \in \N$ and $\beta \in \R$,
the expression
$$
\delta_\ell(k):=\arg\Big(\Gamma\big(\ell+1-\i\frac{\beta}{2k}\big)\Big)
$$
is called \emph{the Coulomb phase shift}.
Note that an expression close to \eqref{eq_def_G} was introduced in \cite[Eq.~(2.3a)]{TB}.
In the setting of the Coulomb potential in $d=3$, an additional function called \emph{the Gamow factor} is often introduced, see
for example \cite[Eq.~14.1.7]{AS}. In our framework this factor does not seem to play an important role.
\end{remark}

\subsection{Scattering theory for real parameters}

Since the Coulomb potential is long-range, we do not have the standard short-range scattering theory between arbitrary $H_{\beta,m}$ and $H_{\beta',m'}$.
However, if we fix $\beta$ then the scattering theory between $H_{\beta,m}$ and $H_{\beta,m'}$ is short-range.
One can argue that for $\beta\neq0$ there is only one natural reference Hamiltonian,
namely $H_{\beta,\mathrm{S}}:=H_{\beta,-\frac12}= H_{\beta,\frac12}$. Here we use se subscript $S$ for \emph{standard}.
The situation of two equally justified reference Hamiltonians $H_{0,-\frac12}=H_\mathrm{N}$ and $H_{0,\frac12}=H_\D$ seems to be specific for $\beta=0$.

By the standard methods of time-dependent long-range scattering theory, as described for example in \cite{DG0},
we can show the existence for any real $\beta,m$ with $m>-1$ of the M{\o}ller operators
\begin{equation*}
W_{\beta,m;\beta,\mathrm{S}}^\pm :=\slim_{t\to\infty}\e^{\pm\i tH_{\beta,m}}\e^{\mp\i tH_{\beta,\mathrm{S}}} \one_{\R_+}(H_{\beta,\mathrm{S}}).
\end{equation*}
These operators can also be expressed in terms of the Hankel-Whittaker transform:
\begin{equation*}
W_{\beta,m;\beta,\mathrm{S}}^\pm =\sF_{\beta,m}^\pm\sF_{\beta,\mathrm{S}}^{\mp\#}
=\sF_{\beta,m}\e^{\mp\i (\delta_{\beta,m}(X)-\delta_{\beta,\mathrm{S}}(X))} \sF_{\beta,\mathrm{S}}^{\#}.
\end{equation*}

In order to compare distinct $\beta$ and $\beta'$ we need to use modified wave operators.
There exist various constructions, and we select a construction involving a time-independent modifier
that is similar in some sense to the celebrated Isozaki-Kitada construction.
As the reference Hamiltonian we use $H_\D$. The modifier is chosen to be $\sF_{\beta,\mathrm{S}}^\pm\sF_\D^{\mp \#}$.
Note that the modifier does not depend on $m$.
It depends on $\pm$, and as mentioned above this allows us to obtain expressions analytic in the parameters.
With the results obtained so far, one can easily prove the following statement:

\begin{theorem}
If $\beta,m\in\R$, $m>-1$, then there exist
\begin{equation*}
W_{\beta,m;\D}^\pm :=\slim_{t\to\infty}\e^{\pm\i tH_{\beta,m}}
\sF_{\beta,\mathrm{S}}^\pm\sF_\D^{\mp\#}\e^{\mp\i tH_{\D}},
\end{equation*}
and the following equalities hold:
\begin{equation}\label{wave1}
W_{\beta,m;\D}^\pm
=\sF_{\beta,m}^\pm\sF_{\D}^{\mp\#}
=\sF_{\beta,m}\e^{\mp \i (
\delta_{\beta,m}(X)-\delta_{\D}(X))}
\sF_{\D}^{\#}.
\end{equation}
In addition, the scattering operator
$S_{\beta,m;\D}:= W_{\beta,m;\D}^{-\#} W_{\beta,m;\D}^-$
satisfies
\begin{equation*}
\sF_\D S_{\beta,m;\D}\sF_\D = \i\cG_{\beta,m}.
\end{equation*}
\end{theorem}

Let us now compare what we obtained with the literature.
Recall from \eqref{eq_def_F} that the kernel of $\sF_{\beta,m}^{\pm}(x,k)$ given by
\begin{align}
\nonumber &\sF^\pm_{\beta,m}(x,k) \\\label{eq_f2} & = \frac {1}{\sqrt{2\pi}} \e^{\pm \i \frac{\pi}{2}m}\e^{\frac{\pi\beta}{4k}}
\Gamma\big(\tfrac{1}{2}+m\pm\tfrac{\i\beta}{2k}\big)\J{\frac{\beta}{2k}}{m}(2xk) \\
\notag& = \frac {1}{\sqrt{2\pi}} \e^{\pm \i \frac{\pi}{2}m}\e^{\frac{\pi\beta}{4k}}
\Gamma\big(\tfrac{1}{2}+m\pm\tfrac{\i\beta}{2k}\big) (2kx)^{\frac{1}{2}+m} \e^{-\i kx}
{}_1\mathbf{F}_1\Big(\frac{1}{2}+m+\i \frac{\beta}{2k};\,1+2m;\,2\i kx\Big).
\end{align}
For the scattering operators, we obtain the multiplication operator
\begin{equation}\label{eq_f3}
\big(\sF_\D S_{\beta,m;\D}\sF_\D\big)(k)
= \e^{\i\pi(\frac{1}{2}-m)}\frac{\Gamma(\frac12+m-\i\frac{\beta}{2k})}{\Gamma(\frac12+m+\i\frac{\beta}{2k})}.
\end{equation}

Without any surprise, the expressions obtained in \eqref{eq_f2} and \eqref{eq_f3} coincide with the ones available in the literature,
as for example in \cite{GPT,MOC,Mic,Muk,Yaf}. Note that in these references, only the cases $m\geqslant 0$ are considered,
and most of the time only the case $m=\ell+\frac{1}{2}$ with $\ell \in \N$.

Let us conclude by one feature of the scattering theory for Whittaker operators that is worth pointing out.
The common wisdom says that for long-range potentials modified M{\o}ller operators, and hence also modified scattering operator,
are not canonically defined, namely, they are defined only up to an arbitrary momentum dependent phase factor.
However, in the case of Whittaker operators there exists a choice that can be viewed as canonical, namely the one
provided in \eqref{wave1}.

\appendix

\section{Pictures of spectrum}

In the appendix we provide a few pictures of the spectrum of the operators $H_{\beta,m}$.
We concentrate on the case of non-zero $\beta$, since for $\beta=0$ the spectrum is simply $[0,\infty[$.
By a scaling argument, it is enough to consider $|\beta|=1$.
We restrict ourselves to $m$ with a non-zero imaginary part, since this situation corresponds to the most interesting spectrum.

First we present a few examples of the spectrum for a fixed $\Im(m)$ and a fixed non-zero real $\beta$.
We set $m_\i:=\Im(m)=-2.4$ and $\beta=1$, and we select 8 values of the real part of $m$, namely
$$
m_\real:=\Re(m)=-0.75,\ -0.5,\ -0.25,\ 0,\ 0.25,\ 0.5,\ 1,\ 2.
$$
The spectrum is marked with the blue color.

In all these pictures the point spectrum is located on the same trajectory,
\begin{equation*}
\left\{-\frac{|\beta|^2}{4\big(t+\i \Im(m)\big)^2}\ |\ t\in\R\cup\{\infty\}\right\},
\end{equation*}
which depends only on $\beta$ and $\Im( m)$. This trajectory is marked with a thin gray line.

The point spectrum is a sequence of points on the lower half-plane converging to $0$, moving clockwise as $\Re(m)$ increases.
For $\Re(m)\in]-1,\frac12]$, one of the points of the sequence is \emph{hidden on the non-physical sheet of the complex plane}
and is not an eigenvalue. It is marked in red and called a \emph{resonance}.
When $\Re(m)$ crosses $-\frac12$, it appears on the \emph{physical sheet} and becomes an eigenvalue.

Note that for any $-1<m_\real\leqslant -\frac12$, we have the identity
\begin{equation}
\sigma(H_{1,m_\real+\i m_\i})=
\sigma(H_{1,m_\real+1+\i m_\i}).
\end{equation}
Therefore, Fig.~1 and 2 have the same spectrum as Fig.~5 and 6. However on Fig.~1 and 2 we have in addition a resonance.

\begin{figure}[H]
\begin{subfigure}{0.5\textwidth}
\includegraphics[scale=0.7]{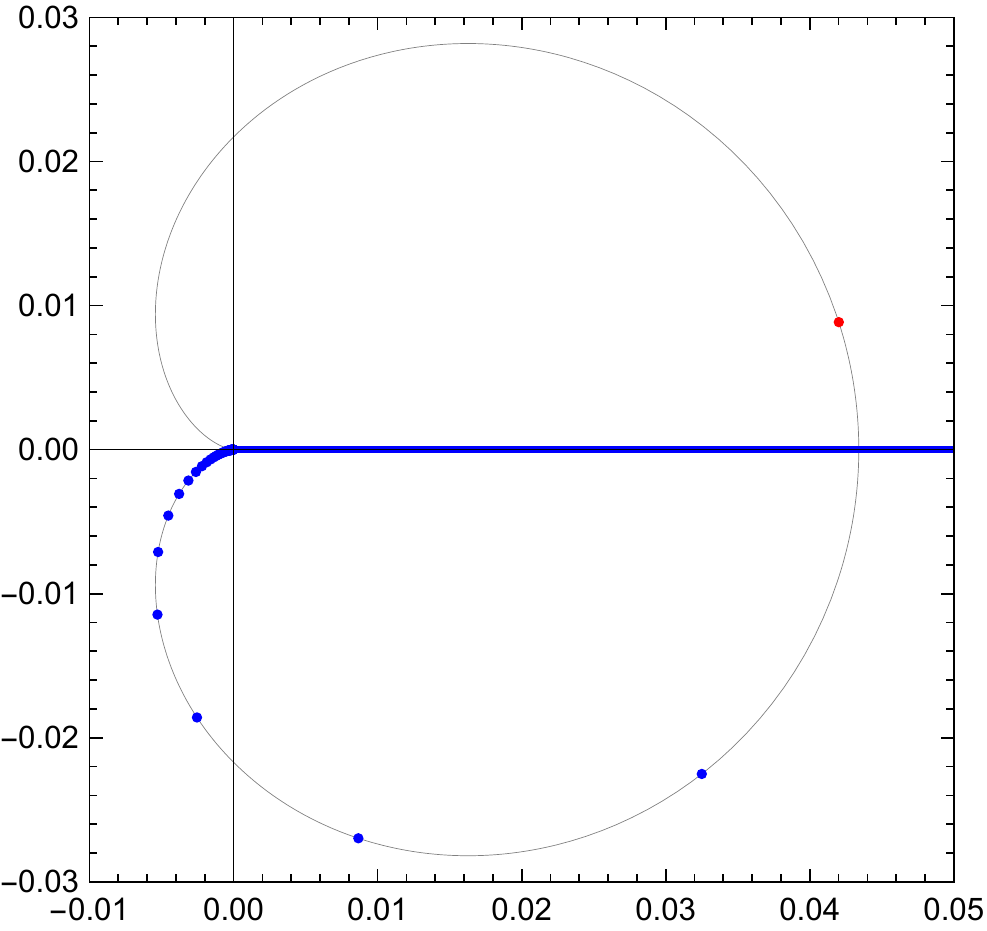}
\caption*{Fig.~1. $\sigma(H_{1,-0.75-2.4\i})$}
\end{subfigure}
\begin{subfigure}{0.5\textwidth}
\includegraphics[scale=0.7]{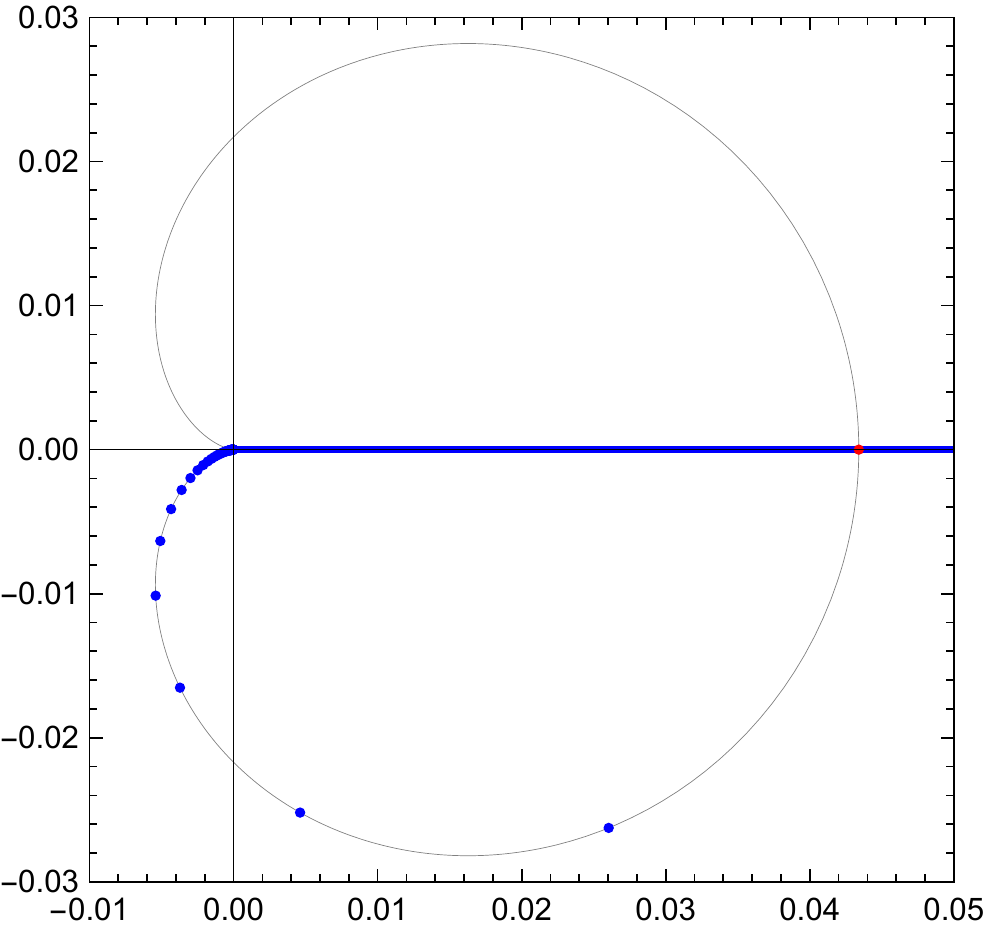}
\caption*{Fig.~2. $\sigma(H_{1,-0.5-2.4\i})$}
\end{subfigure}
\end{figure}

\begin{figure}[H]
\begin{subfigure}{0.5\textwidth}
\includegraphics[scale=0.7]{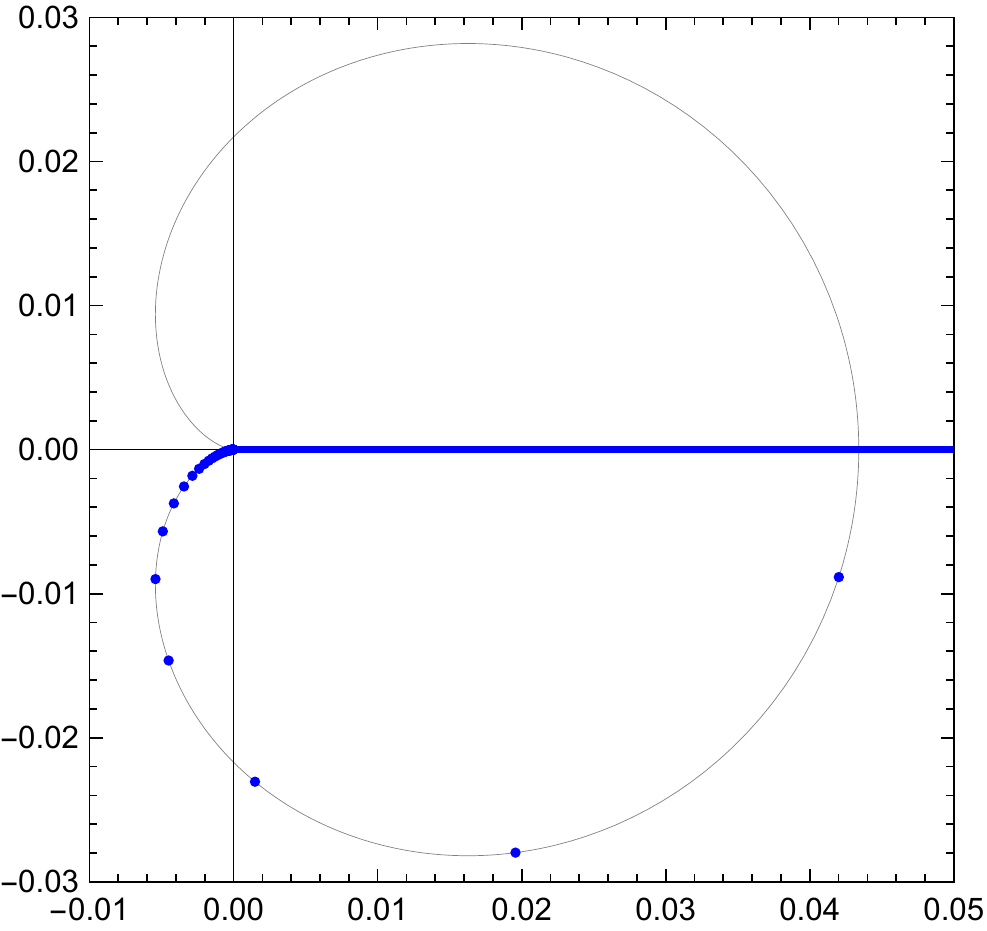}
\caption*{Fig.~3. $\sigma(H_{1,-0.25-2.4\i})$}
\end{subfigure}
\begin{subfigure}{0.5\textwidth}
\includegraphics[scale=0.7]{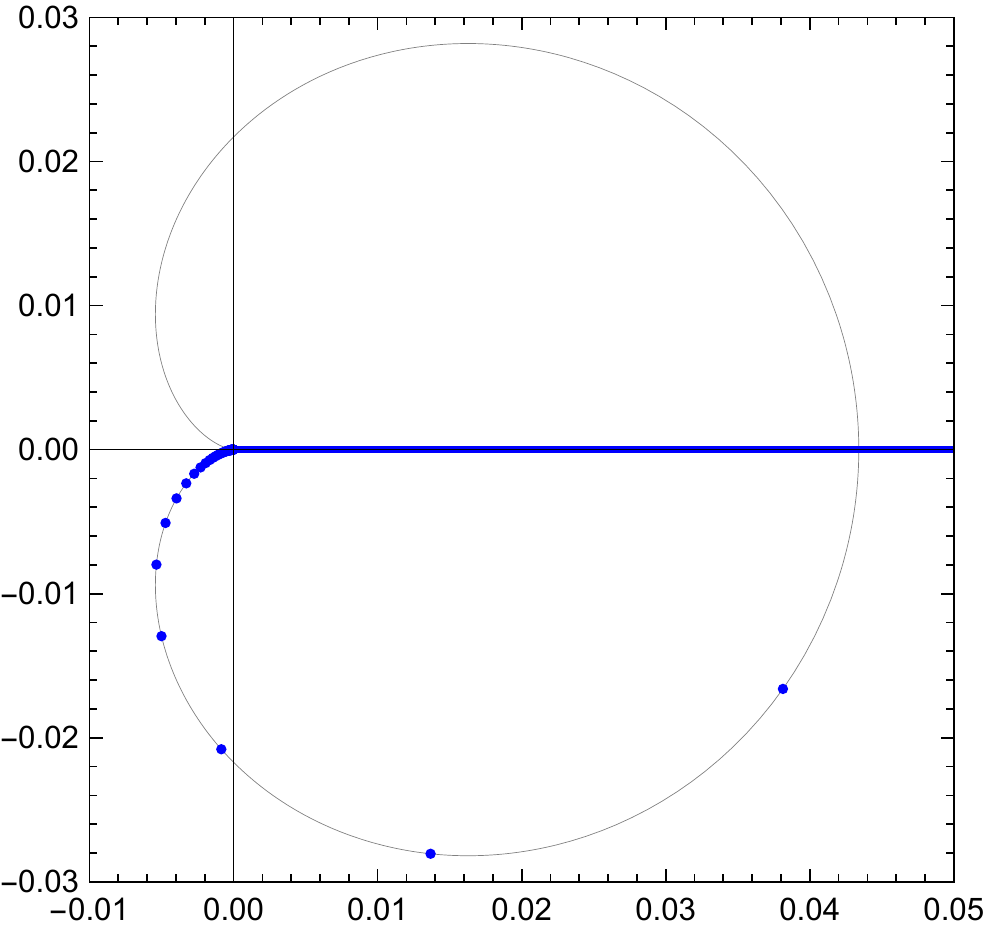}
\caption*{Fig.~4. $\sigma(H_{1,0-2.4\i})$}
\end{subfigure}
\end{figure}

\begin{figure}[H]
\begin{subfigure}{0.5\textwidth}
\includegraphics[scale=0.7]{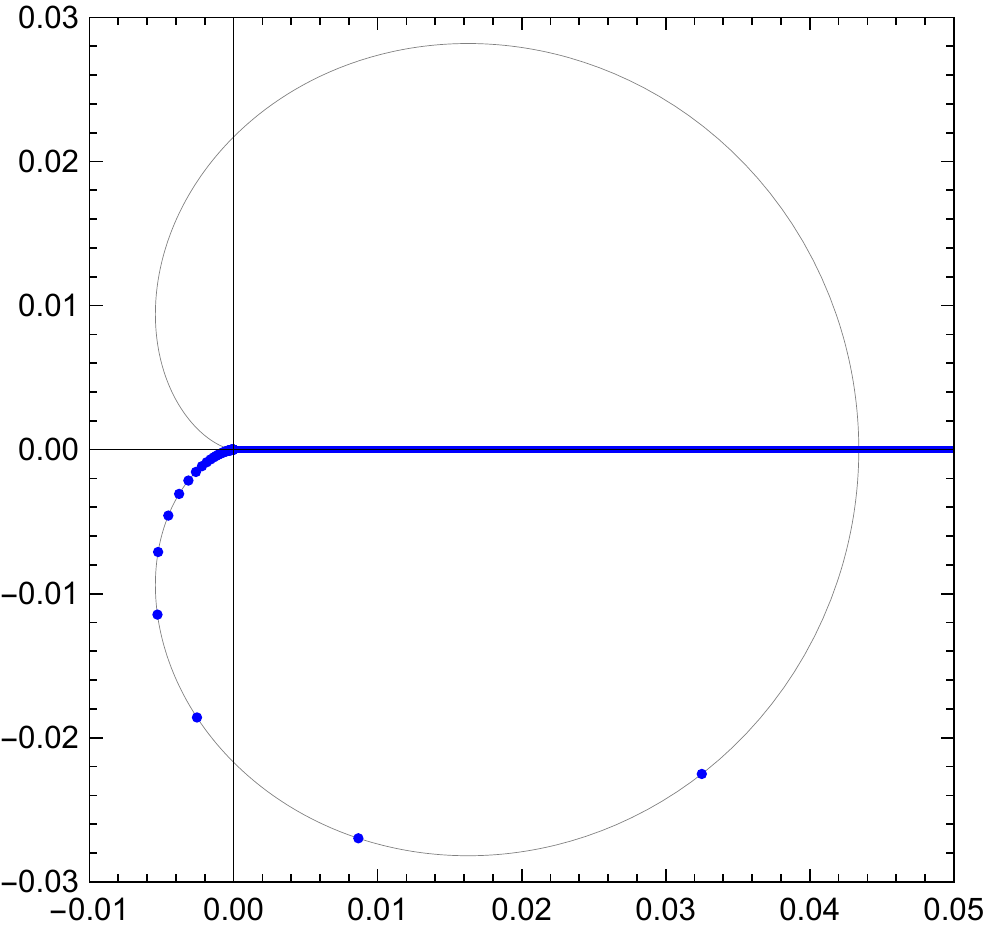}
\caption*{Fig.~5. $\sigma(H_{1,0.25-2.4\i})$}
\end{subfigure}
\begin{subfigure}{0.5\textwidth}
\includegraphics[scale=0.7]{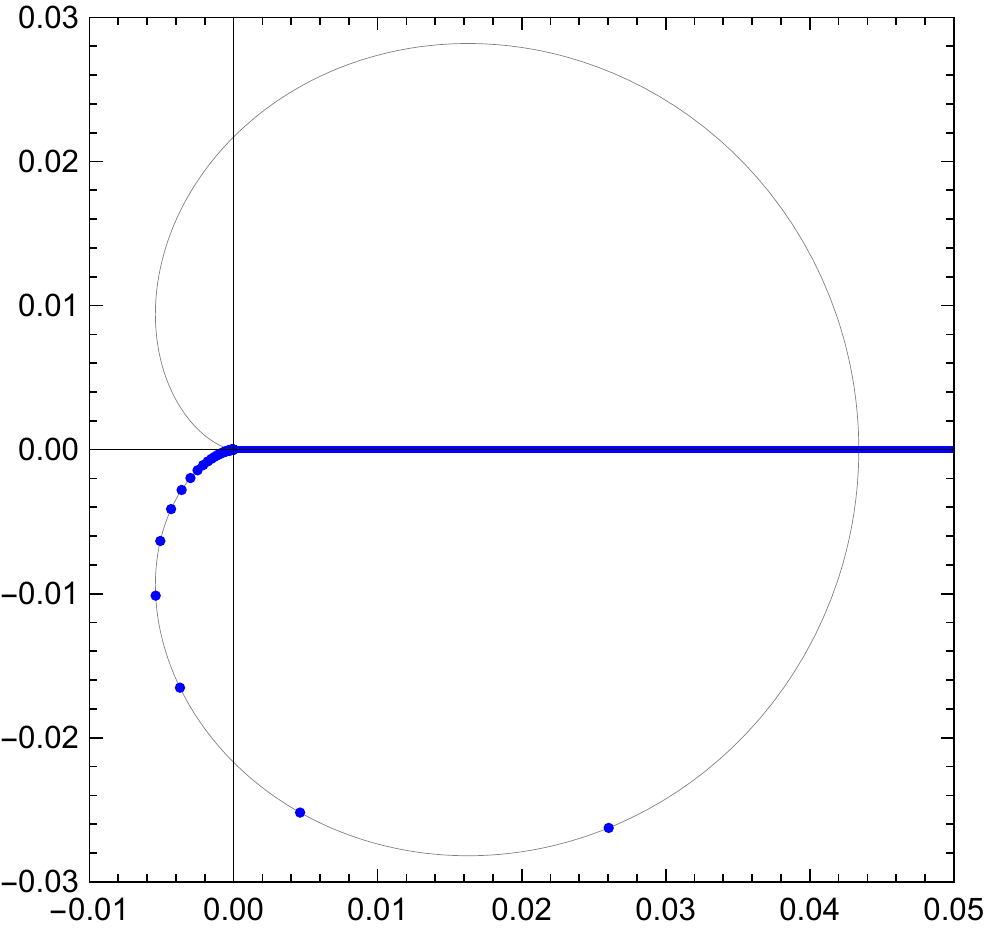}
\caption*{Fig.~6. $\sigma(H_{1,0.5-2.4\i})$}
\end{subfigure}
\end{figure}

\begin{figure}[H]
\begin{subfigure}{0.5\textwidth}
\includegraphics[scale=0.7]{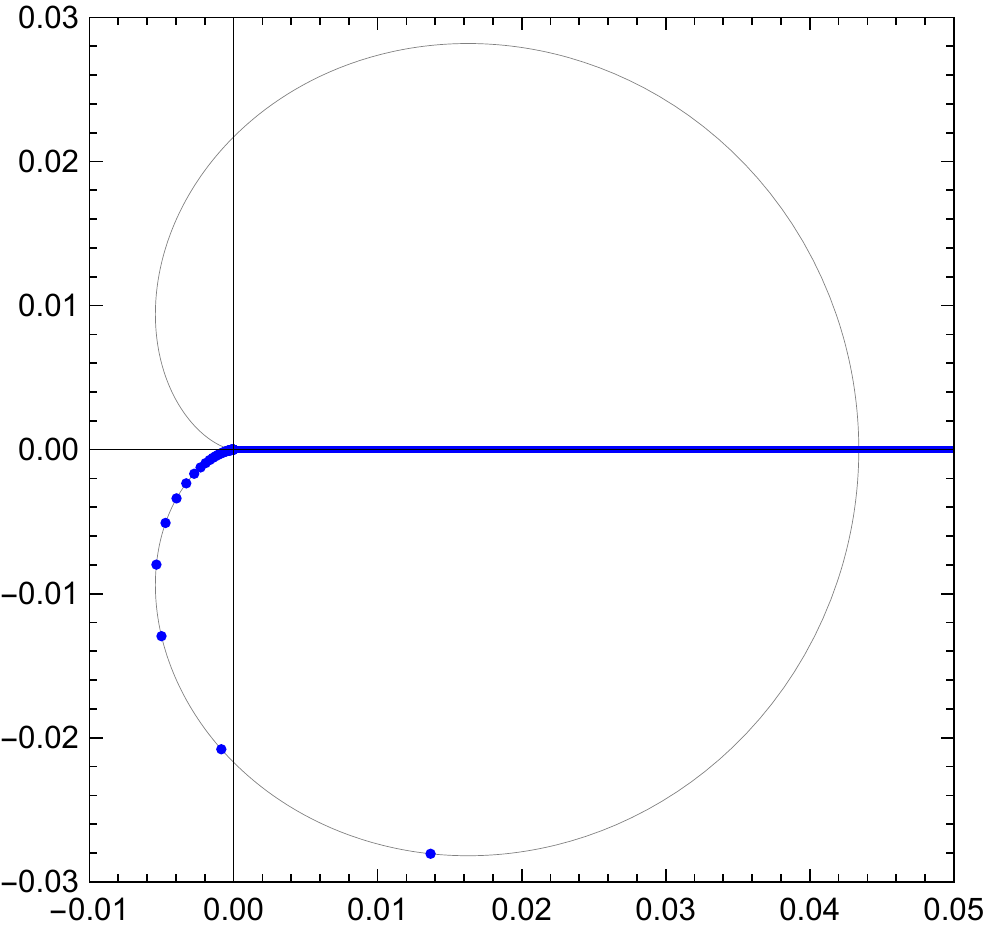}
\caption*{Fig.~7. $\sigma(H_{1,1-2.4\i})$}
\end{subfigure}
\begin{subfigure}{0.5\textwidth}
\includegraphics[scale=0.7]{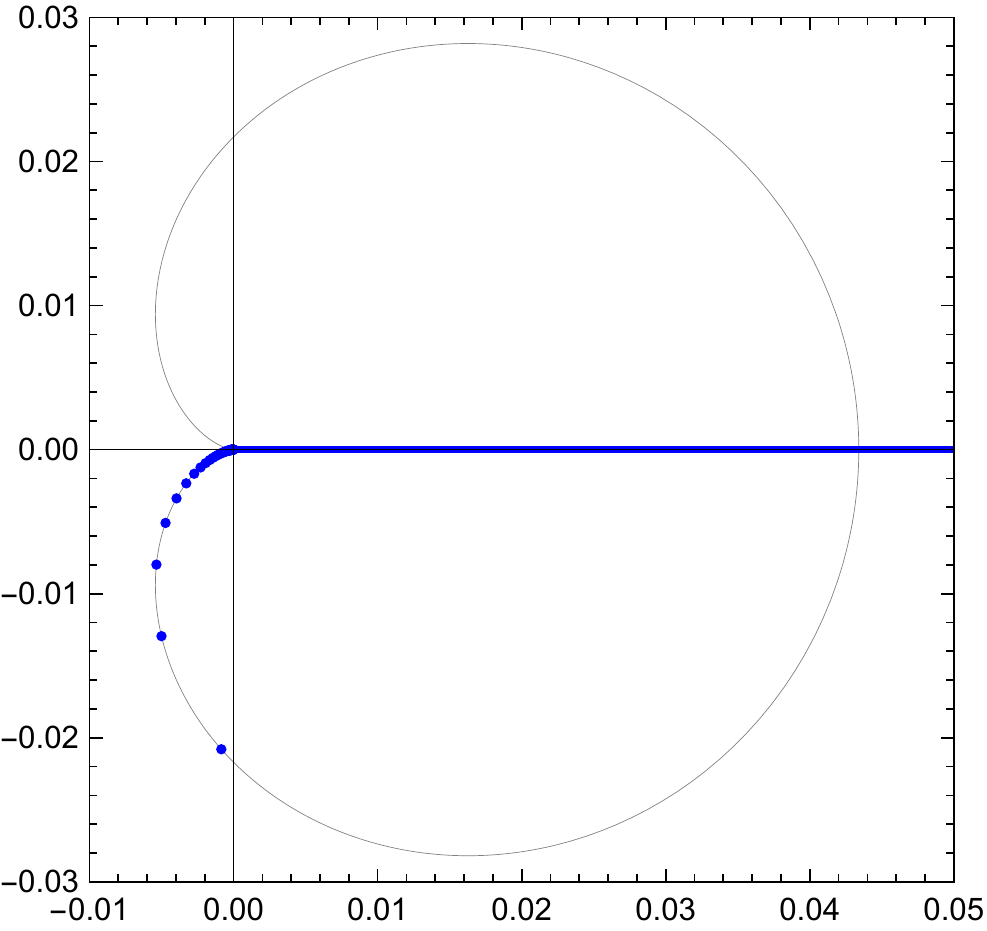}
\caption*{Fig.~8. $\sigma(H_{1,2-2.4\i})$}
\end{subfigure}\end{figure}

Next we show the spectrum for fixed $m$ and $|\beta|$.
More precisely, we present
$$
\e^{-\i2\phi}\sigma\Big(H_{\beta,-0.75+3.2\i}\Big),\quad \beta=\e^{\i\phi} \quad \hbox{for} \quad \phi=\frac{n}{8}\pi \quad \hbox{with}\quad n=0,\dots,15.
$$
This is suggested by the dilation analyticity theory, see \eqref{dila}.
With this choice the point spectrum does not move.
The continuous spectrum, on the other hand, rotates as $\e^{-\i2\phi}$, like a giant hand of a clock.
Eigenvalues hit by the continuous spectrum disappear and become resonances.
Then they reappear when the hand of the clock comes again.
The spectrum is again marked in blue and resonances in red.
We have selected $m$ such that $-1<\Re(m)<\frac12$ on purpose.
The spectrum is then more interesting because there is a lonely resonance on the upper half-plane which appears as an eigenvalue
for some phases of $\beta$.

\begin{figure}[H]
\begin{subfigure}{0.5\textwidth}
\includegraphics[scale=0.7]{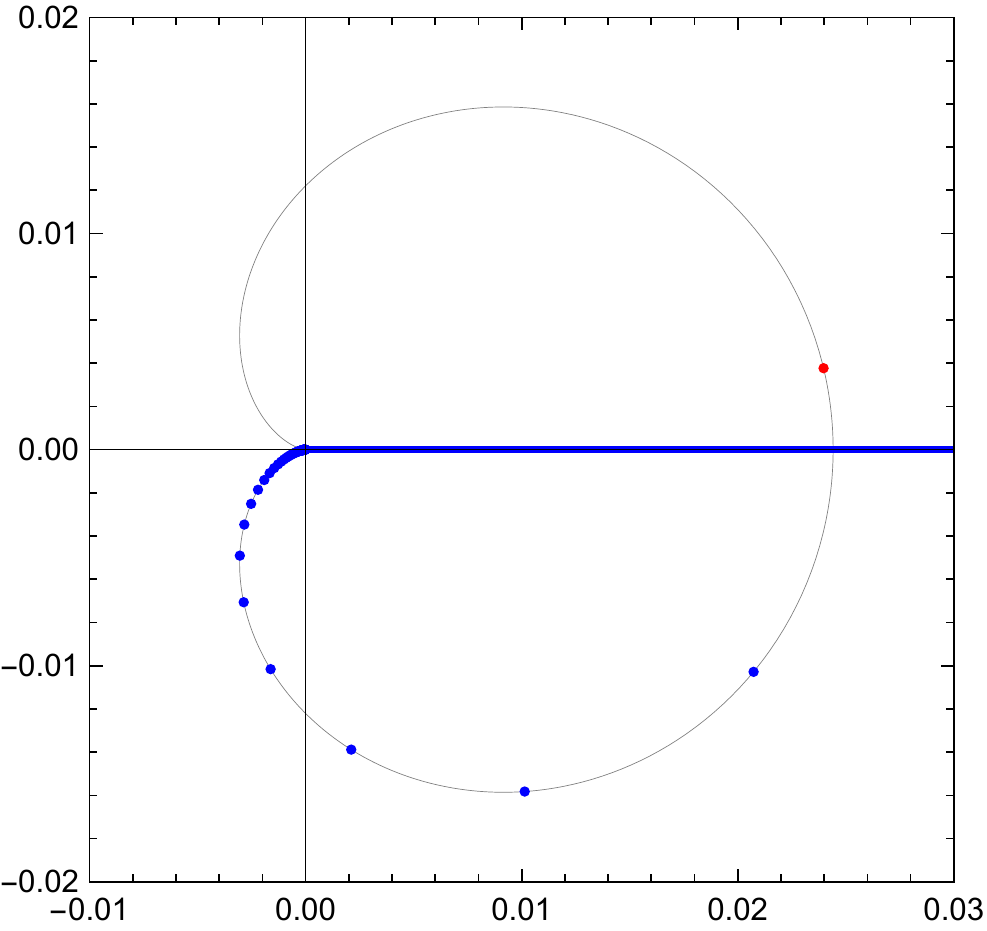}
\caption*{Fig.~9. $\phi=0$}
\end{subfigure}
\begin{subfigure}{0.5\textwidth}
\includegraphics[scale=0.7]{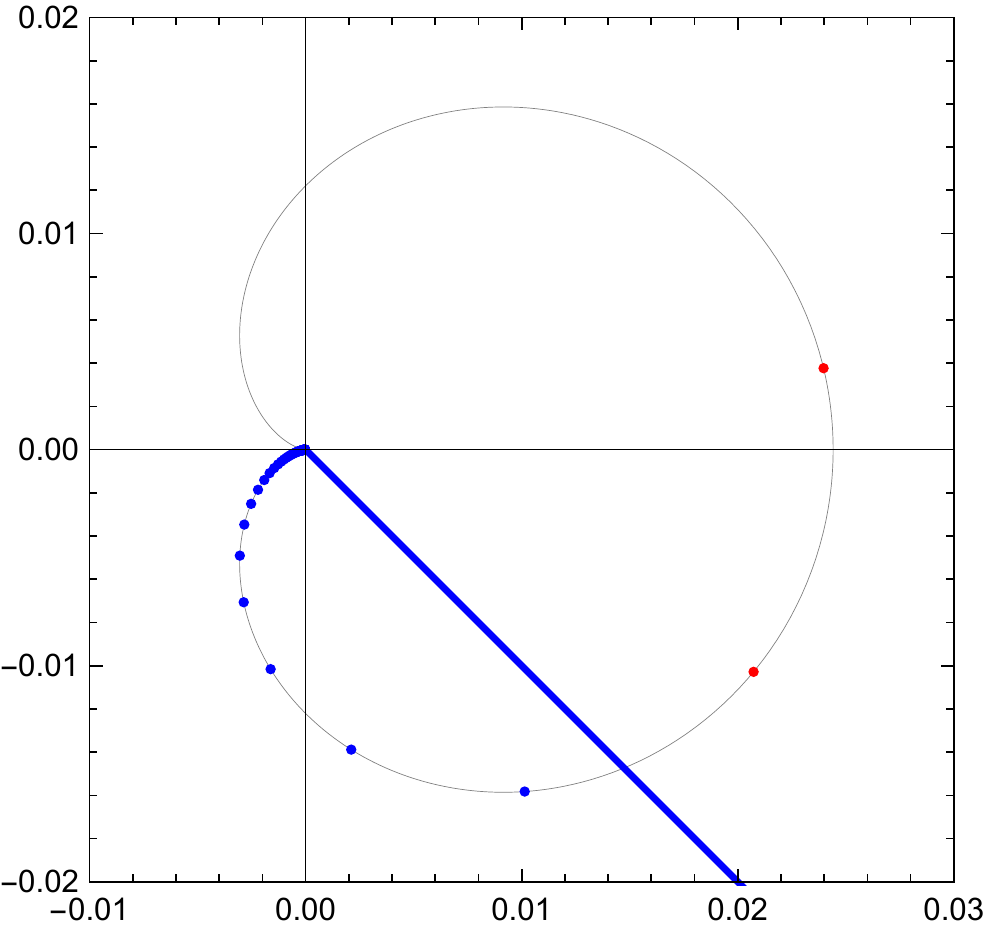}
\caption*{Fig.~10. $\phi=\frac18\pi$}
\end{subfigure}\end{figure}

\begin{figure}[H]
\begin{subfigure}{0.5\textwidth}
\includegraphics[scale=0.7]{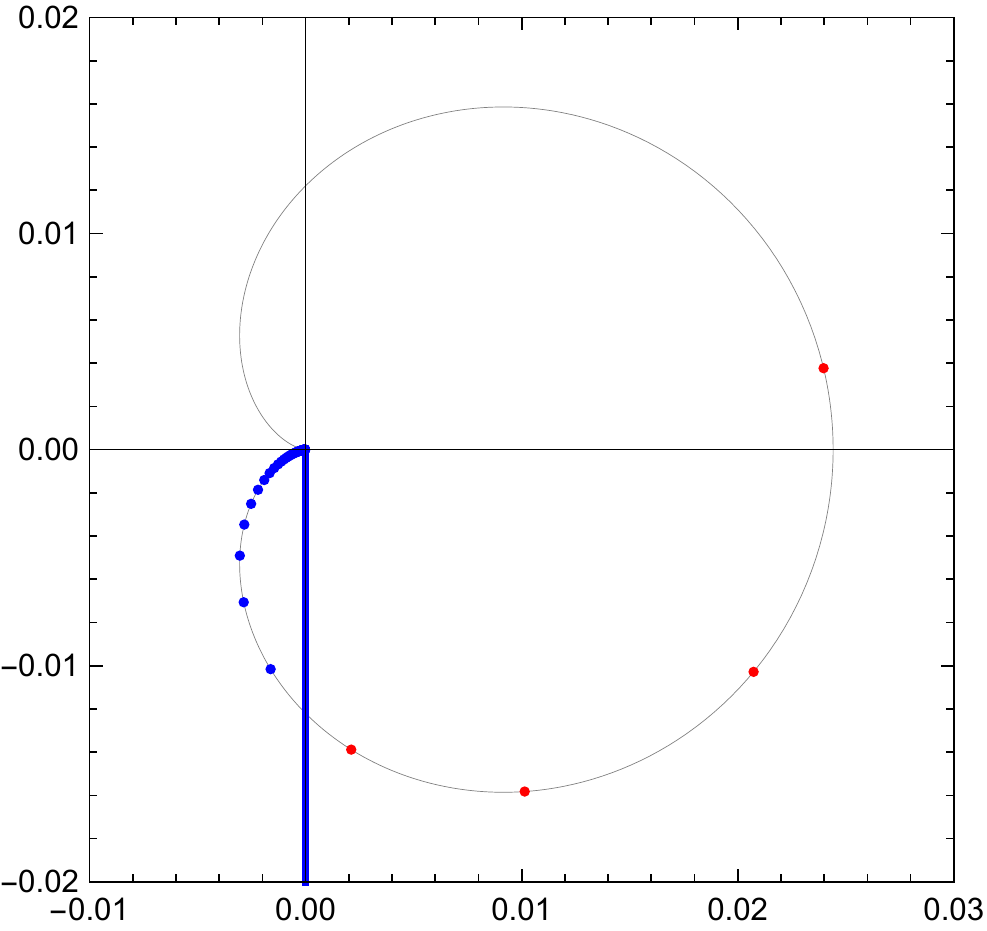}
\caption*{Fig.~11. $\phi=\frac14\pi$}
\end{subfigure}
\begin{subfigure}{0.5\textwidth}
\includegraphics[scale=0.7]{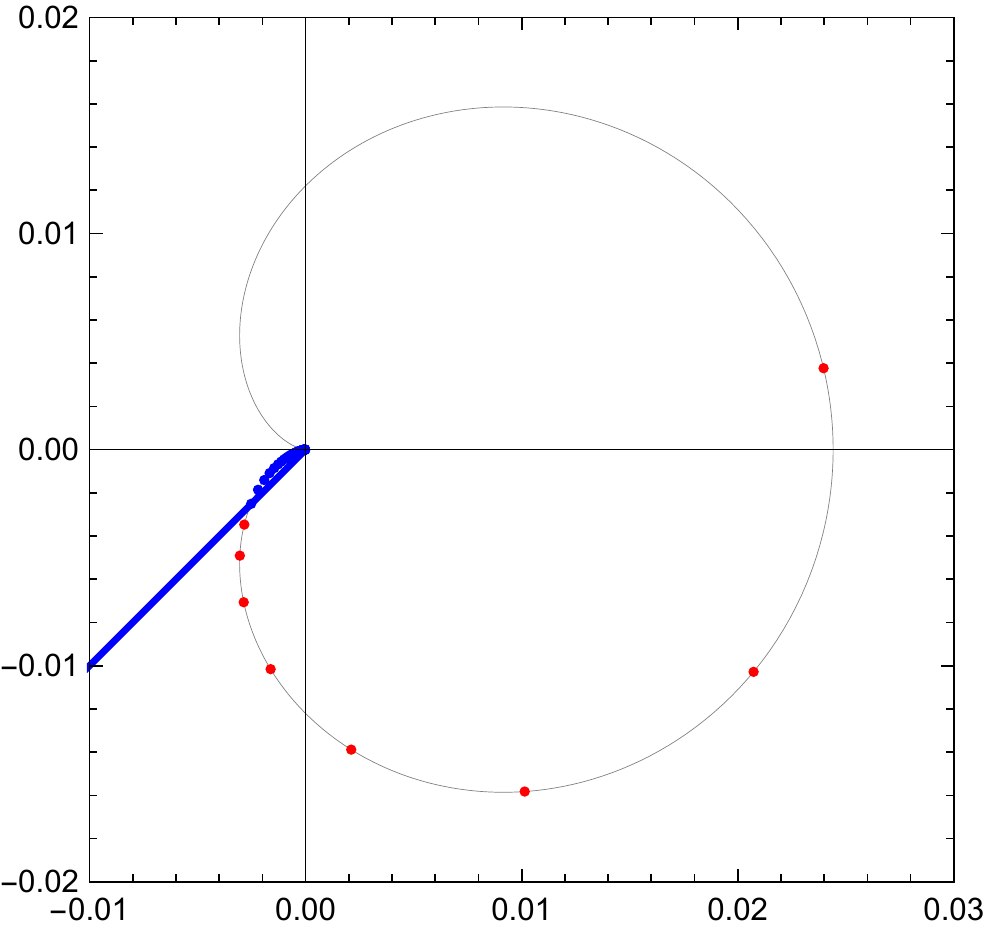}
\caption*{Fig.~12. $\phi=\frac38\pi$}
\end{subfigure}\end{figure}

\begin{figure}[H]
\begin{subfigure}{0.5\textwidth}
\includegraphics[scale=0.7]{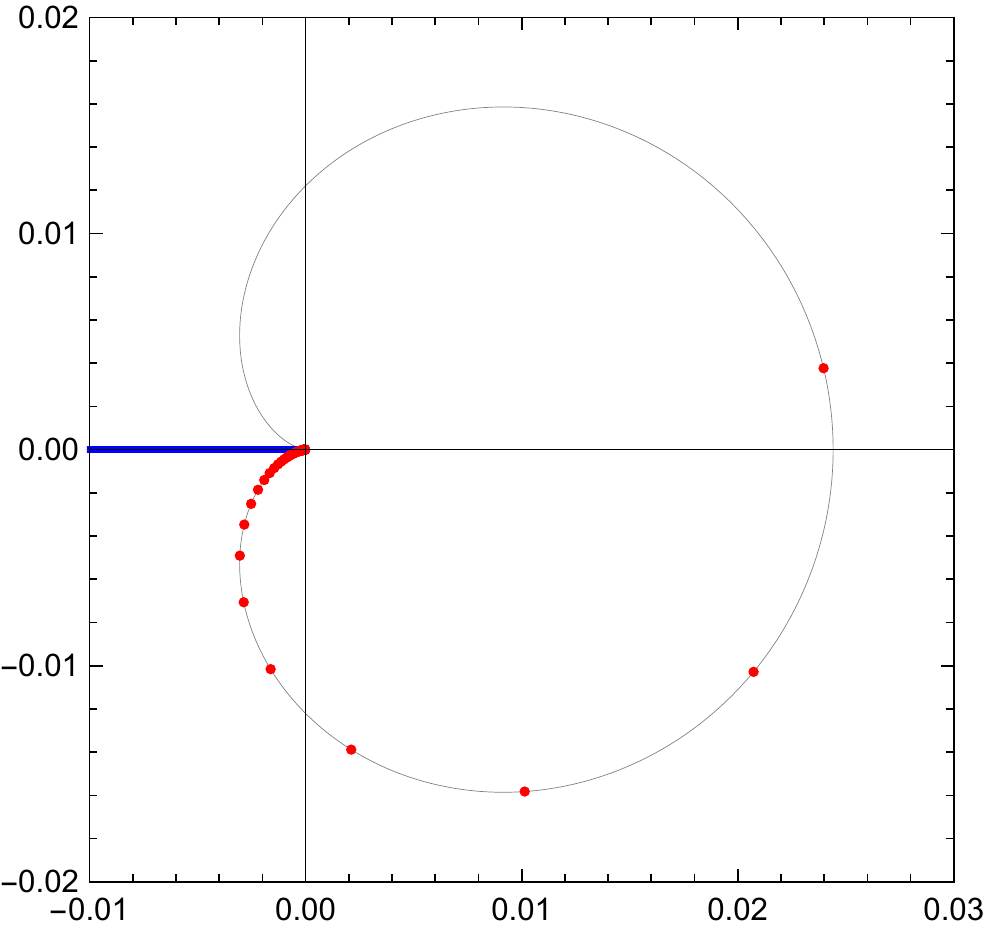}
\caption*{Fig.~13. $\phi=\frac12\pi$}
\end{subfigure}
\begin{subfigure}{0.5\textwidth}
\includegraphics[scale=0.7]{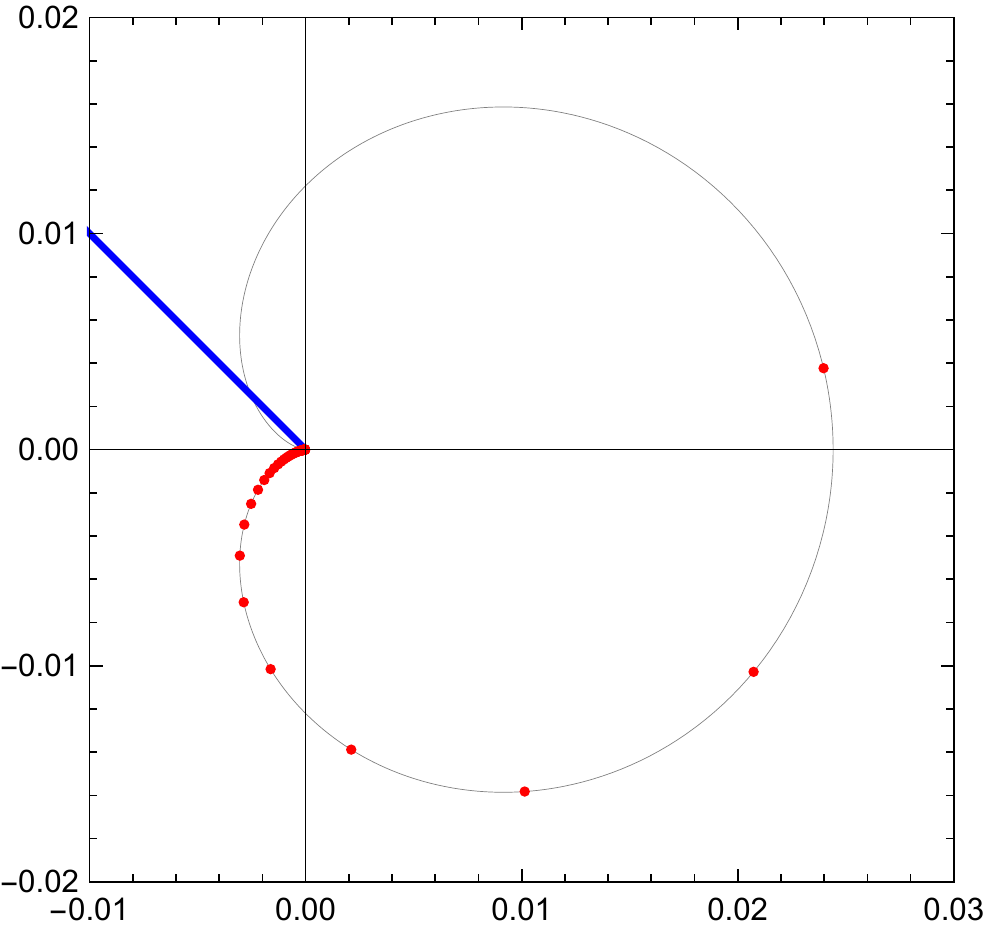}
\caption*{Fig.~14. $\phi=\frac58\pi$}
\end{subfigure}\end{figure}

\begin{figure}[H]
\begin{subfigure}{0.5\textwidth}
\includegraphics[scale=0.7]{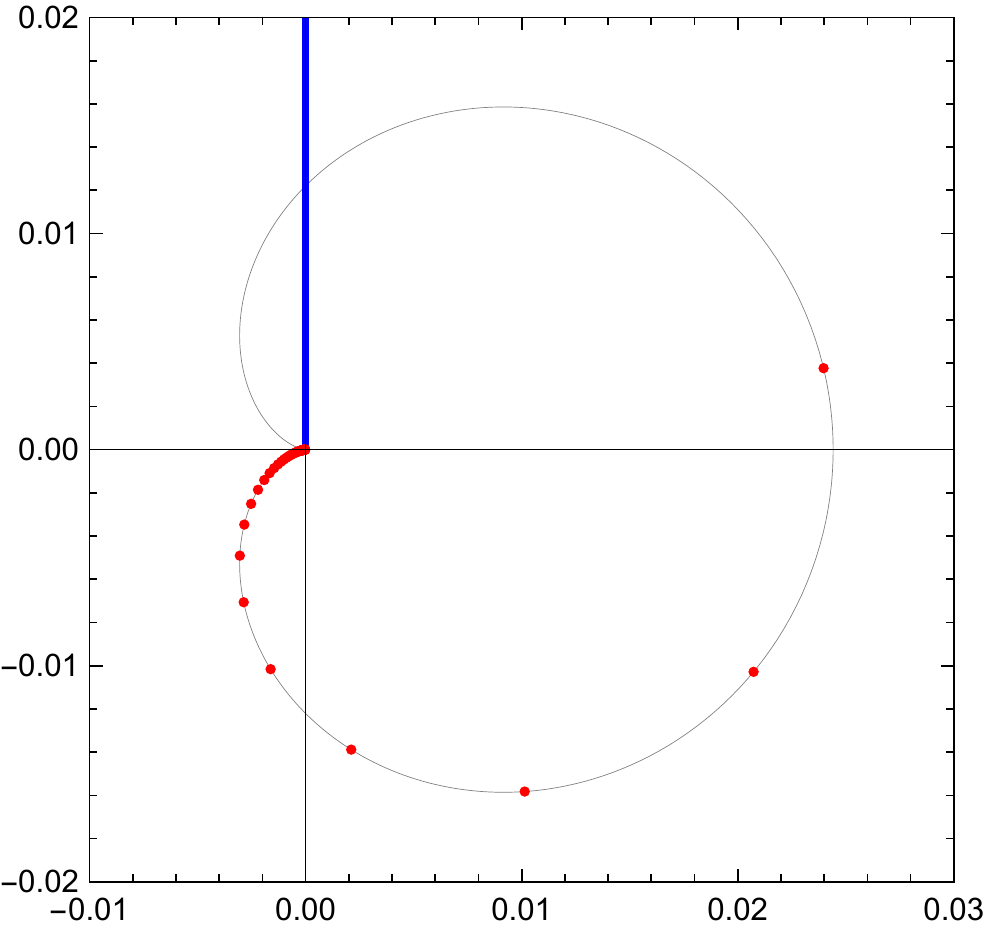}
\caption*{Fig.~15. $\phi=\frac34\pi$}
\end{subfigure}
\begin{subfigure}{0.5\textwidth}
\includegraphics[scale=0.7]{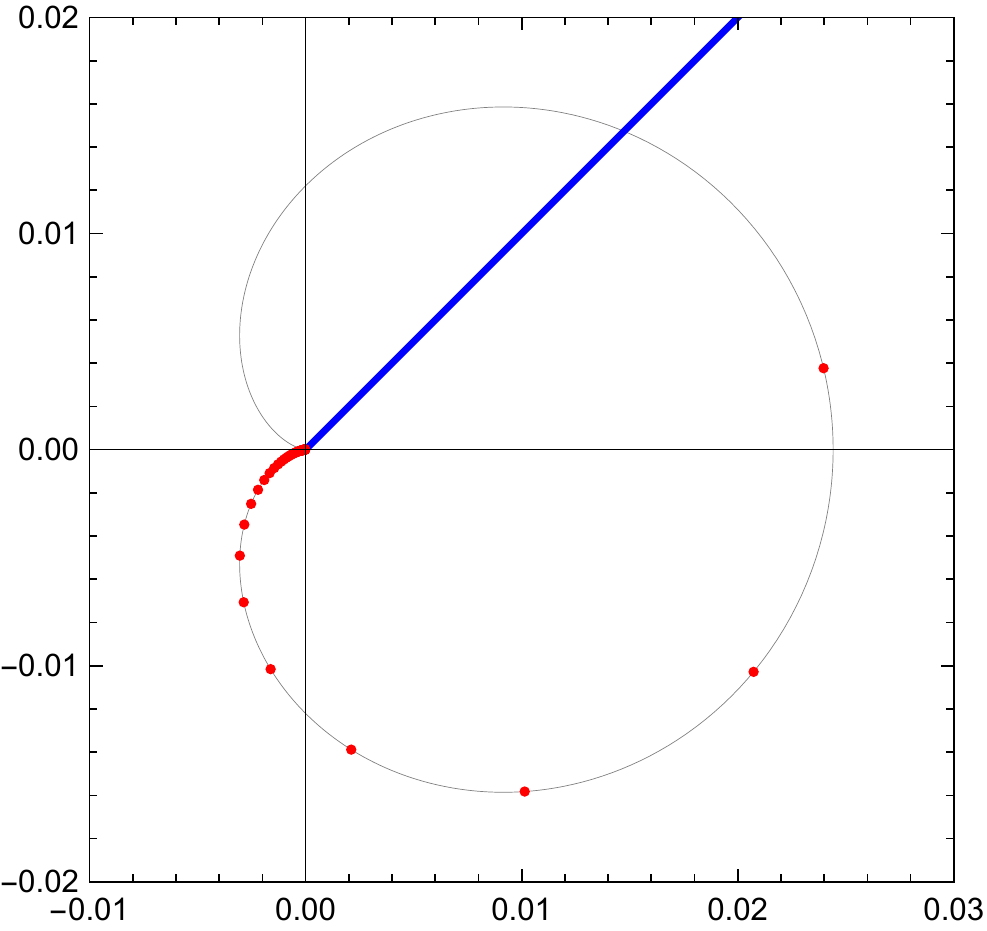}
\caption*{Fig.~16. $\phi=\frac78\pi$}
\end{subfigure}\end{figure}

\begin{figure}[H]
\begin{subfigure}{0.5\textwidth}
\includegraphics[scale=0.7]{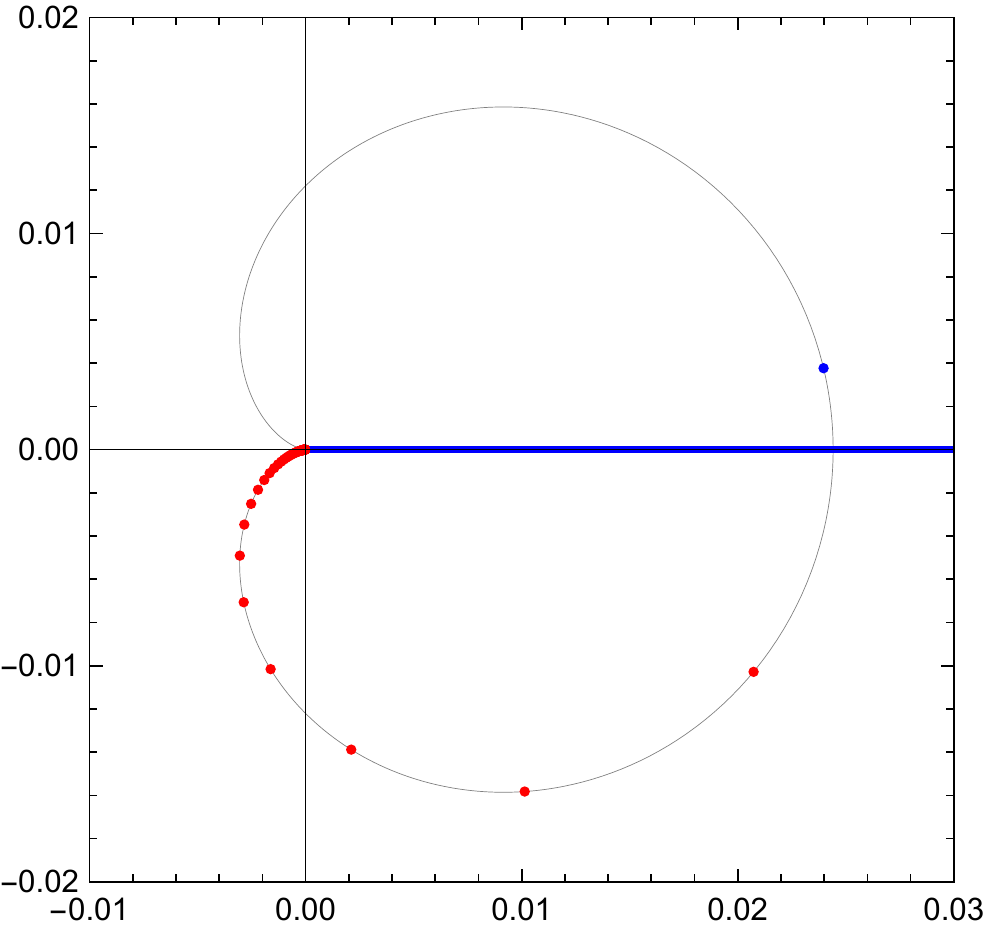}
\caption*{Fig.~17. $\phi=\pi$}
\end{subfigure}
\begin{subfigure}{0.5\textwidth}
\includegraphics[scale=0.7]{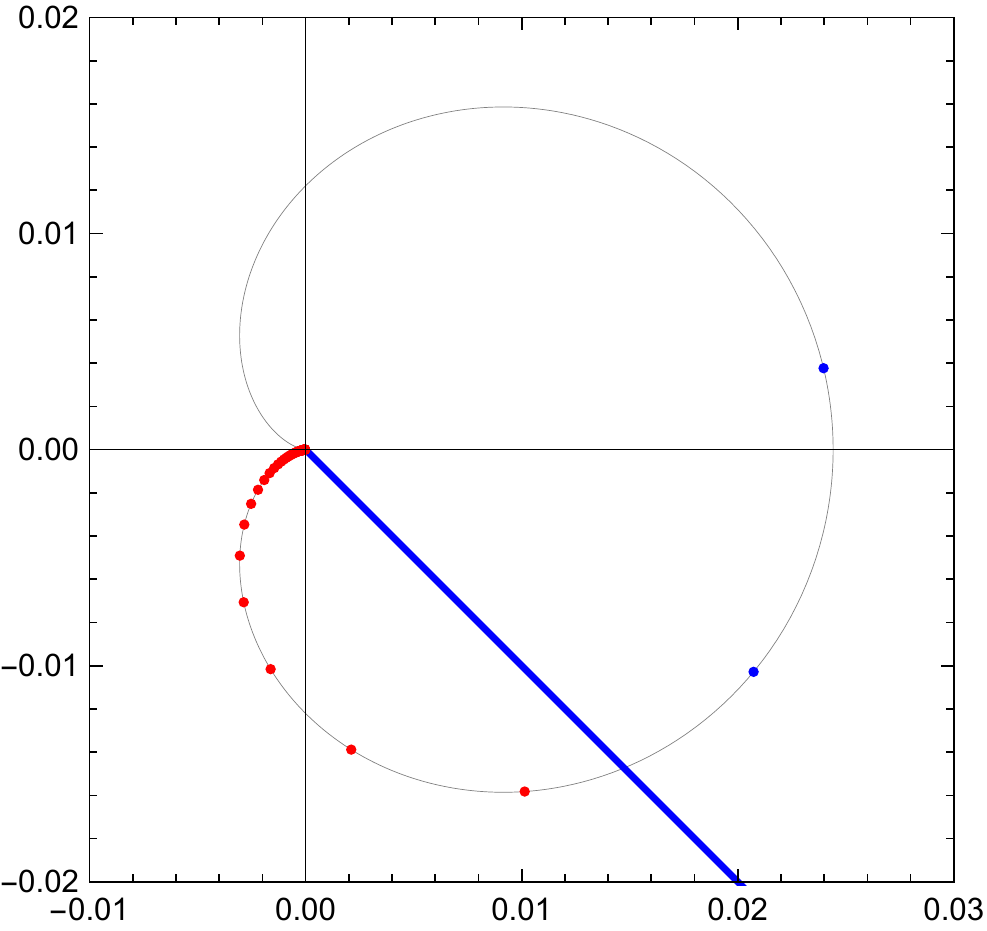}
\caption*{Fig.~18. $\phi=\frac98\pi$}
\end{subfigure}\end{figure}

\begin{figure}[H]
\begin{subfigure}{0.5\textwidth}
\includegraphics[scale=0.7]{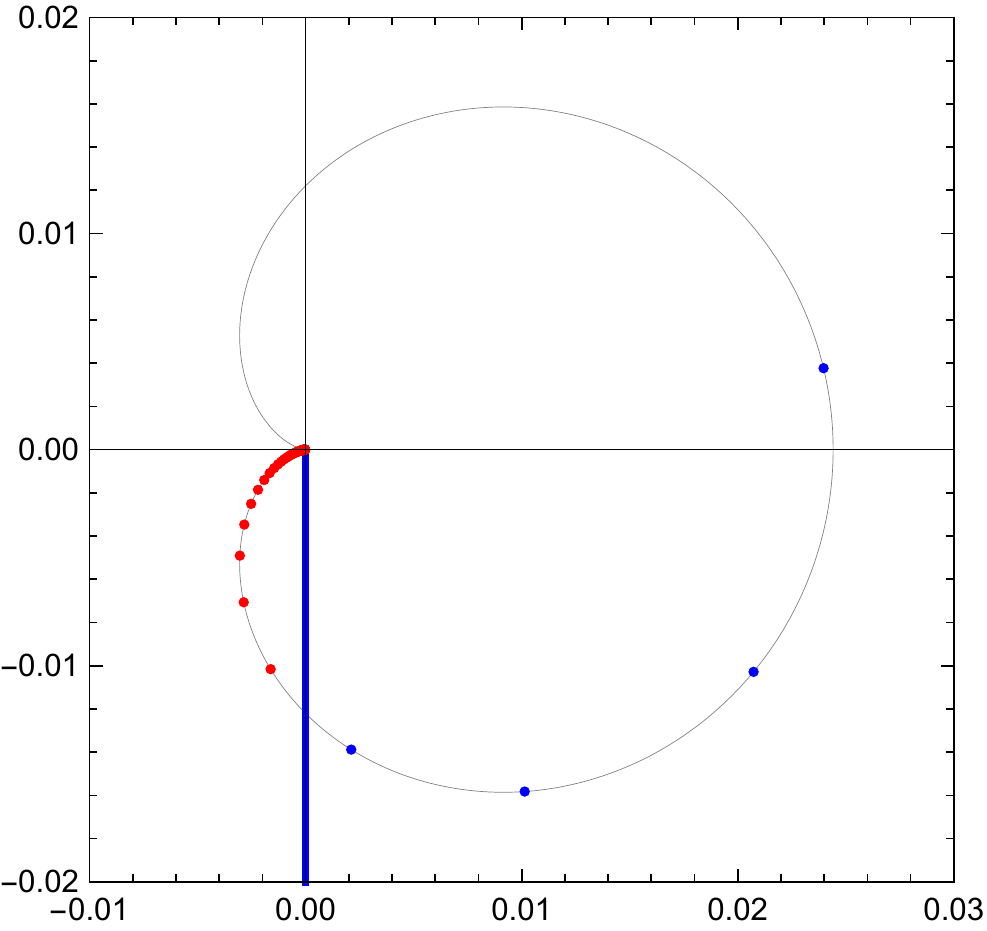}
\caption*{Fig.~19. $\phi=\frac54\pi$}
\end{subfigure}
\begin{subfigure}{0.5\textwidth}
\includegraphics[scale=0.7]{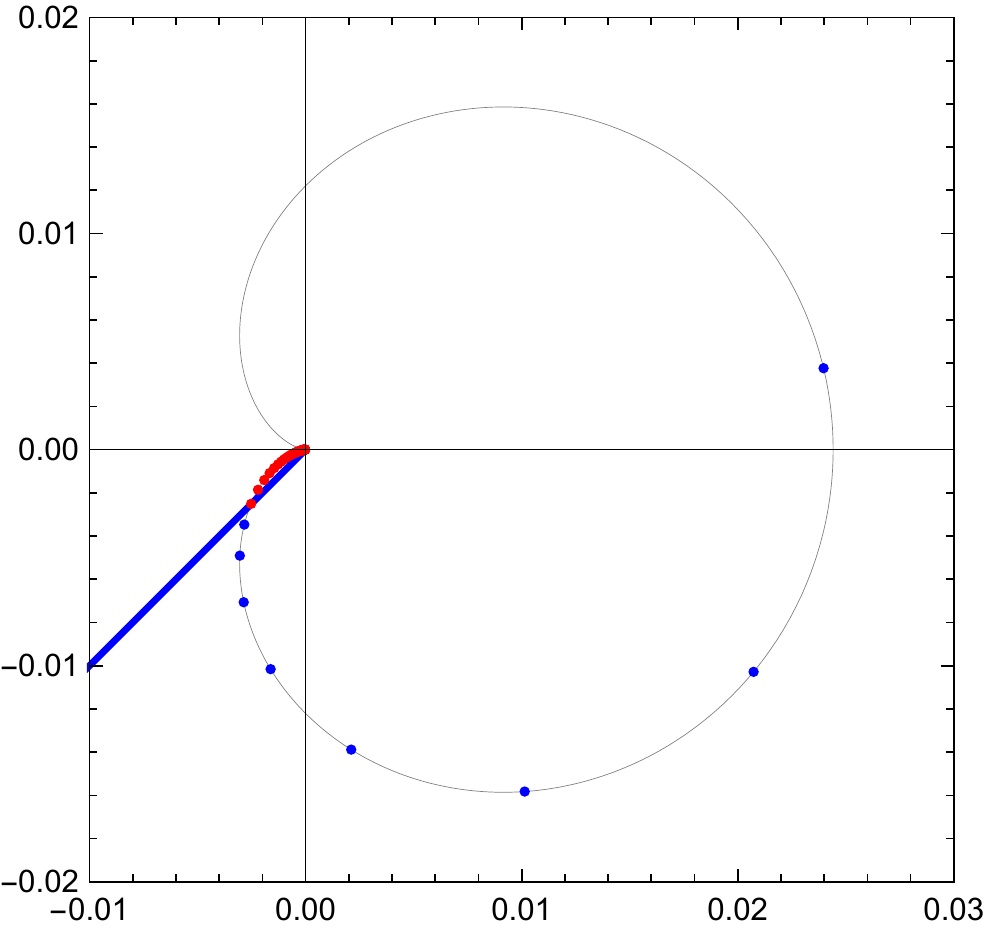}
\caption*{Fig.~20. $\phi=\frac{11}{8}\pi$}
\end{subfigure}\end{figure}

\begin{figure}[H]
\begin{subfigure}{0.5\textwidth}
\includegraphics[scale=0.7]{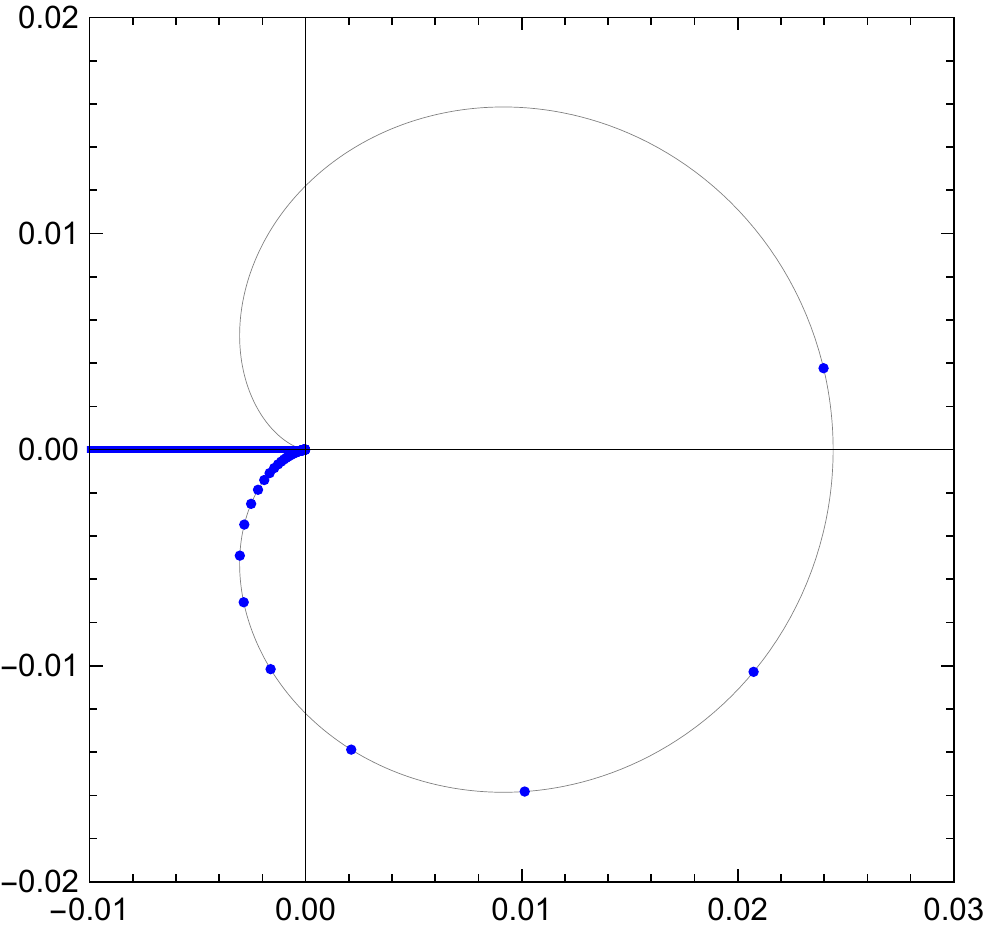}
\caption*{Fig.~21. $\phi=\frac32\pi$}
\end{subfigure}
\begin{subfigure}{0.5\textwidth}
\includegraphics[scale=0.7]{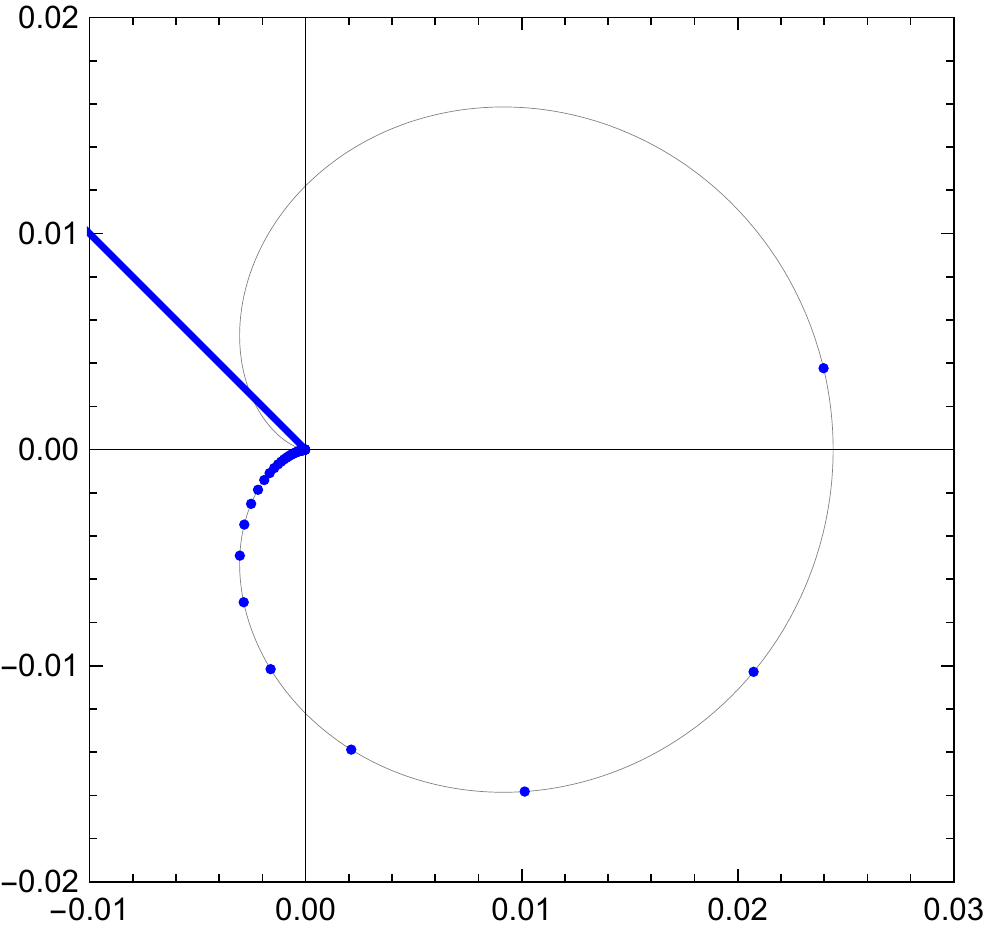}
\caption*{Fig.~22. $\phi=\frac{13}{8}\pi$}
\end{subfigure}\end{figure}

\begin{figure}[H]
\begin{subfigure}{0.5\textwidth}
\includegraphics[scale=0.7]{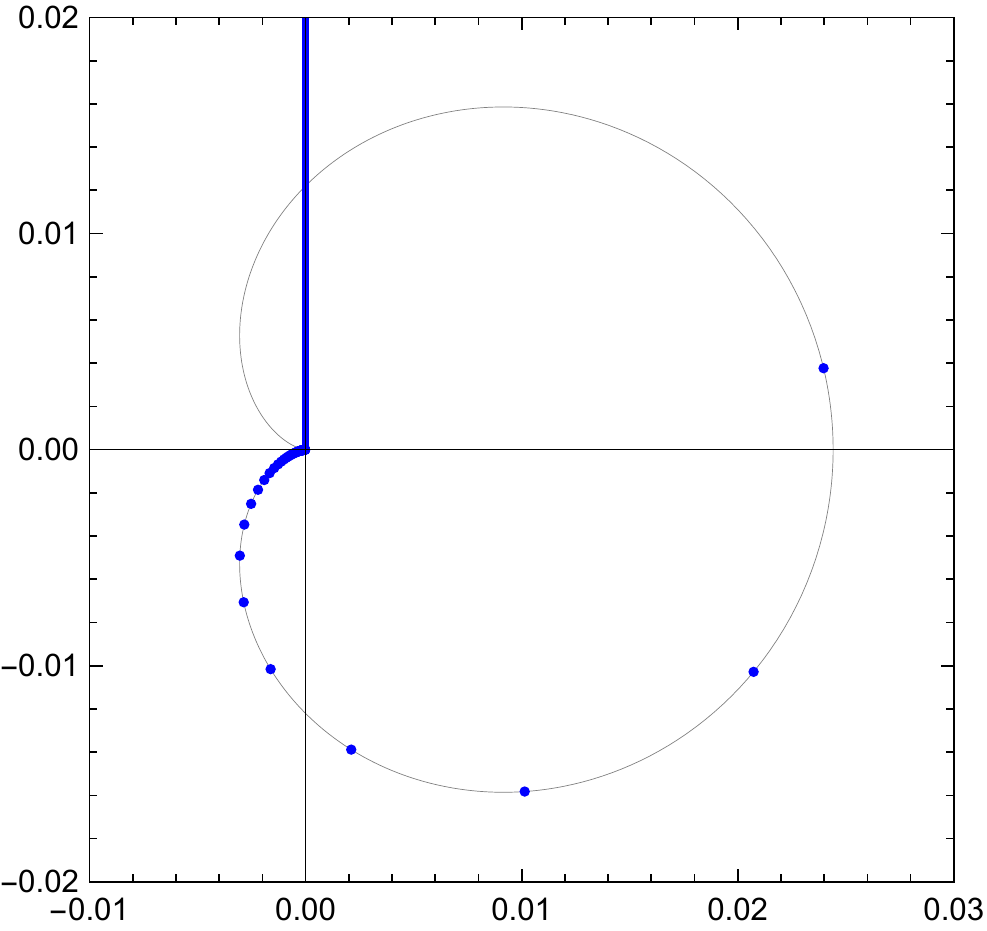}
\caption*{Fig.~23. $\phi=\frac74\pi$}
\end{subfigure}
\begin{subfigure}{0.5\textwidth}
\includegraphics[scale=0.7]{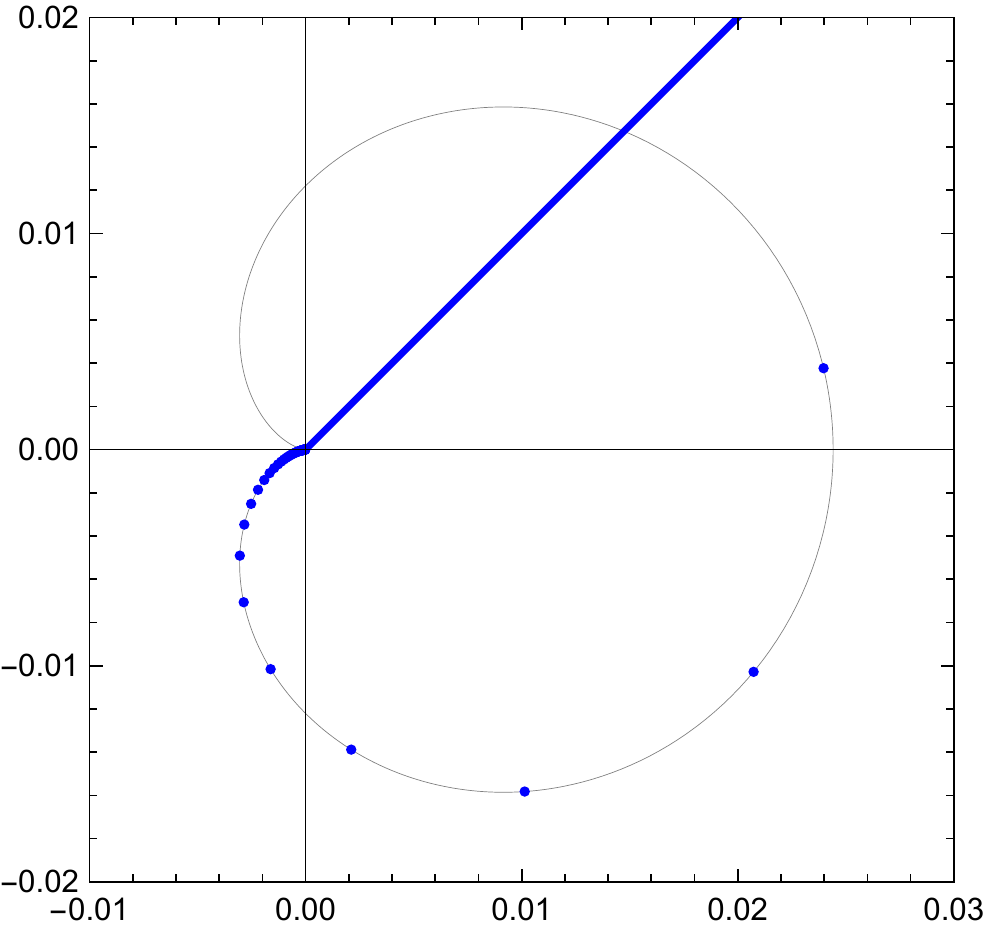}
\caption*{Fig.~24. $\phi=\frac{15}{8}\pi$}
\end{subfigure}\end{figure}

\end{document}